\pdfoutput=1
\documentclass[sigconf]{acmart}
\newif\ifshortVersion
\shortVersionfalse %

\usepackage{xr}

\ifshortVersion
\fi

\newcommand{\appendixRef}[1]{%
  \hspace{1sp}%
  \ifshortVersion%
    \cite[#1]{full-version}%
  \else%
    {#1}%
  \fi%
}

\copyrightyear{2024}
\acmYear{2024}
\setcopyright{acmlicensed}\acmConference[CCS '24]{Proceedings of the 2024 ACM SIGSAC Conference on Computer and Communications Security}{October 14--18, 2024}{Salt Lake City, UT, USA}
\acmBooktitle{Proceedings of the 2024 ACM SIGSAC Conference on Computer and Communications Security (CCS '24), October 14--18, 2024, Salt Lake City, UT, USA}
\acmDOI{10.1145/3658644.3670347}
\acmISBN{979-8-4007-0636-3/24/10}


\usepackage{import}
\usepackage[utf8]{inputenc}

\PassOptionsToPackage{hyphens}{url}%
\usepackage{hyperref}
\usepackage{graphicx}

\usepackage{amsmath}
\usepackage{amsthm}
\usepackage{amsfonts}
\usepackage{mathtools}

\usepackage{shortsym}

\usepackage{IEEEtrantools}
\allowdisplaybreaks

\usepackage[shortlabels,inline]{enumitem}
\setlist[itemize]{leftmargin=1.25em}
\setlist[enumerate]{leftmargin=1.75em}

\usepackage{xstring}   %
\usepackage{import}
\theoremstyle{plain}
\newtheorem{result}{Result}
\newtheorem*{theorem*}{Theorem}
\newtheorem*{corollary*}{Corollary}
\newtheorem*{lemma*}{Lemma}
\newtheorem*{proposition*}{Proposition}
\newtheorem*{conjecture*}{Conjecture}
\newtheorem*{result*}{Result}

\theoremstyle{definition}
\makeatletter
\renewcommand*\env@matrix[1][*\c@MaxMatrixCols c]{%
    \hskip -\arraycolsep
    \let\@ifnextchar\new@ifnextchar
    \array{#1}}
\makeatother

\DeclarePairedDelimiter\abs{\lvert}{\rvert}
\DeclarePairedDelimiter\len{\lvert}{\rvert}
\DeclarePairedDelimiter\norm{\lVert}{\rVert}

\makeatletter
\let\oldabs\abs
\def\abs{\@ifstar{\oldabs}{\oldabs*}}
\let\oldlen\len
\def\len{\@ifstar{\oldlen}{\oldlen*}}
\let\oldnorm\norm
\def\norm{\@ifstar{\oldnorm}{\oldnorm*}}
\makeatother

\newcommand{\cond}[0]{\ensuremath{\,\middle|\,}}

\DeclareMathOperator*{\argmax}{arg\,max}
\DeclareMathOperator*{\argmin}{arg\,min}

\usepackage{pifont}

\newcommand{\eqA}[0]{\ensuremath{\overset{\text{{(a)}}}{=}}}

\newcommand{\leqA}[0]{\ensuremath{\overset{\text{{(a)}}}{\leq}}}
\newcommand{\leqB}[0]{\ensuremath{\overset{\text{{(b)}}}{\leq}}}
\newcommand{\leqC}[0]{\ensuremath{\overset{\text{{(c)}}}{\leq}}}

\usepackage{dsfont}
\newcommand{\Ione}{{\mathds{1}}}

\DeclareRobustCommand{\svdots}{%
    \vcenter{%
        \offinterlineskip%
        \hbox{.}%
        \vskip0.25\normalbaselineskip%
        \hbox{.}%
        \vskip0.25\normalbaselineskip%
        \hbox{.}%
    }%
}

\newcommand*\circledScriptBlackSlim[1]{\tikz[baseline=(char.base)]{
                \node[shape=circle,draw=none,fill=black,text=white,inner sep=1pt] (char) {\scriptsize #1};}}

\usepackage{xspace}

\newcommand{\cf}[0]{cf.\xspace}
\newcommand{\ie}[0]{\emph{i.e.}\xspace}
\newcommand{\eg}[0]{\emph{e.g.}\xspace}

\usepackage{xcolor}

\definecolor{jnSUCardinalRed}{HTML}{8c1515}
\definecolor{jnSUCardinalRedLight}{HTML}{B83A4B}
\definecolor{jnSUCardinalRedDark}{HTML}{820000}
\definecolor{jnSUWhite}{HTML}{ffffff}
\definecolor{jnSUCoolGrey}{HTML}{53565A}
\definecolor{jnSUBlack}{HTML}{2e2d29}
\definecolor{jnSUBlack100}{HTML}{2e2d29}
\definecolor{jnSUBlack90}{HTML}{43423E}
\definecolor{jnSUBlack80}{HTML}{585754}
\definecolor{jnSUBlack70}{HTML}{6D6C69}
\definecolor{jnSUBlack60}{HTML}{767674}
\definecolor{jnSUBlack50}{HTML}{979694}
\definecolor{jnSUBlack40}{HTML}{ABABA9}
\definecolor{jnSUBlack30}{HTML}{C0C0BF}
\definecolor{jnSUBlack20}{HTML}{D5D5D4}
\definecolor{jnSUBlack10}{HTML}{EAEAEA}

\definecolor{jnSUPaloAlto}{HTML}{175E54}
\definecolor{jnSUPaloAltoLight}{HTML}{2D716F}
\definecolor{jnSUPaloAltoDark}{HTML}{014240}
\definecolor{jnSUPaloVerde}{HTML}{279989}
\definecolor{jnSUPaloVerdeLight}{HTML}{59B3A9}
\definecolor{jnSUPaloVerdeDark}{HTML}{017E7C}
\definecolor{jnSUOlive}{HTML}{8F993E}
\definecolor{jnSUOliveLight}{HTML}{A6B168}
\definecolor{jnSUOliveDark}{HTML}{7A863B}
\definecolor{jnSUBay}{HTML}{6FA287}
\definecolor{jnSUBayLight}{HTML}{8AB8A7}
\definecolor{jnSUBayDark}{HTML}{417865}
\definecolor{jnSUSky}{HTML}{4298B5}
\definecolor{jnSUSkyLight}{HTML}{67AFD2}
\definecolor{jnSUSkyDark}{HTML}{016895}
\definecolor{jnSULagunita}{HTML}{007C92}
\definecolor{jnSULagunitaLight}{HTML}{009AB4}
\definecolor{jnSULagunitaDark}{HTML}{006B81}
\definecolor{jnSUPoppy}{HTML}{E98300}
\definecolor{jnSUPoppyLight}{HTML}{F9A44A}
\definecolor{jnSUPoppyDark}{HTML}{D1660F}
\definecolor{jnSUSpirited}{HTML}{E04F39}
\definecolor{jnSUSpiritedLight}{HTML}{F4795B}
\definecolor{jnSUSpiritedDark}{HTML}{C74632}
\definecolor{jnSUIlluminating}{HTML}{FEDD5C}
\definecolor{jnSUIlluminatingLight}{HTML}{FFE781}
\definecolor{jnSUIlluminatingDark}{HTML}{FEC51D}
\definecolor{jnSUPlum}{HTML}{620059}
\definecolor{jnSUPlumLight}{HTML}{734675}
\definecolor{jnSUPlumDark}{HTML}{350D36}
\definecolor{jnSUBrick}{HTML}{651C32}
\definecolor{jnSUBrickLight}{HTML}{7F2D48}
\definecolor{jnSUBrickDark}{HTML}{42081B}
\definecolor{jnSUArchway}{HTML}{5D4B3C}
\definecolor{jnSUArchwayLight}{HTML}{766253}
\definecolor{jnSUArchwayDark}{HTML}{2F2424}
\definecolor{jnSUStone}{HTML}{7F7776}
\definecolor{jnSUStoneLight}{HTML}{D4D1D1}
\definecolor{jnSUStoneDark}{HTML}{544948}
\definecolor{jnSUFog}{HTML}{DAD7CB}
\definecolor{jnSUFogLight}{HTML}{F4F4F4}
\definecolor{jnSUFogDark}{HTML}{B6B1A9}

\definecolor{jnSUDigitalRed}{HTML}{B1040E}
\definecolor{jnSUDigitalRedLight}{HTML}{E50808}
\definecolor{jnSUDigitalRedDark}{HTML}{820000}
\definecolor{jnSUDigitalBlue}{HTML}{006CB8}
\definecolor{jnSUDigitalBlueLight}{HTML}{6FC3FF}
\definecolor{jnSUDigitalBlueDark}{HTML}{00548f}
\definecolor{jnSUDigitalGreen}{HTML}{008566}
\definecolor{jnSUDigitalGreenLight}{HTML}{1AECBA}
\definecolor{jnSUDigitalGreenDark}{HTML}{006F54}

\definecolor{myParula01Blue}{RGB}{0,114,189}
\definecolor{myParula02Orange}{RGB}{217,83,25}
\definecolor{myParula03Yellow}{RGB}{237,177,32}
\definecolor{myParula04Purple}{RGB}{126,47,142}
\definecolor{myParula05Green}{RGB}{119,172,48}
\definecolor{myParula06LightBlue}{RGB}{77,190,238}
\definecolor{myParula07Red}{RGB}{162,20,47}

\usepackage{tikz}
\usetikzlibrary{calc}
\usetikzlibrary{arrows}
\usetikzlibrary{arrows.meta}
\usetikzlibrary{patterns}
\usetikzlibrary{positioning}
\usetikzlibrary{decorations.pathreplacing}
\usetikzlibrary{calligraphy}
\usetikzlibrary{shapes.misc}
\usetikzlibrary{shapes.symbols}
\usetikzlibrary{shapes.arrows}
\usetikzlibrary{spy}

\usepackage{pgfplots}
\pgfplotsset{compat=1.17}
\usepgfplotslibrary{fillbetween}

\pgfplotsset{
    discard if not/.style 2 args={
            x filter/.code={
                    \edef\tempa{\thisrow{#1}}
                    \edef\tempb{#2}
                    \ifx\tempa\tempb
                    \else
                        
                    \fi
                }
        },
}

\tikzset{myparula11/.style={color=myParula01Blue,solid,mark=+,mark options={solid}}}
\tikzset{myparula12/.style={color=myParula01Blue,densely dashed,mark=x,mark options={solid}}}
\tikzset{myparula13/.style={color=myParula01Blue,densely dotted,mark=o,mark options={solid}}}
\tikzset{myparula14/.style={color=myParula01Blue,dashdotted,mark=triangle,mark options={solid}}}
\tikzset{myparula15/.style={color=myParula01Blue,dashdotdotted,mark=square,mark options={solid}}}

\tikzset{myparula21/.style={color=myParula02Orange,solid,mark=+,mark options={solid}}}
\tikzset{myparula22/.style={color=myParula02Orange,densely dashed,mark=x,mark options={solid}}}
\tikzset{myparula23/.style={color=myParula02Orange,densely dotted,mark=o,mark options={solid}}}
\tikzset{myparula24/.style={color=myParula02Orange,dashdotted,mark=triangle,mark options={solid}}}
\tikzset{myparula25/.style={color=myParula02Orange,dashdotdotted,mark=square,mark options={solid}}}

\tikzset{myparula31/.style={color=myParula03Yellow,solid,mark=+,mark options={solid}}}
\tikzset{myparula32/.style={color=myParula03Yellow,densely dashed,mark=x,mark options={solid}}}
\tikzset{myparula33/.style={color=myParula03Yellow,densely dotted,mark=o,mark options={solid}}}
\tikzset{myparula34/.style={color=myParula03Yellow,dashdotted,mark=triangle,mark options={solid}}}
\tikzset{myparula35/.style={color=myParula03Yellow,dashdotdotted,mark=square,mark options={solid}}}

\tikzset{myparula41/.style={color=myParula04Purple,solid,mark=+,mark options={solid}}}
\tikzset{myparula42/.style={color=myParula04Purple,densely dashed,mark=x,mark options={solid}}}
\tikzset{myparula43/.style={color=myParula04Purple,densely dotted,mark=o,mark options={solid}}}
\tikzset{myparula44/.style={color=myParula04Purple,dashdotted,mark=triangle,mark options={solid}}}
\tikzset{myparula45/.style={color=myParula04Purple,dashdotdotted,mark=square,mark options={solid}}}

\tikzset{myparula51/.style={color=myParula05Green,solid,mark=+,mark options={solid}}}
\tikzset{myparula52/.style={color=myParula05Green,densely dashed,mark=x,mark options={solid}}}
\tikzset{myparula53/.style={color=myParula05Green,densely dotted,mark=o,mark options={solid}}}
\tikzset{myparula54/.style={color=myParula05Green,dashdotted,mark=triangle,mark options={solid}}}
\tikzset{myparula55/.style={color=myParula05Green,dashdotdotted,mark=square,mark options={solid}}}

\tikzset{myparula61/.style={color=myParula06LightBlue,solid,mark=+,mark options={solid}}}
\tikzset{myparula62/.style={color=myParula06LightBlue,densely dashed,mark=x,mark options={solid}}}
\tikzset{myparula63/.style={color=myParula06LightBlue,densely dotted,mark=o,mark options={solid}}}
\tikzset{myparula64/.style={color=myParula06LightBlue,dashdotted,mark=triangle,mark options={solid}}}
\tikzset{myparula65/.style={color=myParula06LightBlue,dashdotdotted,mark=square,mark options={solid}}}

\tikzset{myparula71/.style={color=myParula07Red,solid,mark=+,mark options={solid}}}
\tikzset{myparula72/.style={color=myParula07Red,densely dashed,mark=x,mark options={solid}}}
\tikzset{myparula73/.style={color=myParula07Red,densely dotted,mark=o,mark options={solid}}}
\tikzset{myparula74/.style={color=myParula07Red,dashdotted,mark=triangle,mark options={solid}}}
\tikzset{myparula75/.style={color=myParula07Red,dashdotdotted,mark=square,mark options={solid}}}

\pgfplotsset{
    mysimpleplot/.style = {
            every axis plot/.prefix style={thick},
            width=1.05\linewidth,
            height=0.75\linewidth,
            title style={font=\scriptsize,align=center},
            legend cell align=left,
            legend style={font=\scriptsize},
            legend columns=3,
            legend style={
                    at={(0.5,1)},
                    yshift=0.3em,
                    anchor=south,
                    draw=none,
                    /tikz/every even column/.append style={
                            column sep=0.3em
                        },
                    cells={
                            align=left
                        }
                },
            grid=both,
            minor tick num=4,
            major grid style={solid,very thin,draw=gray!50},
            minor grid style={solid,ultra thin,draw=gray!20},
            label style={font=\scriptsize,align=center},
            tick label style={font=\scriptsize},
        },
}

\pgfplotsset{
    mysimpleresilienceplot01/.style = {
            mysimpleplot,
            ylabel={Adversary resilience $\beta$},
            height=0.45\linewidth,
            width=\linewidth,
            ymin=0,ymax=0.5,
            ytick={0,0.1,0.2,0.3,0.4,0.5},
            xmin=1e-5,xmax=1e2,
            grid=major,
        },
}

\tikzset{blockchainold/.style={
            x=1.5cm,
            y=0.6cm,
            node distance=0.5cm,
            block/.style = {
                    minimum width=0.25cm,
                    minimum height=0.25cm,
                    draw,
                    shade,
                    top color=white,
                    bottom color=black!10,
                },
            block-adv1/.style = {
                    block,
                    bottom color=myParula01Blue!50,
                    draw=myParula01Blue!50!black
                },
            block-adv2/.style = {
                    block,
                    bottom color=myParula07Red!50,
                    draw=myParula07Red!50!black,
                },
            block-adv3/.style = {
                    block,
                    bottom color=myParula05Green!50,
                    draw=myParula05Green!50!black,
                },
            block-green/.style = {
                    block,
                    bottom color=myParula05Green!50,
                    draw=myParula05Green!50!black,
                },
            block-red/.style = {
                    block,
                    bottom color=myParula07Red!50,
                    draw=myParula07Red!70!black,
                },
            block-gray/.style = {
                    block,
                    bottom color=black!30,
                },
            block-big/.style = {
                    minimum width=0.7cm,
                    minimum height=0.7cm,
                    draw,
                    shade,
                    top color=white,
                    bottom color=black!10,
                },
            branch/.style = {
                    minimum width=0.1cm,
                    minimum height=0.1cm,
                    draw,
                    circle,
                    inner sep=0,
                    fill=black,
                },
            link/.style = {
                    -latex,
                },
            hiddenlink/.style = {
                    dashed,
                },
            hiddenlink-adv1/.style = {
                    hiddenlink,
                    draw=myParula01Blue!50!black,
                },
            hiddenlink-adv2/.style = {
                    hiddenlink,
                    draw=myParula07Red!50!black,
                },
            hiddenlink-adv3/.style = {
                    hiddenlink,
                    draw=myParula05Green!50!black,
                },
            label/.style = {
                },
            label-adv1/.style = {
                    label,
                    text=myParula01Blue!50!black,
                },
            label-adv2/.style = {
                    label,
                    text=myParula07Red!50!black,
                },
            label-adv3/.style = {
                    label,
                    text=myParula05Green!50!black,
                },
        }
}

\tikzset{blockchain/.style={
            x=0.5cm,
            y=0.55cm,
            node distance=0.5cm,
            block/.style = {
                    minimum width=0.3cm,
                    minimum height=0.3cm,
                    draw,
                    shade,
                    top color=white,
                    bottom color=black!10,
                    inner sep=0,
                },
            block-adv/.style = {
                    block,
                    bottom color=myParula07Red!50,
                    draw=myParula07Red!50!black,
                },
            block-hon/.style = {
                    block,
                    bottom color=myParula05Green!50,
                    draw=myParula05Green!50!black,
                },
            block-blank/.style = {
                    minimum width=0.3cm,
                    minimum height=0.3cm,
                    rounded corners,
                    inner sep=0,
                },
            link/.style = {
                    -latex,
                },
            link-adv/.style = {
                    link,
                },
            link-hon/.style = {
                    link,
                },
        }
}

\usepackage{multirow,booktabs}
\usepackage{makecell}

\usepackage{algorithm}
\usepackage{algorithmicx}
\usepackage[noend]{algpseudocode}

\algnewcommand{\algorithmicswitch}{\textbf{switch}}
\algdef{SE}[SWITCH]{Switch}{EndSwitch}[1]{\algorithmicswitch\ #1\ \algorithmicdo}{\algorithmicend\ \algorithmicswitch}%
\algtext*{EndSwitch}%

\algnewcommand{\algorithmiccase}{\textbf{case}}
\algdef{SE}[CASE]{Case}{EndCase}[1]{\algorithmiccase\ #1}{\algorithmicend\ \algorithmiccase}%
\algtext*{EndCase}%

\algnewcommand{\algorithmicon}{\textbf{on}}
\algdef{SE}[ON]{On}{EndOn}[1]{\algorithmicon\ #1\ \algorithmicdo}{\algorithmicend\ \algorithmicon}%
\algtext*{EndOn}%

\algnewcommand{\algorithmicat}{\textbf{at}}
\algdef{SE}[AT]{At}{EndAt}[1]{\algorithmicat\ #1\ \algorithmicdo}{\algorithmicend\ \algorithmicat}%
\algtext*{EndAt}%

\algnewcommand{\algorithmicrealfunction}{\textbf{function}}
\algdef{SE}[REALFUNCTION]{RealFunction}{EndRealFunction}[1]{\algorithmicrealfunction\ #1\ \algorithmicdo}{\algorithmicend\ \algorithmicrealfunction}%
\algtext*{EndRealFunction}%

\algnewcommand{\algorithmicthroughout}{\textbf{do throughout}}
\algdef{SE}[Throughout]{Throughout}{EndThroughout}[1]{\algorithmicthroughout\ #1\ \algorithmicdo}{\algorithmicend\ \algorithmicthroughout}%
\algtext*{EndThroughout}%

\algrenewcommand{\algorithmicdo}{}
\algrenewcommand{\algorithmicthen}{}

\algnewcommand{\algorithmicgoto}{\textbf{goto}}%
\algnewcommand{\Goto}[1]{\algorithmicgoto~\ref{#1}}%

\algnewcommand{\algorithmicassert}{\textbf{assert}}%
\algnewcommand{\Assert}[1]{\algorithmicassert~{#1}}%

\algnewcommand{\algorithmicbreak}{\textbf{break}}%
\algnewcommand{\Break}[0]{\algorithmicbreak}%

\algnewcommand{\algorithmicwaiton}{\textbf{wait on}}%
\algnewcommand{\WaitOn}[1]{\algorithmicwaiton~{#1}}%

\algnewcommand{\LineComment}[1]{\State \(\triangleright\) \textit{#1}}
\algnewcommand{\InlineRequire}[1]{\textbf{require} {#1}}

\usepackage{refstyle}

\usepackage[sort&compress,capitalize,nameinlink]{cleveref}

\crefname{figure}{Fig.}{Figs.}
\Crefname{figure}{Fig.}{Figs.}

\crefname{table}{Tab.}{Tabs.}
\Crefname{table}{Tab.}{Tabs.}

\crefname{section}{Sec.}{Secs.}
\Crefname{section}{Sec.}{Secs.}
\crefname{appendix}{App.}{Apps.}
\Crefname{appendix}{App.}{Apps.}

\crefname{algorithm}{Alg.}{Algs.}
\Crefname{algorithm}{Alg.}{Algs.}
\crefname{line}{l.}{ll.}
\Crefname{line}{l.}{ll.}

\crefname{proposition}{Prop.}{Props.}
\Crefname{proposition}{Prop.}{Props.}
\crefname{lemma}{Lem.}{Lems.}
\Crefname{lemma}{Lem.}{Lems.}
\crefname{theorem}{Thm.}{Thms.}
\Crefname{theorem}{Thm.}{Thms.}
\crefname{corollary}{Cor.}{Cors.}
\Crefname{corollary}{Cor.}{Cors.}
\crefname{definition}{Def.}{Defs.}
\Crefname{definition}{Def.}{Defs.}

\usepackage{color}

\newcommand{\myalgref}[2]{\cref{#1},~\cref{#2}}

\usepackage{xspace}

\newcommand{\Prob}[1]{\ensuremath{\operatorname{Pr}\left[#1\right]}}
\newcommand{\Exp}[1]{\ensuremath{\mathbb E\left[#1\right]}}
\newcommand{\poly}{\ensuremath{\operatorname{poly}}}
\newcommand{\negl}{\ensuremath{\operatorname{negl}}}

\newcommand{\intvl}[2]{\ensuremath{\left(#1,#2\right]}}
\newcommand{\intvlset}[0]{\ensuremath{\CI}}
\newcommand{\intvlg}[0]{\ensuremath{\succ}}
\newcommand{\intvll}[0]{\ensuremath{\prec}}
\newcommand{\intvlgeq}[0]{\ensuremath{\succeq}}
\newcommand{\intvlleq}[0]{\ensuremath{\preceq}}
\newcommand{\intvleq}[0]{\ensuremath{\asymp}}

\renewcommand{\Hat}[1]{\ensuremath{H_{#1}}}
\newcommand{\Hin}[2]{\ensuremath{H_{\intvl{#1}{#2}}}}
\newcommand{\Aat}[1]{\ensuremath{A_{#1}}}
\newcommand{\Ain}[2]{\ensuremath{A_{\intvl{#1}{#2}}}}
\newcommand{\Qat}[1]{\ensuremath{Q_{#1}}}
\newcommand{\Qin}[2]{\ensuremath{Q_{\intvl{#1}{#2}}}}

\newcommand{\predGood}[1]{\ensuremath{\mathsf{Good}(#1)}}
\newcommand{\predBad}[1]{\ensuremath{\mathsf{Bad}(#1)}}
\newcommand{\predEmpty}[1]{\ensuremath{\mathsf{Empty}(#1)}}

\newcommand{\sltGood}[0]{Good\xspace}
\newcommand{\sltgood}[0]{good\xspace}
\newcommand{\sltBad}[0]{Bad\xspace}
\newcommand{\sltbad}[0]{bad\xspace}
\newcommand{\sltEmpty}[0]{Empty\xspace}
\newcommand{\sltempty}[0]{empty\xspace}

\newcommand{\Gat}[1]{\ensuremath{G_{#1}}}
\newcommand{\Gin}[2]{\ensuremath{G_{\intvl{#1}{#2}}}}
\newcommand{\Bat}[1]{\ensuremath{\overline{G}_{#1}}}
\newcommand{\Bin}[2]{\ensuremath{\overline{G}_{\intvl{#1}{#2}}}}
\newcommand{\probGood}[0]{\ensuremath{p_{\mathrm{G}}}}
\newcommand{\probEmpty}[0]{\ensuremath{p_{\mathrm{E}}}}

\newcommand{\Dat}[1]{\ensuremath{D_{#1}}}
\newcommand{\Din}[2]{\ensuremath{D_{\intvl{#1}{#2}}}}
\newcommand{\Nat}[1]{\ensuremath{\overline{D}_{#1}}}
\newcommand{\Nin}[2]{\ensuremath{\overline{D}_{\intvl{#1}{#2}}}}

\newcommand{\Xat}[1]{\ensuremath{X_{#1}}}
\newcommand{\Xin}[2]{\ensuremath{X_{\intvl{#1}{#2}}}}
\newcommand{\Yat}[1]{\ensuremath{Y_{#1}}}
\newcommand{\Yin}[2]{\ensuremath{Y_{\intvl{#1}{#2}}}}

\newcommand{\probPP}[0]{\ensuremath{p_{\mathrm{\sltpp}}}}

\newcommand{\predPP}[1]{\ensuremath{\mathsf{PPivot}(#1)}}
\newcommand{\predCP}[1]{\ensuremath{\mathsf{CPivot}(#1)}}
\newcommand{\Pat}[1]{\ensuremath{P_{#1}}}
\newcommand{\Pin}[2]{\ensuremath{P_{\intvl{#1}{#2}}}}

\newcommand{\sltPps}[0]{Ppivots\xspace}
\newcommand{\sltpp}[0]{ppivot\xspace}
\newcommand{\sltpps}[0]{ppivots\xspace}

\newcommand{\sltCps}[0]{Cpivots\xspace}
\newcommand{\sltcp}[0]{cpivot\xspace}
\newcommand{\sltcps}[0]{cpivots\xspace}

\newcommand{\timeslot}[0]{slot\xspace}
\newcommand{\timeslots}[0]{slots\xspace}

\newcommand{\Timeslots}[0]{Slots\xspace}
\newcommand{\iindex}[0]{index\xspace}
\newcommand{\iindices}[0]{indices\xspace}

\newcommand{\Iindices}[0]{Indices\xspace}

\newcommand{\BPO}[0]{BPO\xspace}
\newcommand{\BPOs}[0]{BPOs\xspace}

\newcommand{\slotduration}{\ensuremath{\tau}}
\newcommand{\blkrateslot}{\ensuremath{\rho}}
\newcommand{\blkratetime}{\ensuremath{\lambda}}
\newcommand{\blkratetimeGrowth}{\ensuremath{\blkratetime_{\mathrm{grwth}}}}
\newcommand{\blkratetimeHon}{\ensuremath{\blkratetime_{\mathrm{hon}}}}
\newcommand{\blkratetimeAdv}
{\ensuremath{\blkratetime_{\mathrm{adv}}}}
\newcommand{\blkratetimeSPV}
{\ensuremath{\blkratetime_{\mathrm{spv}}}}
\newcommand{\blkratetimeGrowthSilent}{\ensuremath{\blkratetime_{\mathrm{grwth}}^{\mathrm{privt}}}}
\newcommand{\blkratetimeGrowthTeaser}{\ensuremath{\blkratetime_{\mathrm{grwth}}^{\mathrm{teaser}}}}
\newcommand{\DeltaHeader}{\ensuremath{\Delta_{\mathrm{h}}}}
\newcommand{\goodsep}{\ensuremath{\nu}}
\newcommand{\goodsepbw}{\ensuremath{\widetilde{C}}}
\newcommand{\bwtime}{\ensuremath{C}}

\newcommand{\Tlive}{\ensuremath{T_{\mathrm{live}}}}
\newcommand{\Thorizon}[0]{\ensuremath{T_{\mathrm{hrzn}}}}
\newcommand{\tput}[0]{\ensuremath{\theta}}
\newcommand{\Ttput}[0]{\ensuremath{T_{\mathrm{txlim}}}}
\newcommand{\Ktput}[0]{\ensuremath{K_{\mathrm{txlim}}}}

\newcommand{\TliveReal}{\ensuremath{T_{\mathrm{live}}^{\mathrm{real}}}}

\newcommand{\Chain}{\CC}
\newcommand{\dC}{\mathrm{d}\Chain}

\newcommand{\Tree}{\CT}

\newcommand{\hT}{\mathrm{h}\Tree}

\newcommand{\LOG}[2]{\ensuremath{\mathsf{LOG}_{#1}^{#2}}}
\newcommand{\trunc}[1]{\ensuremath{^{\lceil #1}}}

\newcommand{\Khorizon}[0]{\ensuremath{K_{\mathrm{hrzn}}}}
\newcommand{\Kcp}[0]{\ensuremath{K_{\mathrm{cp}}}}

\newcommand{\eps}[0]{\ensuremath{\varepsilon}}
\newcommand{\epsGood}[0]{\ensuremath{\eps_{\mathrm{G}}}}

\newcommand{\iid}[0]{iid\xspace}
\newcommand{\iidPERIOD}[0]{\iid}

\newcommand{\Event}[0]{\ensuremath{\CE}}

\newcommand{\probGoodFormula}{\ensuremath{(1-\beta)\frac{\blkrateslot e^{-\blkrateslot(\goodsep+1)}}{1-e^{-\blkrateslot}}}}
\newcommand{\probPPFormula}{\ensuremath{(2\probGood - 1)^2 / \probGood}}

\newcommand{\BWEquation}{\ensuremath{(\goodsep+1)\slotduration \geq \DeltaHeader + \goodsepbw/\bwtime}}

\newcommand{\alphaLowerTailX}[0]{\ensuremath{\alpha_2}}
\newcommand{\alphaLowerTailPP}[0]{\ensuremath{\alpha_3}}

\newcommand{\Ind}[1]{\ensuremath{\Ione_{\{#1\}}}}

\newcommand{\dlrulelong}{\ensuremath{\mathcal{D}_{\mathsf{long}}}}

\newcommand{\tx}[0]{\ensuremath{\mathsf{tx}}}
\newcommand{\txs}[0]{\ensuremath{\mathsf{txs}}}

\newcommand{\Env}[0]{\CZ}
\newcommand{\Adv}[0]{\CA}

\newcommand{\teaserattack}[0]{teasing strategy\xspace}
\newcommand{\Teaserattack}[0]{Teasing strategy\xspace}
\newcommand{\TeaserAttack}[0]{Teasing Strategy\xspace}

\newcommand{\PoSteaserattack}[0]{equiv-teasing attack\xspace}
\newcommand{\PoSTeaserattack}[0]{Equiv-teasing attack\xspace}
\newcommand{\PoSTeaserAttack}[0]{Equiv-Teasing Attack\xspace}

\newcommand{\greedyattack}[0]{forking attack\xspace}

\newcommand{\GreedyAttack}[0]{Forking Attack\xspace}

\newcommand{\rulegreedy}[0]{greedy\xspace}
\newcommand{\ruleGreedy}[0]{Greedy\xspace}
\newcommand{\rulelc}[0]{longest-header-chain\xspace}

\newcommand{\ruleLc}[0]{Longest-header-chain\xspace}
\newcommand{\confDepth}{\ensuremath{k_{\mathrm{conf}}}}
\newcommand{\protocol}{\ensuremath{\Pi^{\blkrateslot,\slotduration,\confDepth}}}

\newcommand{\TRUE}[0]{\ensuremath{\mathtt{true}}}
\newcommand{\FALSE}[0]{\ensuremath{\mathtt{false}}}
\newcommand{\BLOCKUNKNOWN}[0]{\mathtt{UNKNOWN}}

\newcommand{\FtreePoW}[0]{\ensuremath{\CF_{\mathrm{hdrtree}}^{\mathrm{PoW},\blkrateslot,\slotduration}}}
\newcommand{\FtreePoS}[0]{\ensuremath{\CF_{\mathrm{hdrtree}}^{\mathrm{PoS},\blkrateslot,\slotduration}}}

\newcommand{\TxsMap}[0]{\ensuremath{\mathrm{blkTxs}}}

\newcommand{\Lottery}[0]{\ensuremath{\mathrm{lottery}}}

\newcommand{\genesisHeaderChain}[0]{\ensuremath{\mathsf{genesis}\Chain}}

\newcommand{\FormalVersion}[1]{Formal version: #1}

\newcommand{\ProtShort}[0]{BlaNC\xspace}
\newcommand{\ProtMid}[0]{Blanking NC\xspace}
\newcommand{\ProtLong}[0]{Blanking Nakamoto Consensus\xspace}

\newcommand{\keqproof}{\ensuremath{k_{\mathrm{epf}}}}

\title[Nakamoto Consensus under Bounded Processing Capacity]{Nakamoto Consensus under Bounded Processing Capacity}

\author{Lucianna Kiffer}
\orcid{0000-0003-2022-7993}
\affiliation{%
    \institution{ETH Z\"urich}%
    \city{Zurich}%
    \country{Switzerland}%
}
\email{lkiffer@ethz.ch}
\author{Joachim Neu}
\orcid{0000-0002-9777-6168}
\affiliation{%
    \institution{Stanford University}%
    \city{Stanford}%
    \state{CA}%
    \country{USA}%
}
\email{jneu@stanford.edu}
\author{Srivatsan Sridhar}
\orcid{0000-0002-1880-570X}
\affiliation{%
    \institution{Stanford University}%
    \city{Stanford}%
    \state{CA}%
    \country{USA}%
}
\email{svatsan@stanford.edu}
\author{Aviv Zohar}
\orcid{0000-0001-8539-9222}
\affiliation{%
    \institution{The Hebrew University}%
    \city{Jerusalem}%
    \country{Israel}%
}
\email{avivz@cs.huji.ac.il}
\author{David Tse}
\orcid{0000-0003-1460-5900}
\affiliation{%
    \institution{Stanford University}%
    \city{Stanford}%
    \state{CA}%
    \country{USA}%
}
\email{dntse@stanford.edu}

\thanks{LK, JN, SS and AZ are listed alphabetically.
\ifshortVersion%
    Full version with appendices: \cite{full-version}%
\fi}

\newcommand{\gitSourceUrl}[0]{\url{https://github.com/avivz/finitebwlc}}
\makeatletter
\def\ps@headings{%
\def\@oddhead{\mbox{}\scriptsize\rightmark \hfil \thepage}%
\def\@evenhead{\scriptsize\thepage \hfil \leftmark\mbox{}}}
\makeatother
\crefname{section}{Sec.}{Secs.}
\Crefname{section}{Sec.}{Secs.}
\crefname{appendix}{App.}{Apps.}
\Crefname{appendix}{App.}{Apps.}
\crefname{subsection}{Sec.}{Secs.}
\Crefname{subsection}{Sec.}{Secs.}
\crefname{subappendix}{App.}{Apps.}
\Crefname{subappendix}{App.}{Apps.}
\crefname{subsubsection}{Sec.}{Secs.}
\Crefname{subsubsection}{Sec.}{Secs.}
\crefname{subsubappendix}{App.}{Apps.}
\Crefname{subsubappendix}{App.}{Apps.}
\crefname{subsubsubsection}{Sec.}{Secs.}
\Crefname{subsubsubsection}{Sec.}{Secs.}
\crefname{subsubsubappendix}{App.}{Apps.}
\Crefname{subsubsubappendix}{App.}{Apps.}

\newcommand{\myparagraph}[1]{\smallskip\noindent\textbf{#1.}~~}

\usepackage{thmtools}

\begin{document}
\begin{abstract}
    For Nakamoto's longest-chain consensus protocol, whose proof-of-work (PoW) and proof-of-stake (PoS) variants power major blockchains such as Bitcoin and Cardano, we revisit the classic problem of the security--performance tradeoff:
    Given a network of nodes with finite 
    communication- and computation-resources, against what fraction of adversary power is Nakamoto consensus (NC) secure for a given block production rate?
    State-of-the-art analyses of NC fail to answer this question, because their \emph{bounded-delay} model does not capture the rate limits to nodes' processing of blocks, which cause congestion when blocks are released in quick succession.
    We develop a new analysis technique to prove a refined security--performance tradeoff for PoW NC in a \emph{bounded-capacity} model.
    In this model, we show that, in contrast to the classic bounded-delay model, Nakamoto's private attack is no longer the worst attack, and a new attack we call the \emph{\teaserattack}, that exploits congestion, is strictly worse.
    In PoS, equivocating blocks can exacerbate congestion, making  traditional PoS NC insecure except at very low block production rates.
    To counter such equivocation spamming, we present a variant of PoS NC we call \emph{\ProtMid} (\emph{\ProtShort}), which achieves the \emph{same resilience} as PoW NC.
\end{abstract}


\maketitle

\section{Introduction}
\label{sec:introduction}

In order to remain secure against adversaries controlling up to $50\%$ of the network,
blockchains that utilize Nakamoto's proof-of-work (PoW) longest-chain consensus protocol~\cite{nakamoto_paper,backbone} 
have been parameterized
to leave a security margin between the 
throughput
under normal operation and each node's capacity limits.
For instance, in expectation, Bitcoin produces only one block of transactions every ten minutes, though it usually only takes a few seconds for a node to download and verify a block's contents~\cite{decker}.
On the other hand, Bitcoin Cash forked off to increase the block size for better throughput, a proposal whose security implications were hotly debated~\cite{btc-blocksize-war}.
The fundamental question that protocol designers face is: \emph{%
    What is the security--performance trade-off between the block production rate
    (relative to the nodes' capacity limit)
    and the fraction of adversary power that the protocol tolerates?}
In this work, we 
show the inadequacy of the bounded-delay model that most previous works utilized to analyze the security of Nakamoto consensus (NC)~\cite{kiayias2017ouroboros,backbone,dem20,sleepy,ren,tight_bitcoin,pss16,kiffer2018better}, and instead provide security analysis in a bounded-capacity model that better captures real-world effects such as congestion due to a backlog of blocks that need to be communicated and validated by nodes.

\begin{figure}[tb]%
    \centering%
    \begin{tikzpicture}[]%
        \scriptsize
        \begin{axis}[
                mysimpleresilienceplot01,
                name=plot1,
                xmode=log,
                xlabel={Block production rate $\lambda$},
                height=0.5\linewidth,
            ]

            \addplot [draw=none,name path=xaxis,domain={1e-8:1e8}] {0};
            \addplot [draw=none,name path=xaxisplus1,domain={1e-8:1e8}] {1};

            \addplot [black,dashed,no marks,name path=resiliencebddelay] table [x=relblockfrequency,y=beta] {figures/fig-comparison-bddelay-bdbandwidth-bddelay.txt};
            \label{leg:comparison-bddelay-bdbandwidth-privateattack}

            \addplot [myparula51,fill opacity=0.5,pattern=crosshatch dots,pattern color=myParula05Green] fill between [of=resiliencebddelay and xaxis];
            \addplot [myparula71,pattern=north east lines,pattern color=myParula07Red] fill between [of=resiliencebddelay and xaxisplus1];

            \addplot [myparula71,thick,no marks,name path=teaserattack] table [x=lbyC,y=beta] {figures/fig-comparison-bddelay-bdbandwidth-exp_teaser-attack.txt};
            \addplot [myparula71,fill opacity=0.3] fill between [of=teaserattack and xaxisplus1];

            \addplot [myparula51,no marks,name path=resiliencebdbandwidth] table [x=lbyC,y=beta] {figures/fig-comparison-bddelay-bdbandwidth-bdbandwidth-newresult.txt};

            \addplot [myparula51,fill opacity=0.45] fill between [of=resiliencebdbandwidth and xaxis];

            \node [align=center] (attack) at (axis cs:1e-1,0.1) {\Teaserattack\\{}(\cref{sec:teaser-attack})};
            \draw [-latex,shorten >=1pt,shorten <=-2pt] (attack) -- (axis cs:1.5,0.2);

            \node [align=center] (security) at (axis cs:5e-3,0.40) {Our security proof\\{}(\cref{thm:safety-and-liveness-pow})};
            \draw [-latex,shorten >=2pt] (security) -- (axis cs:1e-3,0.27);

            \node [align=center] (bddelay) at (axis cs:1e1,0.40) {Private attack\\{}\cite{dem20,tight_bitcoin}};
            \draw [-latex] (bddelay) -- (axis cs:4e0,0.21);

        \end{axis}
    \end{tikzpicture}%
    \vspace{-0.5em}%
    \caption[]{%
        Regions of
        fraction $\beta$ of adversary nodes and
        block production rate $\lambda$
        with
        security proofs
        (\tikz[x=0.75em,y=0.75em]{ \draw [draw=myParula05Green,thick,fill=myParula05Green,fill opacity=0.3] (0,0) rectangle (1,1); })
        and attacks
        (\tikz[x=0.75em,y=0.75em]{ \draw [draw=myParula07Red,thick,fill=myParula07Red,fill opacity=0.3] (0,0) rectangle (1,1); })
        for NC under a fixed processing capacity of $C=1$ block per second.
        Analysis in the bounded-delay model \cite{dem20,tight_bitcoin} (with $\Delta = 1\;\mathrm{s}$) proves that the private attack
        (\ref{leg:comparison-bddelay-bdbandwidth-privateattack})
        succeeds
        (\tikz[x=0.75em,y=0.75em]{ \draw [draw=none,fill opacity=1,pattern=north east lines,pattern color=myParula07Red] (0,0) rectangle (1,1); })
        iff
        $\beta \geq \frac{1-\beta}{1+(1-\beta)\blkratetime}$, and that for all other values of $\beta,\lambda$, no attack succeeds
        (\tikz[x=0.75em,y=0.75em]{ \draw [draw=none,fill opacity=1,pattern=crosshatch dots,pattern color=myParula05Green] (0,0) rectangle (1,1); }).
        Our \teaserattack exploits congested block processing and succeeds at lower adversary $\beta$ than the private attack
        (\tikz[x=0.75em,y=0.75em]{ \draw [draw=myParula07Red,thick,fill=myParula07Red,fill opacity=0.3] (0,0) rectangle (1,1); }, \tikz[x=0.75em,y=0.75em]{ \draw [draw=none,fill opacity=1,pattern=north east lines,pattern color=myParula07Red] (0,0) rectangle (1,1); }).
        Our analysis in a bounded-capacity model 
        yields
        a security region
        (\tikz[x=0.75em,y=0.75em]{ \draw [draw=myParula05Green,thick,fill=myParula05Green,fill opacity=0.3] (0,0) rectangle (1,1); })
        for PoW NC.%
    }%
    \label{fig:comparison-bddelay-bdbandwidth}%
\end{figure}%

In PoW NC,
collectively starting with a well-known ``genesis'' block,
each node continuously works to solve a computational
puzzle to extend the longest chain of blocks 
it sees with a new block containing pending transactions (``\emph{mining} a new block'').
When successful, the node pushes the new block's \emph{header} to the network,
and makes its \emph{content} available for download.
In order to extend a chain, nodes must first \emph{process}, \ie, download and verify, the content of blocks in that chain, to ensure that the content is both \emph{available} and \emph{valid}.
Downloading blocks may take time, especially if blocks are extremely large~\cite{decker}, but in systems that contain smart contracts, 
even smaller blocks may take a while to process---mostly due to the time it takes to execute and validate smart contracts~\cite{demystifying-incentives}.

In PoW, block production occurs at random times, which makes the processing load of the network bursty.
Moreover, the adversary can selectively withold its own mined blocks and release them opportunistically.
Both these factors further stress 
the processing (communication, computation, ...) capacities of nodes.
With limited processing capacity, during times of high load, blocks will be \emph{queued} for processing. 
Since nodes cannot mine new blocks extending chains that they have not yet fully processed, queueing further delays the growth of the honest nodes' chain. As the security of NC is based on the honest chain outgrowing any adversary chain,
the reduced growth of the honest nodes' chain makes it easier for an adversary to attack the system.
To analyze security under such effects,
it is important to consider the \emph{scheduling policy} that nodes use in deciding which blocks to download and verify first, given a set of new block headers.
Since a node extends its longest chain to produce new blocks, an obvious policy is to first process blocks along the longest header chain that the node has seen. Indeed, this policy can be found in the Bitcoin implementation \cite{btcdevp2pnetworkheadersfirst}.

\myparagraph{Limitations of the Bounded-Delay Model}
Previous work has focused on the security analysis of Bitcoin in the synchronous setting: All messages are assumed to arrive after a maximum delay of $\Delta$~\cite{kiayias2017ouroboros,backbone,dem20,sleepy,ren,tight_bitcoin,pss16,kiffer2018better}.
Using this model, 
\cite{dem20,tight_bitcoin} calculate a tight bound on the fraction $\beta$ of adversary nodes, for given block production rate $\blkratetime$ and delay bound $\Delta$, for which the protocol is secure against all attacks.
However, the $\Delta$-delay model assumes that the delay is the same \emph{irrespective of the total processing load}, and specifically, that the 
adversary cannot manipulate the load to its advantage.
Thus, the model fails to capture the security implications of bursty release of blocks by an adversary or due to the stochastic nature of PoW mining even by honest nodes alone.

The bounded-delay analysis \cite{dem20,tight_bitcoin} 
concludes that the
well-known
\emph{private attack}~\cite{nakamoto_paper} (along with delaying every message by $\Delta$) is the worst-case attack strategy since 
its attack threshold matches
the
security threshold,
\ie,
under parameters where the private attack fails, the analysis concludes that all other attacks must fail, too.
If we only consider the private attack and low block production rates, then the bounded-delay analysis, with $\Delta$ taken as the time to process one block, is a good approximation to calculate the fraction of adversary power with which the attack succeeds (see \cref{fig:comparison-bddelay-bdbandwidth}, validated by simulations in \cref{sec:experiments}).
This is because during the private attack, the adversary does not release any blocks (only ``benign'' random congestion), and the effect of bursty honest mining is less significant at low block rates.

However, there are other strategies in which the adversary adds to the processing load to increase queuing delays.
We simulate one such strategy, the \teaserattack (\cref{sec:teaser-attack}), that is stronger than the private attack, \ie, it succeeds in regions of $(\lambda, \beta)$ where the private attack does not (\cref{fig:comparison-bddelay-bdbandwidth}).
In the \teaserattack, the adversary ``teases'' honest nodes to process a longer chain it announces, but makes this effort ``useless'' by not releasing the block contents for the entire chain.
The adversary effectively doubles the processing load and queuing delays, thus slowing the growth of the honest nodes' chain, while the adversary builds a longer chain to break security. This halves the maximum secure block rate $\blkratetime$ for any given $\beta$ (\cref{fig:comparison-bddelay-bdbandwidth}).
While the concrete \emph{quantitative} impact of this attack may be considered modest, it highlights \emph{conceptual} limitations of earlier analyses and emphasizes the need for security analysis under more realistic network models,
especially to rule out that unbeknown to us there could be even more serious queuing-based attacks.

\myparagraph{Security Bounds under Bounded Processing Capacity}
To
re-establish the security of NC in a more realistic model,
we adopt the \emph{bounded-capacity} model from~\cite{bwlimitedposlc}.
Under this model, we consider the scheduling policy as a part of the protocol description as it affects the security of the protocol.
\emph{Henceforth, 
we continue to use the verb
``to process'' to abstract a variety (or combination) of tasks (communication/download, computation/execution/verification, input-output/storage access, ...) that are typically subject to rate constraints in real-world systems, and we refer to the corresponding rate limit abstractly as ``(processing) capacity''.}

\begin{result}
\label{res:result-pow-security}
    Using the bounded-capacity model and a novel analysis technique, we characterize a region of block mining rate $\blkratetime$ and adversary fraction $\beta$ for which we prove that PoW NC, with a wide range of suitable scheduling policies, is secure (\cref{thm:safety-and-liveness-pow}).
    This region is shown in \cref{fig:comparison-bddelay-bdbandwidth}.
    Specifically, this analysis expands the set of adversary strategies
    captured by earlier bounded-delay analyses to include adversary strategies
    that exploit effects from bounded capacity.
\end{result}

In \cref{fig:bitcoin-cardano-resilience-bandwidth}, we plot the adversary resilience versus bandwidth requirement for PoW NC with cautious (\eg Bitcoin) and ambitious (\eg Bitcoin Cash) parameters.
It shows the importance of modeling and studying congestion effects on security,
in particular, for protocols that aim for maximum performance,
and our analysis provides tools to do so.
While our work demonstrates that earlier analyses
have failed to capture some security-critical phenomena,
the quantitative gap between our best-known attack
(\cref{fig:comparison-bddelay-bdbandwidth}~\tikz[x=0.75em,y=0.75em]{ \draw [draw=myParula07Red,thick,fill=myParula07Red,fill opacity=0.3] (0,0) rectangle (1,1); })
and our best-known security analysis
(\cref{fig:comparison-bddelay-bdbandwidth}~\tikz[x=0.75em,y=0.75em]{ \draw [draw=myParula05Green,thick,fill=myParula05Green,fill opacity=0.3] (0,0) rectangle (1,1); })
points to a need for future
work.\footnote{%
We focus on the security--throughput tradeoff
of NC, \ie, for what tuples $(\lambda, \beta)$
does the consensus-failure probability $\varepsilon$
decay exponentially to zero
under \emph{some appropriate}
scaling of the confirmation latency $\Tlive$.
On a finer point, in our 
bounded-capacity analysis,
latency scales 
$\Tlive \sim (\log(1/\varepsilon))^2$ (\cref{thm:safety-and-liveness-pow}),
in contrast to earlier bounded-delay analyses
that required only
$\Tlive \sim \log(1/\varepsilon)$~\cite{tight_bitcoin,linear-cons-pos,ren}.
Exploring the possibility of tighter latency scaling under bounded capacity
requires future work.
}

\begin{figure}[tb]%
    \centering%
    \begin{tikzpicture}[]%
        \scriptsize
        \begin{axis}[
                mysimpleresilienceplot01,
                xmode=log,
                xlabel={Bandwidth requirement $C$ (Mbps)},
                xmin=1e-3,xmax=1e2,
                x dir=reverse,
                legend columns=2,
                grid=both,
                height=0.5\linewidth,
                legend style={
                        xshift=-1.5em,
                    },
            ]

            \addlegendimage{gray,dashed,mark=none};
            \addlegendentry{Bounded-delay security (private attack)};

            \addlegendimage{area legend,myparula21,fill=myParula02Orange,thin,draw=none};
            \addlegendentry{Cautious ($4\;\mathrm{MB}$ blocks)};

            \addlegendimage{gray,mark=none};
            \addlegendentry{\Teaserattack (\cref{sec:teaser-attack})};

            \addlegendimage{area legend,myparula11,fill=myParula01Blue,thin,draw=none};
            \addlegendentry{Ambitious ($32\;\mathrm{MB}$ blocks)};

            \addplot [myparula21,dashed,no marks,mark=none,name path=bandwidth] table [y=beta,x=C(Mbps)] {figures/fig-bitcoin-resilience-delay.txt};
            \label{leg:bitcoin-cardano-resilience-bandwidth-pow-bd};

            \addplot [myparula11,dashed,no marks,mark=none,name path=bandwidth] table [y=beta,x=C(Mbps)] {figures/fig-bitcoin-cash-resilience-delay.txt};
            \label{leg:bitcoin-cash-resilience-bandwidth-pow-bd};

            \addplot [myparula23,draw=none,mark size=1.8pt,draw=none,line width=1,only marks] coordinates { (0.39402665088757344,0.4831) };
            \label{leg:bitcoin-cardano-resilience-bandwidth-pow-bd-mark};

            \addplot [myparula13,draw=none,mark size=1.8pt,draw=none,line width=1,only marks] coordinates { (0.40036293995196154,0.3751) };
            \label{leg:bitcoin-cash-resilience-bandwidth-pow-bd-mark};

            \addplot [myparula21,no marks,mark=none,name path=bandwidth] table [y=beta,x=C(Mbps)] {figures/fig-bitcoin-teasing-attack-bandwidth.txt};
            \label{leg:bitcoin-cardano-resilience-bandwidth-pow-teaser};

            \addplot [myparula11,no marks,mark=none,name path=bandwidth] table [y=beta,x=C(Mbps)] {figures/fig-bitcoin-cash-teasing-attack-bandwidth.txt};
            \label{leg:bitcoin-cash-resilience-bandwidth-pow-teaser};

            \addplot [myparula21,mark=*,mark size=1.5pt,draw=none,line width=0,only marks] coordinates { (0.4011622161735472,0.46749793788576444) };
            \label{leg:bitcoin-cardano-resilience-bandwidth-pow-teaser-mark};

            \addplot [myparula11,mark=*,mark size=1.5pt,draw=none,line width=0,only marks] coordinates { (0.38846308875683366,0.2767956260599464) };
            \label{leg:bitcoin-cash-resilience-bandwidth-pow-teaser-mark};

        \end{axis}
    \end{tikzpicture}%
    \vspace{-0.5em}%
    \caption[]{%
        For cautiously parameterized PoW NC
        (\eg, Bitcoin's
        $\blkratetime = 1/600\;\mathrm{blocks/s}$,
        block size $4\;\mathrm{MB}$,
        recommended min.\ per-node bandwidth 
        $0.4\;\mathrm{Mbps}$~\cite{bitcoin_requirements,DBLP:conf/fc/CromanDEGJKMSSS16}),
        earlier analyses assuming bounded delay
        (\ref{leg:bitcoin-cardano-resilience-bandwidth-pow-bd})
        predicted security
        against any adversary controlling up to 48\% of hash power
        (\ref{leg:bitcoin-cardano-resilience-bandwidth-pow-bd-mark}),
        including Nakamoto's private attack \cite{nakamoto_paper},
        which was 
        concluded
        to be worst-case.
        The \teaserattack still requires $46\%$ adversary
        (\ref{leg:bitcoin-cardano-resilience-bandwidth-pow-teaser}a, \ref{leg:bitcoin-cardano-resilience-bandwidth-pow-teaser-mark}).
        In contrast, 
        PoW NC parameterized ambitiously
        (\eg, Bitcoin Cash's same $\lambda$, but max.\ 
        block size $32\;\mathrm{MB}$, same bandwidth~\cite{bitcoin_cash_requirements})
        withstands only a $37\%$ private attacker
        (\ref{leg:bitcoin-cash-resilience-bandwidth-pow-bd-mark}), while the \teaserattack resilience drops to $27\%$
        (\ref{leg:bitcoin-cash-resilience-bandwidth-pow-teaser-mark}).
    }%
    \label{fig:bitcoin-cardano-resilience-bandwidth}%
\end{figure}%

\myparagraph{Proof-of-Stake (PoS) NC}
Nakamoto consensus has been adapted to proof-of-stake in protocols of the Ouroboros~\cite{kiayias2017ouroboros,david2018ouroboros,badertscher2018ouroboros} and Sleepy Consensus~\cite{sleepy,snowwhite} families.
In PoS NC, the block production lottery is independent of the block's content or parent~\cite{pos_paper}.
This, unlike in PoW NC, allows an adversary to \emph{reuse} a ``winning PoS lottery ticket'' to create infinitely many valid blocks (called \emph{equivocations}).
As observed in~\cite{bwlimitedposlc}, the adversary can 
\emph{spam} nodes with many \emph{equivocating} blocks, aggravating the problem of congestion.
While slashing~\cite{weaksubjectivity,casper,aadilemma,bftforensics} may deter \emph{rational} adversaries to \emph{some} extent,
protocols need to tolerate equivocations to 
handle plausibly irrational \emph{Byzantine} adversaries~\cite{bwlimitedposlc}.

Analytical work \cite{bwlimitedposlc} gives a security proof for PoS NC in the bounded-capacity model.
However,
\cite{bwlimitedposlc} proves security only when
nodes have enough capacity so
that for each block, 
they can process potentially different versions of its $k$ predecessors, where $k$ is the confirmation depth chosen for the chain.
This increases the network load by $k$ times, thus reducing the maximum secure block rate $\blkratetime$ by $k$ times
(\cref{fig:comparison-bddelay-bdbandwidth-pos}).
Decreasing the probability of consensus failure requires increasing $k$, which means that for security with overwhelming probability, the throughput must approach zero.
This is not merely an artifact of the security analysis of \cite{bwlimitedposlc}:
Augmenting our \teaserattack with equivocations demonstrates this behavior (we discuss this in \cref{sec:experiments} and \appendixRef{\cref{sec:pos-teaser-attack}}).
On the other hand, PoW NC does not suffer from such vanishing throughput (\cref{fig:comparison-bddelay-bdbandwidth}).

\begin{figure}[tb]%
    \centering%
    
    \begin{tikzpicture}[]%
        \footnotesize
        \begin{axis}[
                mysimpleresilienceplot01,
                name=plot1,
                xmode=log,
                xlabel={Block production rate $\lambda$},
                ylabel={Adversary res.\ $\beta$},
                legend columns=2,
                xmax=1e-2, xmin=2e-6,
            ]

            \addlegendimage{empty legend}
            \addlegendentry{%
                \tikz[x=2em,y=0.75em]{%
                    \draw [myparula21,thick] (0,0) -- (1,0);
                    \draw [myparula22,thick] (0,-0.5) -- (1,-0.5);
                    \draw [myparula23,thick] (0,-1.0) -- (1,-1.0);
                }%
                \hspace{0.5em}PoS NC security \cite{bwlimitedposlc}%
            };

            \addlegendimage{empty legend}
            \addlegendentry{%
                \tikz[x=2em,y=0.75em,baseline=-0.35em]{%
                    \draw [myParula01Blue,thick] (0,0) -- (1,0);
                }%
                \hspace{0.5em}\ProtShort security (this work)%
            };

            \addplot [draw=none,name path=xaxis,domain={1e-8:1e8}] {0};
            \addplot [draw=none,name path=xaxisplus1,domain={1e-8:1e8}] {1};

            \addplot [myParula01Blue,no marks,name path=resiliencebdbandwidth] table [x=lbyC,y=beta] {figures/fig-comparison-bddelay-bdbandwidth-bdbandwidth-newresult.txt};

            \addplot [myParula01Blue,fill opacity=0.2] fill between [of=resiliencebdbandwidth and xaxis];

            \addplot [myparula21,thick,no marks,name path=resilience10] table [x=lbyC,y=beta] {figures/fig-comparison-bddelay-bdbandwidth-bdbandwidth-oldresult-kappa10.txt};
            \addplot [myparula22,thick,no marks,name path=resilience100] table [x=lbyC,y=beta] {figures/fig-comparison-bddelay-bdbandwidth-bdbandwidth-oldresult-kappa100.txt};
            \addplot [myparula23,thick,no marks,name path=resilience1000] table [x=lbyC,y=beta] {figures/fig-comparison-bddelay-bdbandwidth-bdbandwidth-oldresult-kappa1000.txt};

            \addplot [myparula21,fill opacity=0.2] fill between [of=resilience10 and xaxis];
            \addplot [myparula21,fill opacity=0.2] fill between [of=resilience100 and xaxis];
            \addplot [myparula21,fill opacity=0.2] fill between [of=resilience1000 and xaxis];

        \end{axis}
    \end{tikzpicture}%
    \vspace{-0.5em}%
    \caption[]{%
        The region of 
        fraction $\beta$ of adversary nodes and
        block production rate $\lambda$
        where
        PoS NC is secure according to~\cite{bwlimitedposlc}
        (\tikz[x=0.75em,y=0.75em]{ \draw [draw=none,thick,fill=myParula02Orange,fill opacity=0.3] (0,0) rectangle (1,1); })
        shrinks as
        the NC confirmation depth increases, \ie,
        the desired consensus failure probability decreases
        (in order:
        \tikz[x=2em,y=0.75em,baseline=-0.35em]{%
        \draw [myparula21,thick] (0,0) -- (1,0);
        } 
        to
        \tikz[x=2em,y=0.75em,baseline=-0.35em]{%
        \draw [myparula23,thick] (0,0) -- (1,0);
        }).
        Thus, for the PoS NC protocol of~\cite{bwlimitedposlc},
        security requires vanishing throughput.
        In contrast, our new \ProtShort protocol
        achieves a security region (\tikz[x=0.75em,y=0.75em]{ \draw [draw=myParula01Blue,thick,fill=myParula01Blue,fill opacity=0.3] (0,0) rectangle (1,1); }) 
        that is
        independent of
        the desired consensus failure probability.
        Thus, \ProtShort is secure with non-vanishing constant throughput.
        (For all lines, processing capacity is fixed to $C=1$ block/s.)
    }%
    \label{fig:comparison-bddelay-bdbandwidth-pos}%
\end{figure}%

\begin{result}
    We propose and prove the security of a new PoS protocol we call \emph{\ProtMid} (\emph{\ProtShort}), a variant of PoS NC, that is secure in the same region of block production rate $\blkratetime$ and adversary fraction $\beta$ as PoW NC. Thus, similar to PoW NC, security with overwhelming probability requires increasing the confirmation depth, which affects latency, but not decreasing the block production rate, which affects throughput (see \cref{fig:comparison-bddelay-bdbandwidth-pos}).
\end{result}

On a high level, in \ProtShort,
honest nodes establish consensus on PoS lottery tickets for which
they have seen 
equivocations.
The contents of blocks from those equivocating PoS lottery tickets
are then \emph{blanked}, \ie, all blocks from those tickets are treated as empty blocks.
This absolves honest nodes from processing more than one block
per PoS lottery ticket, restoring the non-equivocation
behavior of PoW from a capacity point of view.
From there,
the security proof closely follows that of PoW NC.\footnote{%
The confirmation latency of \ProtShort under bounded capacity scales quadratically with the security parameter, just like PoW NC's latency.}

Blanking block contents
undermines predictability of transaction validity
(\cf \cref{sec:introduction-methods-throughputloss}).
In particular, it is harder to ensure,
at the time of composing a block, whether
transactions
are able to pay their fees.
Many modern consensus protocols share this problem
(\eg, \cite{spiegelman2022bullshark,danezis2022narwhal,dispersedledger,honeybadger,al2019lazyledger}).
We suggest some solutions in \cref{sec:predictablevalidity}.

\subsection{Related Works}
\label{sec:introduction-relatedworks}

Earlier works have analyzed the security of
PoW~\cite{backbone,nakamoto_paper,dem20,pss16,kiffer2018better,ren,tight_bitcoin} and
PoS~\cite{kiayias2017ouroboros,david2018ouroboros,badertscher2018ouroboros,sleepy,snowwhite,dem20,pos_paper} NC
in the bounded-delay model.
Our analysis builds on tools from several of these works, primarily pivots~\cite{sleepy} (Nakamoto blocks~\cite{dem20}) and convergence opportunities~\cite{pss16,sleepy,kiffer2018better} (or similar~\cite{dem20,ren}).
Markov decision processes
were used~\cite{sompolinsky2016bitcoins, gervais2016security} 
to computationally find optimal attack strategies,
assuming honest nodes do not suffer any delay.

Limitations of the bounded-delay model have been observed in previous work~\cite{prism,near-optimal-thruput,bwlimitedposlc}.
To use the bounded-delay model to set the protocol's block production rate, one needs to find the value of the bound $\Delta$.
This is tricky because unlike the capacity limit, which is a physical limit of the hardware used, delay depends on the network load.
One approach is to set the delay to the time taken to process one block, \ie, $\Delta = 1/\bwtime$.
While this may be reasonable at rates much smaller than the capacity (as processing queues are mostly empty), queuing delay breaks this bound otherwise.
In \cite{near-optimal-thruput}, a queuing model is used to 
calculate a delay bound that holds throughout the execution with overwhelming probability.
However, such a tail bound is too pessimistic because the queuing delays cannot always be large, due to limited block production.
In contrast, our finer-grained analysis captures limited block production.
Another work \cite{longest-chain-random-delay} analyzes security in a random (\iid) delay model.
However,
the network load, hence queuing delay, is not purely a random process, but is controlled by the adversary.
Network experiments~\cite{decker,kiffer2021under,revisiting-asynchronous} help estimate the delay distribution but cannot show us the impact of all possible adversary manipulations.

In analytical work \cite{bwlimitedposlc}, the bounded-capacity model captures adversarial manipulations.
Our paper's bounded-capacity model is that of~\cite{bwlimitedposlc}. Our paper differs from~\cite{bwlimitedposlc} two-fold:
(a) Only PoS is studied in \cite{bwlimitedposlc}. Due to equivocations in PoS, the security bounds of \cite{bwlimitedposlc} are too pessimistic for PoW NC. We develop \emph{new analysis machinery} (\cf~\cref{sec:introduction-methods-analysis}) to prove security \emph{for PoW}. Furthermore, the attack in \cite{bwlimitedposlc} does not apply to PoW, while our teasing strategy does.
(b) PoS NC with the freshest-block policy proposed in \cite{bwlimitedposlc} is secure only when its throughput approaches zero. In contrast, our \emph{new PoS protocol}, \ProtShort, is secure \emph{with non-vanishing throughput}.

Concurrently, \cite{dag-pow-bandwidth} analyzes specific congestion-based attacks on PoW DAG protocols but does not provide a security proof against all attacks.
Propagation delays also exacerbate selfish mining strategies~\cite{zhang-slower-block}, and congestion is another way to increase propagation delays.
However, the goal of this work is to analyze \emph{security} under bounded capacity, and selfish mining does not affect the two security properties of consensus: safety and liveness. It affects incentives and fairness, which are orthogonal.

Capacity limits apply not only to downloads but also to execution of transactions and smart contracts.
For instance,
earlier works \cite{brokenmetre} have shown that
execution times can vary by orders of magnitude
between benign 
and maliciously crafted transactions
(with equal gas consumption).
While download and execution are similar in that the time taken increases with the number of transactions, they are different in some aspects.
Execution is harder to parallelize due to transactions that depend on each other.
Methods to parallelize execution of smart contracts are studied in \cite{adding-concurrency-smart-contracts,speculative-concurrency-eth}.
Additionally, executing transactions can be delayed until after confirmation, such as in~\cite{al2019lazyledger,tuxedo}, but delaying downloads could lead to data availability attacks (\cf~\cref{sec:sapos}).

\subsection{Overview of Key Ideas and Methods}
\label{sec:introduction-methods}

\subsubsection{New Analysis Technique}
\label{sec:introduction-methods-analysis}

\begin{figure}[tb]%
    \centering%
    \begin{tikzpicture}[x=2.55cm]%
        \footnotesize

        \begin{scope}
            \node [anchor=west,align=left,xshift=-0.7cm] at (0,1.5) {\small\textbf{(a)\ \ Sleepy analysis \cite{sleepy,bwlimitedposlc}}:};

            \node [align=center] (probarg) at (0,0) {Probabil.\\argument};
            \node [align=center,black!50] (ppivot) at (1.1,0) {Probabil.\\pivot};
            \node [align=center,black!50] (cpivot) at (1.9,0) {Combinat.\\pivot};
            \node [align=center] (security) at (3,0) {Security};
            \node [align=center,black!50] (model) at (1.5,0.75) {Bounded delay model};

            \draw [double equal sign distance,-Implies,shorten >=5pt] (probarg) -- (ppivot) node [midway,above,align=center] {Exists\\one};
            \draw [double equal sign distance,Implies-Implies,black!30] (ppivot) -- (cpivot);
            \draw [double equal sign distance,-Implies,shorten <=3pt] (cpivot) -- (security) node [midway,above,align=center] {Combinat.\\argument};
            \draw [double equal sign distance,-Implies,black!30] (model) -- (1.5,0.1);

            \draw [densely dotted,thick] (0.83,-0.5) rectangle (2.17,1.1);
            \node at (1.5,1.3) {\emph{Pivot}};

        \end{scope}

        \begin{scope}[yshift=-2.5cm]
            \node [anchor=west,align=left,xshift=-0.7cm] at (0,1.25) {\small\textbf{(b)\ \ Our analysis:}};

            \node [align=center,myParula07Red] (probarg) at (0,0) {Probabil.\\argument};
            \node [align=center,myParula07Red] (ppivot) at (0.9,0) {Probabil.\\pivots\\(\emph{\sltpp})};
            \node [align=center,myParula07Red] (cpivot) at (2.1,0) {Combinat.\\pivot\\(\emph{\sltcp})};
            \node [align=center] (security) at (3,0) {Security};
            \node [align=center,myParula07Red] (model) at (1.5,1) {Bounded-capacity model};

            \draw [decorate,decoration={calligraphic brace,mirror,raise=2.3em},thick] ([xshift=3pt]probarg.west) -- ([xshift=-1pt]ppivot.west) node [pos=0.5,below=2.3em,yshift=-3pt] {\cref{lem:many-pps}};

            \draw [decorate,decoration={calligraphic brace,mirror,raise=2.3em},thick] ([xshift=1pt]ppivot.west) -- ([xshift=-1pt]ppivot.east) node [pos=0.5,below=2.3em,yshift=-3pt] {\cref{def:pp}};

            \draw [decorate,decoration={calligraphic brace,mirror,raise=2.3em},thick] ([xshift=1pt]ppivot.east) -- ([xshift=-1pt]cpivot.west) node [pos=0.5,below=2.3em,yshift=-3pt] {\cref{lem:many-pps-one-cps}};

            \draw [decorate,decoration={calligraphic brace,mirror,raise=2.3em},thick] ([xshift=1pt]cpivot.west) -- ([xshift=-1pt]cpivot.east) node [pos=0.5,below=2.3em,yshift=-3pt] {\cref{def:cp}};

            \draw [decorate,decoration={calligraphic brace,mirror,raise=2.3em},thick] ([xshift=1pt]cpivot.east) -- ([xshift=-1pt]security.west) node [pos=0.5,below=2.3em,yshift=-3pt] {\cref{lem:cps-stabilize}};

            \draw [decorate,decoration={calligraphic brace,mirror,raise=2.3em},thick] ([xshift=1pt]security.west) -- ([xshift=-1pt]security.east) node [pos=0.5,below=2.3em,yshift=-3pt,align=center] {\cref{thm:safety-and-liveness-pow}};

            \draw [double equal sign distance,-Implies,myParula07Red] (probarg) -- (ppivot) node [midway,above,align=center] {Exist\\many};
            \draw [double equal sign distance,-Implies,myParula07Red] (ppivot) -- (cpivot) node [midway,above,align=center] {Many to one};
            \draw [double equal sign distance,-Implies] (cpivot) -- (security) node [midway,above,align=center] {Combinat.\\argument};
            \draw [double equal sign distance,-Implies,myParula07Red] (model) -- (1.5,0.4);

        \end{scope}
    \end{tikzpicture}%
    \vspace{-0.5em}%
    \caption[]{%
        (a)~Sleepy analysis~\cite{sleepy} is based on \emph{pivots}.
        Pivots are special \emph{honest} blocks 
        ($\Rightarrow$~liveness)
        which by a combinatorial argument
        remain in the chain forever 
        ($\Rightarrow$~safety),
        and by a probabilistic argument happen frequently.
        Equivalence of the pivot qualities required for 
        each of both arguments
        follows from 
        bounded delay~\cite[Fact 1]{DBLP:journals/iacr/BentovPS16}.
        The bounded-capacity analysis of \cite{bwlimitedposlc} also follows the same procedure by choosing a large enough delay parameter.
        (b)~We (\textcolor{myParula07Red}{red}) decompose
        pivots' probabilistic vs.\ combinatorial qualities
        into \emph{\sltpps} vs.\ \emph{\sltcps}.
        These are no longer equivalent under bounded capacity,
        but of \emph{many} consecutive \sltpps \emph{one} is a \sltcp (new combinatorial argument),
        and \sltpps are abundant (new probabilistic argument).
        \ifshortVersion%
        All references except for \cref{thm:safety-and-liveness-pow} are to the full version~\cite{full-version}.%
        \fi%
    }%
    \label{fig:analysis-comparison-sleepy}%
\end{figure}%

Our key contribution is a new analysis technique for PoW NC
under bounded capacity.
Traditional NC security analysis (\cref{fig:analysis-comparison-sleepy}(a)) is based on the notion of a \emph{pivot}~\cite{sleepy}.
Pivots are special \emph{honest} blocks 
($\Rightarrow$~liveness)
which by a combinatorial argument
remain in the chain forever 
($\Rightarrow$~safety),
and by a probabilistic argument happen frequently.
Safety and liveness of NC with suitable parameters follow swiftly.

Under bounded delay,
the qualities required for the probabilistic and combinatorial argument, respectively,
are equivalent. 
As a result, it has not been 
widely
noted that these properties are 
not identical.
Under bounded capacity,
these properties are no longer equivalent.
Our \textbf{first} conceptual \textbf{contribution} is to
decompose pivots' probabilistic/combinatorial qualities into \emph{\sltpps} and \emph{\sltcps} (\cref{fig:analysis-comparison-sleepy}(b)).
\sltPps are honest block production events where in every time interval around them there are more
honest than adversary block production opportunities (same as pivots in the bounded-delay analysis).
\sltCps are honest block production events where in every time interval around them there are more \emph{chain growth} events than non-chain-growth events (chain growth occurs
only when an honest block is produced \emph{and soon processed} by
honest nodes).

Some \sltpps no longer turn into \sltcps under bounded capacity, because adversary block
release can delay the processing of honestly produced blocks, and thus some honest block production opportunities might not translate to chain growth.
Previous bounded-capacity analysis~\cite{bwlimitedposlc} side-stepped this difference 
by choosing a specific scheduling policy 
and such a low block production rate 
that every \sltpp becomes a \sltcp.
Instead,
our \textbf{second} technical \textbf{contribution} is a combinatorial argument to show that if there is a sufficiently \emph{high density} of \sltpps over a long time interval, then at least one of these \sltpps is typically a \sltcp.
This relies on the adversary's limited budget of blocks it can spam with, and holds for a wide range of scheduling policies (including longest-header-chain and freshest-block~\cite{bwlimitedposlc}).

The original probabilistic argument of~\cite{sleepy}
guarantees only a fairly \emph{low density} of \sltpps.
Proving a high density is challenging because the occurrence of \sltpps are dependent events, so standard Chernoff-style tail bounds are not enough.
Our \textbf{third} technical \textbf{contribution} is to show, by leveraging the weak dependence of \sltpp occurrences, that long time intervals typically have a \emph{high density} of \sltpps.
This completes the analysis for PoW NC.

\subsubsection{\ProtLong}
\label{sec:introduction-methods-pos}

In \ProtShort, every honest node processes at most one out of several equivocations, and instead considers equivocating blocks to be \emph{blank}.
This makes honest nodes immune to the effects of equivocation spamming.
However, we need to ensure that honest nodes can still switch from one chain to another longer chain, both of which might contain different equivocating blocks.
For this, note that headers of two equivocating blocks from the same 
PoS lottery
can serve as a \emph{succinct equivocation proof} to convince other nodes that an equivocation was committed.
Therefore, in \ProtShort, if an honest node sees an equivocation for a block in its longest chain, it publishes an equivocation proof in the block that it produces, which allows all nodes to consistently treat the equivocating block's content as \emph{blank} without processing it.
A caveat so far is that an adversary could reveal an equivocation late and cause inconsistent ledgers across honest nodes and/or time.
To avoid this, we enforce a deadline for how late an equivocation proof can be included in the chain.
Our security proof shows how to parameterize the 
deadline
and the protocol's confirmation depth such that if any honest node has blanked the content of any equivocating block on its longest chain, then an appropriate equivocation proof is timely included on-chain, and all honest nodes blank the block's content before it reaches the output ledger.

\subsubsection{Ensuring Fees Get Paid despite Lack of Predictable Validity}
\label{sec:introduction-methods-throughputloss}

Blanking of blocks in \ProtShort leads to \emph{lack of predictable transaction validity},
\ie, honest nodes do not know whether transactions they include in their block will be valid, since the content of blocks in the prefix may later be blanked due to an equivocation. 
Many modern consensus protocols in which consensus proceeds without executing transactions~\cite{spiegelman2022bullshark,danezis2022narwhal,dispersedledger,honeybadger,al2019lazyledger} also lack predictable transaction validity.
This risks that the adversary gets to spam the ledger with invalid transactions for free.
In one solution to prevent this, we focus on guaranteeing \textit{transaction fees} are always paid regardless of equivocations, by introducing 
 \emph{gas deposit accounts} that can only be used to pay transaction fees.
Any deposit to such an account takes effect only after the deadline has passed for the inclusion of any equivocation proof that might lead to removal of transactions from the deposit's prefix.
This gives honest block producers a lower bound on the account's balance 
which they can use to reliably determine whether a transaction can pay fees.

\section{The \TeaserAttack}
    \label{sec:experiments}

\begin{figure}[tb]
    \centering
    \begin{tikzpicture}[x=1cm,y=0.75cm]
        \scriptsize
        
        \draw [draw=myParula05Green,ultra thick,line width=5pt] (0,0) -- (-1,-0.5) -- (-1,-1.5);
        \draw [draw=myParula05Green,ultra thick,line width=5pt] (-1,-0.7) -- ++(-0.4,-0.3);
        \draw [draw=myParula05Green,ultra thick,line width=5pt] (-1,-1) -- ++(0.4,-0.3);
        \draw [draw=myParula07Red,ultra thick,line width=5pt] (0,0) -- (+1,-0.5) -- (+1,-1.5);
        \draw [dotted,line width=2pt] (0,0.7) -- (0,0) node [circle,fill=black,draw=none,inner sep=1.5pt] {};

        \draw [-latex] (-1,-1.65) -- ++(0,-0.5) node [midway,left,xshift=-2pt] {$\blkratetimeGrowth$};
        \draw [-latex] (+1,-1.65) -- ++(0,-0.5) node [midway,right,xshift=2pt] {$\blkratetimeAdv$};

        \node [align=right,anchor=east,xshift=-7mm] at (-1,-1) {Honest chain};
        \node [align=left,anchor=west,xshift=2mm] at (1,-1) {Adversary chain};
        \node (start) [align=left,anchor=west,xshift=2mm] at (1,0.2) {Longest chain when\\private attack starts};
        \draw [-latex,shorten >=4pt,bend right=20] (start.west) to ([yshift=2pt] 0,0);

        \node [draw=myParula05Green,thick,fill=white,circle,inner sep=1pt] (tx) at ($(0,0)!0.4!(-1,-0.5)$) {$\tx$};

        \draw [latex-,shorten <=2pt,bend right=30] (tx) to ([xshift=-7mm] -1,-0.2) node [align=right,anchor=east] {Transaction to deconfirm};

    \end{tikzpicture}%
    \vspace{-0.5em}%
    \caption{%
        \textbf{Private attack (recap):}
        Based on the tip of the longest chain
        when the private attack starts,
        the adversary mines
        a \emph{private} \textcolor{myParula07Red}{adversary chain},
        while honest nodes jointly grow a \emph{public} \textcolor{myParula05Green}{honest chain}.
        The adversary's goal is to deconfirm
        a transaction $\tx$ included on the honest chain just below
        where the adversary chain forks off.
        Adversary mining is perfectly coordinated
        so that the adversary chain grows at the
        adversary block production rate $\blkratetimeAdv$.
        Honest nodes suffer from forking due to network delay
        so that the honest chain grows at
        a lower rate $\blkratetimeGrowth < \blkratetimeHon$
        than the total block production rate $\blkratetimeHon$
        of honest nodes.
        The attack succeeds if the adversary chain grows faster
        than the honest chain ($\blkratetimeAdv > \blkratetimeGrowth$) and thus,
        irrespective of the confirmation depth $\confDepth$ chosen for NC,
        the adversary
        chain can eventually displace the honest chain 
        as the longest chain and with that deconfirm $\tx$.
    }
    \label{fig:attack-race-concept}
\end{figure}

We begin by exploring a strategy that the attacker can adopt which forces honest nodes to waste capacity on blocks that do not contribute to chain growth. This strategy demonstrates that the well-studied \emph{private attack} \cite{nakamoto_paper,dem20} is not the worst case behavior of the attacker, and that the previously established security bounds of the bounded-delay model do not hold in the bounded-capacity setting. 
We go on to simulate our \teaserattack and to show how it compares to the private attack (summarized in \cref{fig:attack-race-concept}%
). 
Previous analyses concluded that the private attack (\cref{fig:attack-race-concept}) is worst-case
based on the false assumption that delays, and hence the honest chain growth rate, do not depend on whether the adversary releases blocks and causes congestion. We exploit congestion to develop the \teaserattack.

\import{./figures/}{fig-attack-teaser.tex}

\begin{figure}[tb]%
    \centering%
    \begin{tikzpicture}[]
        \footnotesize
        \begin{axis}[
                mysimpleplot,
                xlabel={Capacity: $\bwtime$ [blocks per second]},
                ylabel={Honest chain growth\\rate ($\blkratetimeGrowth/\blkratetimeHon$)},
                xmin=0, xmax=2,
                ymin=0, ymax=0.7,
                height=0.5\linewidth,
                width=\linewidth,
                yticklabel style={
                        /pgf/number format/fixed,
                        /pgf/number format/precision=2
                },
                scaled y ticks=false,
                legend columns=2,
            ]

            \addplot [myparula11, %
                    only marks, mark size=1.5pt] table [x=bandwidth,y=chain_growth] {figures/fig-experiment-teaser-noattacker-data.txt};
            \addlegendentry{No attack or private attack};
            \label{plt:experiment-teaser-noattacker};

            \addplot [myparula22,thin,solid,mark size=1.5pt] table [x=inverse_delay,y=chain_growth] {figures/fig-experiment-growth-delay-data.txt};
            \addlegendentry{No/private attack under bounded-delay $\Delta = 1/\bwtime$};
            \label{plt:experiment-growth-delay};

            \addplot [myparula73, mark size=1.5pt,%
            only marks] table [x=bandwidth,y=chain_growth] {figures/fig-experiment-teaser-activeattacker-data.txt};
            \addlegendentry{\Teaserattack};
            \label{plt:experiment-teaser-attacker};

        \end{axis}
    \end{tikzpicture}%
    \vspace{-0.5em}%
    \caption{%
    Results of a simulation showing that attackers can slow the growth of the honest chain using the \teaserattack. Shown, is the rate of chain growth relative to honest block production rate, when nodes prioritize processing towards the longest header chain, for various capacity limits. When the attacker does not release any blocks (no attack or private attack), we already see $\blkratetimeGrowth < \blkratetimeHon$ due to natural congestion (\ref{plt:experiment-teaser-noattacker}).
    The honest chain growth rate under the private attack is approximately the same for a network with finite processing capacity $\bwtime$ (\ref{plt:experiment-teaser-noattacker}), or for an idealized network with bounded delay $\Delta = 1/\bwtime$ (\ref{plt:experiment-growth-delay}).
    With a \teaserattack, processing is slowed roughly by a factor of $2$, which lowers the growth rate of the chain further (\ref{plt:experiment-teaser-attacker}).
    This lowers security compared to a private attack, \cf \cref{fig:comparison-bddelay-bdbandwidth}.%
    }%
    
    \label{fig:experiment-teaser}%
\end{figure}%

\myparagraph{Description of the \TeaserAttack}
\label{sec:teaser-attack}
The key idea in the \teaserattack (summarized in \cref{fig:attack-teaser}) is that the adversary can strategically time the release of blocks it had mined in order to take up some of the capacity of honest nodes.\footnote{While similar 
to the BDoS attack of \cite{mirkin2020bdos}, we note that while they exploit miner incentives to depress honest mining, our \teaserattack exploits network and processing congestion to attack safety.}
In a nutshell, while the adversary continues to mine a private chain, every time an honest node announces a block at a new height,
the adversary releases the headers of a segment of its longer withheld chain and the contents of only the first block.
Due to \rulelc scheduling, honest nodes prioritize processing blocks on the chain announced by the adversary.
Only after an honest node has processed the first adversary block and realizes that the content for the remaining blocks in the announced adversary chain segment are unavailable, does the \rulelc rule switch back to processing the newly created honest block.
Therefore, the adversary \emph{`teased'} the honest nodes to spend some of their resources processing the adversary chain,  but without actually gaining a longer chain of blocks compared to the chain they already possessed.
The result of this strategy is delayed processing of honest blocks that extend the longest honest chain. Processing is delayed by a factor of $2$ compared to the private attack.
This in turn results in more honest blocks forking, thus slowing down the honest chain growth rate 
(\cref{fig:experiment-teaser}~\ref{plt:experiment-teaser-attacker}) to $\blkratetimeGrowthTeaser < \blkratetimeGrowthSilent$.

\myparagraph{Conditions for success of the \teaserattack}
Formally, in both the private attack and the \teaserattack, the length difference between the adversary chain and the honest chain is a random walk \cite{dem20} which increases
at the rate $\blkratetimeAdv$
and decreases at the rate $\blkratetimeGrowth$.
If $\blkratetimeAdv > \blkratetimeGrowth$, the random walk has a positive drift, so
in the long run, the adversary chain will outgrow the honest chain indefinitely and the attack succeeds.
Conversely, if $\blkratetimeAdv < \blkratetimeGrowth$, the random walk has a negative drift and the attack will eventually fail. 
Thus, $\blkratetimeGrowth$
determines the fraction $\beta$ of total mining power that the adversary needs for the attack, \ie, the attack succeeds if
\begin{IEEEeqnarray}{c}
    \label{eq:chain-growth-rate-beta}
    \beta \triangleq \frac{\blkratetimeAdv}{\blkratetimeAdv+\blkratetimeHon} > \frac{\blkratetimeGrowth}{\blkratetimeGrowth+\blkratetimeHon}.
\end{IEEEeqnarray}
\begin{figure}[tb]%
    \centering%
    \begin{tikzpicture}[]%
        \scriptsize%
        \begin{axis}[
            mysimpleplot,
            grid = both,
            xlabel={Time (seconds)},
            ylabel={Attacker lead (blocks)},
            xmin=0, xmax=500,
            ymin=0, ymax=100, 
            height=0.5\linewidth,
            width=\linewidth,
            yticklabel style={
                    /pgf/number format/fixed,
                    /pgf/number format/precision=2
            },
            scaled y ticks=false,
            legend columns=4,  %
            legend style={
                    xshift=-3mm,
                },
        ]
            
            \addlegendimage{empty legend}
            \addlegendentry{\Teaserattack:};

            \addplot [myparula11, no markers,mark=none] table [col sep=comma, x=time,y=height_delta,restrict expr to domain={\thisrow{beta}}{0.75:0.75}]
            {figures/fig-experiment-teaser-start-data3.txt};
            \addlegendentry{$\blkratetimeAdv=0.75$};
            \label{plt:teaser-strong2};

            \addplot [myparula71, no markers,mark=none] table [col sep=comma, x=time,y=height_delta,restrict expr to domain={\thisrow{beta}}{0.6:0.6}]
            {figures/fig-experiment-teaser-start-data3.txt};
            \addlegendentry{$\blkratetimeAdv=0.60$};
            \label{plt:teaser-medium2};

            \addplot [myparula51, no markers, solid,mark=none] table [col sep=comma, x=time,y=height_delta,restrict expr to domain={\thisrow{beta}}{0.45:0.45}]
            {figures/fig-experiment-teaser-start-data3.txt};
            \addlegendentry{$\blkratetimeAdv=0.45$};
            \label{plt:teaser-weak2};

            \addlegendimage{empty legend}
            \addlegendentry{Private attack:};

            \addplot [myparula11, opacity=0.3, no markers,mark=none] table [col sep=comma, x=time,y=height_delta,restrict expr to domain={\thisrow{beta}}{0.75:0.75}]
            {figures/fig-experiment-private-start-data3.txt};
            \addlegendentry{$\blkratetimeAdv=0.75$};
            \label{plt:private-strong2};

            \addplot [myparula71, opacity=0.3, no markers,mark=none] table [col sep=comma, x=time,y=height_delta,restrict expr to domain={\thisrow{beta}}{0.6:0.6}]
            {figures/fig-experiment-private-start-data3.txt};
            \addlegendentry{$\blkratetimeAdv=0.60$};
            \label{plt:private-medium2};

        \end{axis}%
    \end{tikzpicture}%
    \vspace{-0.5em}%
    \caption{%
    Adversary lead (difference between adversary and honest  chain lengths) under private attack and \teaserattack.
    The simulation consists of $100$ honest nodes with capacity $C=2$ blocks per second, collectively mining $\blkratetimeHon=1$ block per second, and one adversary node with variable mining rate $\blkratetimeAdv$ blocks per second.
    With these parameters, $\blkratetimeGrowthSilent = 0.67$ and $\blkratetimeGrowthTeaser = 0.50$ (from \cref{fig:experiment-teaser}). 
    A weak adversary ($\blkratetimeAdv = 0.45$, \ref{plt:teaser-weak2}) is unable to mine fast enough to gain a lead on the network. A stronger adversary ($\blkratetimeAdv = 0.60$) fails to gain a permanent lead through the private attack (\ref{plt:private-medium2}). But, through the \teaserattack, after repeatedly retrying during the first 200 seconds, eventually manages to maintain a lead (\ref{plt:teaser-medium2}). An even stronger adversary ($\blkratetimeAdv = 0.75$) succeeds almost at once under both strategies (\ref{plt:teaser-strong2},\ref{plt:private-strong2}).
    }%
    \label{fig:experiment-teaser-start}%
\end{figure}%

Note that the \teaserattack requires the adversary to maintain a lead of at least two blocks with respect to the honest chain at all times (to proceed in steps (a), (e) in \cref{fig:attack-teaser}). If this fails, then the adversary must give up and try the attack again.
We show a sample plot of the adversary's lead for different mining rates in \cref{fig:experiment-teaser-start}. For a large enough adversary (if $\blkratetimeAdv > \blkratetimeGrowthSilent$), it is clear the lead has a positive drift and eventually stays positive. However, the \teaserattack succeeds even when the lead has a negative drift initially (\eg for $\blkratetimeGrowthTeaser < \blkratetimeAdv < \blkratetimeGrowthSilent$), as it only needs a random lucky short burst to kickstart step (a). The resulting congestion then decreases the average growth rate of the honest chain to $\blkratetimeGrowthTeaser$, and the adversary with mining power $\blkratetimeAdv > \blkratetimeGrowthTeaser$ can positively bias the random walk, thus eventually maintaining a positive lead, and succeed. We see this process in \cref{fig:experiment-teaser-start}~\ref{plt:teaser-medium2}: the adversary's lead rises and drops to zero a few times, causing the adversary to try again. However, eventually, the adversary manages
to maintain a permanent lead. On the other hand, when $\blkratetimeAdv < \blkratetimeGrowthTeaser$, the adversary's lead has a negative drift even after the congestion effects kick in (\cref{fig:experiment-teaser-start}~\ref{plt:teaser-weak2}), and therefore the \teaserattack is bound to run out of blocks and fail.

With the combined mining rate $\blkratetime \triangleq \blkratetimeHon + \blkratetimeAdv$ of honest nodes and adversary,
and the honest chain growth rates from \cref{fig:experiment-teaser}, we use \eqref{chain-growth-rate-beta} to calculate the adversary fraction $\beta$
required for each attack
and plot it in \cref{fig:comparison-bddelay-bdbandwidth}.

\myparagraph{Simulation details}
We simulate\footnote{Source code: \gitSourceUrl} both the private attack and the \teaserattack
on a network of $100$ nodes.
Honest nodes collectively 
mine
blocks at a rate $\blkratetimeHon = 1$ block per second.
Each node has a limited processing rate of $\bwtime$ blocks per second.
Blocks consist of content (transactions) and a header
(PoW and parent block pointer).
Since the header contains all information necessary to verify the PoW,
nodes only process validly created blocks.
All honest nodes and the adversary can directly send valid block headers to one another.
Given a tree of valid block headers, nodes run the \emph{\rulelc policy}, \ie,
nodes attempt to process (download and verify) the first unprocessed block along the longest header chain.
If the longest chain is already processed, or if the content of any block on that chain is unavailable or invalid, then the rule considers the next longest header chain, and so on.
We further elaborate on the setup and other simulation details in \appendixRef{\cref{sec:attacks-details}}.

\myparagraph{Practical aspects of the \teaserattack}
The \teaserattack may not acutely break specific real-world implementations of PoW NC,
mainly because 
miners have over-provisioned capacity. 
Although the \teaserattack is specific to the \rulelc policy,
it is possible to devise attacks that exploit congestion
even for other policies
(see \appendixRef{\cref{sec:greedy-attack}}).
We also note that in basic PoS NC, the adversary can 
exacerbate the \teaserattack by 
equivocating the whole adversary chain 
every time before it releases a block.
As the attack goes on, the length of the new announced chain increases. This increases the time honest nodes spend processing this chain, and \emph{decelerates} the honest chain growth until it comes to a halt. As a result, 
the chain growth rate under the \PoSteaserattack is nearly zero (details in \appendixRef{\cref{sec:pos-teaser-attack}}).
The key takeaway from the \teaserattack is 
that exploiting congestion results in attacks that are more severe than the private attack, 
even in PoW where the block production is limited, and even when the block production rate is below the capacity of nodes. This invalidates the bounded-delay model's predictions and emphasizes the need for a security analysis under models that capture the effects of congestion, especially for protocols that aim to saturate physical performance limits.

\begin{figure}[tb]%
    \centering%
    \begin{tikzpicture}[]
        \footnotesize
        \begin{axis}[
                mysimpleplot,
                xlabel={Capacity: $\bwtime$ [blocks per second]},
                ylabel={Honest chain growth\\rate ($\blkratetimeGrowth/\blkratetimeHon$)},
                xmin=0, xmax=2,
                ymin=0, ymax=0.7,
                height=0.5\linewidth,
                width=\linewidth,
                yticklabel style={
                        /pgf/number format/fixed,
                        /pgf/number format/precision=2
                },
                scaled y ticks=false,
                legend columns=2,
            ]

            \addplot [myparula11, %
                    only marks, mark size=1.5pt] table [x=bandwidth,y=chain_growth] {figures/fig-experiment-teaser-noattacker-data.txt};
            \addlegendentry{No attack or private attack ($\blkratetimeSPV = 0$)};
            \label{plt:experiment-teaser-spv-noattacker};

            \addplot [myparula73, mark size=1.5pt,%
            only marks] table [x=bandwidth,y=chain_growth] {figures/fig-experiment-teaser-activeattacker-data.txt};
            \addlegendentry{\Teaserattack ($\blkratetimeSPV = 0$)};
            \label{plt:experiment-teaser-spv0-attacker};

            \addplot [
            myparula45, 
            mark size=1.5pt,
            only marks
            ] table [x=bandwidth,y=chain_growth] {figures/fig-experiment-teaser-spv50new-activeattacker-data.txt};
            \addlegendentry{\Teaserattack ($\blkratetimeSPV = 0.5$)};
            \label{plt:experiment-teaser-spv50new-attacker};

        \end{axis}
    \end{tikzpicture}%
    \vspace{-0.5em}%
    \caption{%
    \Teaserattack in the presence of SPV miners, compared with \teaserattack and private attack without SPV miners (same as in \cref{fig:experiment-teaser}). 
    Total mining rate of SPV miners is $\blkratetimeSPV$ blocks per second.
    Total mining rate of honest miners is 
    $\blkratetimeHon = 1$ block per second. 
    SPV miners are not counted as honest. 
    \Teaserattack still succeeds with lower adversary power than private attack
(\cref{fig:experiment-teaser-spv}).
    }%
    
    \label{fig:experiment-teaser-spv}%
\end{figure}%

\myparagraph{Effect of SPV miners}
Rational miners in PoW NC face a \emph{verifier's dilemma}~\cite{verifier-dilemma,tuxedo,demystifying-incentives}:
there is no incentive to download and verify a block's content before mining to extend it.
Some so called \emph{SPV miners}
(named after simple-payment-verification clients who download only block headers)
mine empty blocks without verifying the parent block's content first,
and thus
get more time to mine, increasing their chances of being rewarded for mining the next block.
Since SPV miners are immune to congestion (as they do not process block content), how does their presence affect the \teaserattack?
Under the \teaserattack, SPV miners would mine on the adversary’s longer header chain (red block $3$ in \cref{fig:attack-teaser}(d)) without waiting for its contents. However, the remaining honest miners (who we assume still outnumber the SPV miners) still do not consider this chain valid (due to unavailable content). They continue mining on the honest chain, and would still be slowed down by the \teaserattack just as before.
We added SPV miners to our simulation and verified that the \teaserattack still succeeds with lower adversary power than the private attack
(\cref{fig:experiment-teaser-spv}).
Thus, the qualitative insight from the \teaserattack, that congestion enables worse attacks than the private attack, persists.

\section{Protocol \& Model}
\label{sec:modelprotocol}

We briefly recap Nakamoto consensus (NC)
and the bounded-ca\-pa\-ci\-ty model of~\cite{bwlimitedposlc}.
Detailed pseudocode of the protocol is provided in \appendixRef{\cref{sec:algos-reference-pseudocode}}.
Technical details about the model are provided in \appendixRef{\cref{sec:algos-reference-environment}}.
For ease of exposition, the execution features a \emph{static} set of $N$ \emph{equipotent} \emph{nodes}, each of which runs an independent instance of the protocol.
Temporary crash faults (`sleepiness') of nodes, heterogeneous distribution of hash power,
or difficulty adjustment
are left to be addressed with techniques from~\cite{backbone,sleepy,garay2017bitcoin}.
We are interested in the large system regime $N\to\infty$.
Nodes interact with each other and with the adversary $\Adv$ through an environment $\Env$ that models the network.
$\Adv$ and $\Env$ are summarized below.

\myparagraph{Nakamoto's Longest Chain Consensus Protocol}
For ease of analysis, we consider the protocol 
(pseudocode in \appendixRef{\cref{alg:generic-lc-protocol}})
to proceed in discrete \emph{\timeslots} of duration $\slotduration$.
Consider $\slotduration$ to be a small quantum of time where $\slotduration \to 0$.
At each \timeslot $t$, the protocol 
queries the PoW block production (`mining') oracle
(idealized functionality in \appendixRef{\cref{alg:hdrtree-pow}})
in an attempt to extend the \emph{longest processed chain} $\dC$ in the node's view with a new block 
of
pending transactions $\txs$.
Each block production attempt is committed to a parent block and block content,
and only a single block is produced when the attempt is successful.
Per \timeslot, each node can make one block production attempt that will be successful with probability $\blkrateslot/N$ where $\blkrateslot = \Theta(\slotduration)$, independently of other nodes and \timeslots.
If successful, the node disseminates both the resulting \emph{(block) header} $\Chain'$ and the associated \emph{(block) content} $\txs$ via the environment $\Env$ to all nodes.
Finally, the protocol identifies the $\confDepth$-deep prefix $\dC\trunc{\confDepth}$ containing all but the last $\confDepth$ blocks of $\dC$.
The transactions along $\dC\trunc{\confDepth}$ are concatenated to produce the \emph{output ledger} $\LOG{}{t}$.

When a node $p$ receives a new valid block header $\Chain$ from $\Env$ (push-based header broadcasting), 
then $p$ adds $\Chain$ to its \emph{header tree} $\hT$ 
and relays $\Chain$ to all other nodes via $\Env$.
Throughout the execution, the protocol requests from $\Env$
(pull-based content downloading)
the content 
for
block headers 
decided by a \emph{scheduling policy}.
As a concrete example, we use the \rulelc rule
(pseudocode in \appendixRef{\cref{alg:longest-header-chain-rule}})
in which
a node downloads content for the first block header with unknown content on the longest header chain it sees.
Once a block's content is received and verified by executing its transactions, 
the node makes it available to other nodes via $\Env$, and updates its 
$\dC$.

\myparagraph{Bounded-Capacity Network}
We borrow the bounded-capacity network model of~\cite{bwlimitedposlc} (see \appendixRef{\cref{fig:model}} for an illustration).
In this model, $\Env$ abstracts \emph{push-based flooding of `small' block headers} and \emph{pull-based downloading of `large' block contents} from peers.
Broadcasted block header chains 
are delivered by $\Env$ to every node,
with a per-node per-header delay determined by $\Adv$, up to a commonly known delay upper bound $\DeltaHeader$.
Block content made available for download
is kept by $\Env$ in what can be thought of as a `cloud'.
Nodes can request the content associated with a particular header.
If content matching the header is available, then it is delivered by $\Env$ to the node.
Content download and verification is subject to a per-node capacity constraint of~$\bwtime$.
Blocks have a fixed maximum size, hence $C$ is measured in blocks per second.
See \appendixRef{\cref{sec:algos-reference-environment}} for a more formal description of $\Env$.

The `cloud' captures key properties of pull-based
peer-to-peer 
downloading. At first, content matching a particular header might not be available (\eg, $\Adv$ produced a block and disseminated its header, but withheld its content). Later, such content can become available (\eg, $\Adv$ releases the content to one node). Thus, the `cloud' ensures neither data availability nor strong consistency of query outcomes, unlike stronger primitives such as
verifiable information dispersal
\cite{avid,avidfp,dispersedledger,semiavidpr}. However, once content for a header does become available, it is unique and remains available. This captures the header's binding commitment to the content, and the fact that honest nodes share content 
with peers. Requests for unavailable content do not count towards the processing budget.

Also note that the adversary can push additional headers and contents to nodes at will.
This models non-uniform capacity (higher than the lower bound $C$)
and non-uniform delay (lower than the upper bound $\DeltaHeader$)
across nodes (analogous to adversary delay up to maximum $\Delta$ in the bounded-delay model).

\myparagraph{The Adversary}
The \emph{static} adversary $\Adv$ chooses a set of nodes (up to a fraction $\beta$ of all $N$ nodes, where $\beta$ is common knowledge) to corrupt before the randomness of the execution is drawn and the execution commences. Uncorrupted \emph{honest} nodes follow the protocol at all times. Corrupted \emph{adversary} nodes have arbitrary computationally-bounded \emph{Byzantine} behavior, coordinated by $\Adv$ in an attempt to break consensus.
Among other things, the adversary can:
withhold block headers and contents, or release them late or selectively to honest nodes;
push headers and contents to nodes while bypassing the delay and capacity constraints;
break ties in the chain selection and schdeuling policy.
Note that all miners that deviate from the honest protocol (including crash faults and SPV miners) are modeled as adversary.

\myparagraph{Security}
For an execution of 
PoW NC
where every honest node $p$ at every \timeslot $t$ outputs a ledger $\LOG{p}{t}$, we recall the security desiderata.

\begin{itemize}
    \item \emph{Safety:}
          For all adversary strategies,
          all \timeslots $t,t'$, and 
          all honest nodes $p, q$ (same or different): $\LOG{p}{t}\preceq\LOG{q}{t'}$ or $\LOG{q}{t'}\preceq\LOG{p}{t}$.
    \item \emph{$\Tlive$-Liveness:}
          For all adversary strategies, if a transaction $\tx$ is received by all honest nodes by \timeslot $t$,
          then for every honest node $p$ and for all \timeslots $t' \geq t+\Tlive$, $\tx \in \LOG{p}{t'}$.
\end{itemize}

Note that since blocks have a fixed maximum size, liveness is expected only if transactions are received at a bounded rate. The following definition captures this.

\begin{definition}
\label{def:env-bounded-tx}
    The environment $\Env$ is \emph{$(\tput,\Ttput)$-tx-limited}, if the cumulative size of all transactions received by honest nodes during any interval of $\Ttput$ \timeslots is at most $\tput\cdot\Ttput$ times the maximum block size.
\end{definition}

Liveness will be proved under transaction-limited environments. The parameter $\tput$ is thus the worst-case throughput ($\blkratetime$ being the best-case throughput).
The burstiness of transaction arrival is measured by $\Ttput$; large $\Ttput$ may increase confirmation latency $\Tlive$.

A consensus protocol is \emph{secure over time horizon $\Thorizon$ \timeslots with transaction rate $\tput$} iff for some finite $\Ttput,\Tlive$, for all $(\tput,\Ttput)$-tx-limited environments, it satisfies safety, and $\Tlive$-liveness 
with overwhelming probability\footnote{As is customary, we denote by $\kappa$ the security parameter. 
Event $\Event_{\kappa}$ occurs \emph{with overwhelming probability} if $\Prob{\Event_{\kappa}} \geq 1 - \negl(\kappa)$.
Here, a function $f(\kappa)$ is \emph{negligible} $\negl(\kappa)$, if for all $n>0$, there exists $\kappa_n^*$ such that for all $\kappa > \kappa_n^*$, $f(\kappa) < \frac{1}{\kappa^n}$.} over executions of time horizon $\Thorizon$ \timeslots.
The properties 
can also be redefined in terms of real-time units instead of \timeslots.

\section{Security Proof}
\label{sec:proof}

Due to space constraints, we
focus on the intuition for the proof.
The security theorem for PoW NC is \cref{thm:safety-and-liveness-pow}.
The detailed full proof is provided in \appendixRef{\cref{sec:fullproof}}.

\subsection{Definitions}
\label{sec:proof-definitions}

For any sequence $\{\Xat{k}\}$ and index set $I$,
let $\Xat{I} \triangleq \sum_{k \in I} X_k$.

\myparagraph{Probabilistic Model for PoW NC Executions}
A \emph{block production opportunity} (\BPO) is a pair $(p,t)$ where according
to the PoW block production lottery,
node $p$ is eligible to produce a block in \timeslot $t$.
A \BPO is \emph{honest} (resp.\ \emph{adversary}) if $p$ is honest (resp.\ adversary).
Since $N\to\infty$, and mining power is homogeneous,
honest (resp.\ adversary) \BPOs per slot are Poisson distributed
with parameter $(1-\beta)\rho$ (resp.\ $\beta\rho$).
An \emph{execution} refers to a particular realization of the block production lottery
for all \timeslots.

\myparagraph{\sltGood, \sltBad, and \sltEmpty \Timeslots}
\Timeslots without a \BPO are called \emph{`\sltempty'}.
A \timeslot is \emph{`\sltgood'} iff
it has 
exactly one honest \BPO and no adversary \BPOs,
and is followed by $\goodsep$ \sltempty \timeslots
(inspired by convergence opportunities \cite{pss16,sleepy,kiffer2018better}, loners \cite{dem20}, and laggers \cite{ren}).
Here, $\goodsep$ is an analysis parameter.
We define another analysis parameter $\goodsepbw$
which is related to $\goodsep$ as 
$(\goodsep+1)\slotduration \triangleq \DeltaHeader + \goodsepbw / \bwtime$.
Thus, $\goodsep, \goodsepbw$ are chosen so
that for a \sltgood \timeslot, every honest node can 
receive the block header for the honest \BPO, and
process content for $\goodsepbw$ blocks, before the next \BPO. 
Any non-\sltempty \timeslot which is not \sltgood is called \emph{`\sltbad'}.

We denote by $t_k$ the $k$-th non-\sltempty \timeslot.
Then, we can introduce random processes over \emph{\iindices},
with \iindex $k$ corresponding
to the $k$-th non-\sltempty \timeslot $t_k$.
Considering only \iindices simplifies notation considerably.
The process $\{\Gat{k}\}$  (`$G$' for \emph{good})
counts good \timeslots,
with $\Gat{k} \triangleq \Ind{ \predGood{t_k} }$.
Correspondingly, let $\Bat{k} \triangleq 1 - \Gat{k}$.
The following fact shows the distribution of \sltgood \iindices.
\begin{restatable}{proposition}{RestatePropXiIsIid}
    \label{prop:X_i-is-iid}
    The 
    $\{\Gat{k}\}$ are independent and identically distributed (\iid) with
        $\Prob{\Gat{k} = 1} \triangleq \probGood = \probGoodFormula$.
\end{restatable}
Throughout the analysis, we assume
$\probGood > \frac{1}{2}$ (`honest majority' assumption).

\myparagraph{Some \sltGood \Timeslots Imply Growth}
A special role is played by \sltgood \timeslots $t_k$
with the additional property that 
\emph{the block produced
at $t_k$ is `soon' processed by all honest nodes}.
Intuitively, these lead to \emph{chain growth},
the cornerstone of NC security~\cite{sleepy,dem20}.
We count these \timeslots with $\{\Dat{k}\}$ 
(`$D$' for \emph{downloaded}).
Specifically,
$\Dat{k} \triangleq 1$ if $t_k$ is good
\emph{and} the block produced at $t_k$
has been processed by all honest nodes by the end
of \timeslot $t_k + \goodsep$,
$\Dat{k} \triangleq 0$ otherwise,
and $\Nat{k} \triangleq 1 - \Dat{k}$.
Note that 
$\{\Gat{k}\}$
are \iid, and not affected by adversary action,
while 
$\{\Dat{k}\}$ \emph{do depend}
on the adversary action and are thus in particular
\emph{not} \iidPERIOD.

\myparagraph{Probabilistic and Combinatorial Pivots}
\begin{definition}
    \label{def:pp-informal}
    We call an \iindex $k$ a \emph{\sltpp} (\emph{probabilistic pivot}),
    denoted as $\predPP{k}$, iff
        $\predPP{k} \triangleq  
        (\forall \intvl{i}{j} \ni k\colon  \Gin{i}{j} > \Bin{i}{j})$.%
    \footnote{We denote intervals 
    as $\intvl{i}{j} \triangleq \{i+1,...,j\}$, with
    $\intvl{i}{j} \triangleq \emptyset$ if $j \leq i$.}
\end{definition}
\begin{definition}
    \label{def:cp-informal}
    We call an \iindex $k$ a \emph{\sltcp} (\emph{combinatorial pivot}),
    denoted as $\predCP{k}$, iff
        $\predCP{k} \triangleq 
        (\forall \intvl{i}{j} \ni k\colon  \Din{i}{j} > \Nin{i}{j})$.
\end{definition}
This definition of \sltpps and \sltcps decouples \cite[Def.~5]{sleepy} into its \emph{probabilistic} aspects~\cite[Sec.~5.6.3]{sleepy} and \emph{combinatorial} aspects~\cite[Sec.~5.6.2]{sleepy},
and casts them as conditions
on a random walk,
inspired by~\cite{dem20,close-latency-security-ling-ren}, to simplify the analysis.
The decoupling is one of the key differences from the analysis in \cite{sleepy} (see \cref{fig:analysis-comparison-sleepy}).
Note that a \sltcp is also a \sltpp because 
$\Dat{i} = 1$ implies $\Gat{i} = 1$.

\subsection{Analysis in the Probabilistic Model}
\label{sec:proof-analysis-overview}

We follow \cref{fig:analysis-comparison-sleepy}(b).
First, 
we show (\cref{sec:proof-analysis-cps-stabilize})
that blocks from \sltcps \emph{stabilize}, \ie,
they are in the longest processed chain of all nodes forever
(\cref{lem:cps-stabilize-informal}).
This is useful because
\emph{if we know that \sltcps occur frequently},
then honest nodes can confirm transactions that
must lie in 
the prefix of
a \sltcp's block (\emph{safety}),
and \sltcps' blocks (being produced by honest nodes)
bring any outstanding transactions onto chain (\emph{liveness}).
We then show that \sltcps occur frequently:
We show with a new probabilistic argument
(\cref{sec:proof-analysis-many-pps})
that \emph{\sltpps are abundant}, \ie,
in every `sufficiently long' interval (\ie, of length $\Omega(\kappa^2)$),
a constant fraction of the \timeslots are \sltpps (\cref{lem:many-pps-informal}).
Due to the decoupling of \sltcps and \sltpps, the proof up to this point does not depend on the capacity constraint and the scheduling policy.
Then, we show with a new combinatorial argument
(\cref{sec:proof-analysis-many-pps-one-cps})
that \emph{the adversary cannot prevent all \sltpps from becoming \sltcps},
\ie, in every `sufficiently long' interval,
there is at least one \sltcp (\cref{lem:many-pps-one-cps-informal}).
As a result,
if honest nodes confirm transactions that are still on their longest processed chain
after `sufficiently long' time (\ie, confirmation latency $\Omega(\kappa^2)$),
then PoW NC 
is safe and live under bounded capacity (\cref{sec:proof-pow}).

\subsubsection{Combinatorial Pivots Stabilize}
\label{sec:proof-analysis-cps-stabilize}

We now show that the honest block produced in a \timeslot corresponding to a \sltcp persists in the longest processed chain of all honest nodes forever after $\goodsep$ \timeslots after it was produced.
Towards this, we first show:

\begin{proposition}[\FormalVersion{\appendixRef{\cref{prop:chain-growth}}}]
\label{prop:chain-growth-informal}
At every \iindex $k$ with $\Dat{k} = 1$,
the length of the ``shortest (across honest nodes) longest processed chain'' grows.
\end{proposition}
That is, 
\sltgood \timeslots where
all honest nodes process the produced block
are \emph{chain growth events}.
Due to this and since, by \cref{def:cp-informal}, all intervals around a \sltcp contain more \iindices with $\Dat{k}=1$ than those with $\Dat{k}=0$,
there are not enough blocks for any other chain to outnumber the chain growth events that contributed to the growth of the processed chain containing the \sltcp's block.
Thus, 
we show the following
(proven
analogously to the combinatorial argument of~\cite{sleepy}):

\begin{lemma}[\FormalVersion{\appendixRef{\cref{lem:cps-stabilize}}}]
\label{lem:cps-stabilize-informal}
Let $b^*$ be the block produced in a non-\sltempty \timeslot $t_k$ such that $\predCP{k}$. 
Then, for all \timeslots $t \geq t_k + \goodsep$: $b^*$ is in the longest processed chains of all honest nodes.
\end{lemma}

\subsubsection{Probabilistic Pivots Are Abundant}
\label{sec:proof-analysis-many-pps}

Previous analyses of NC~\cite{sleepy,dem20} show that sufficiently long intervals contain at least \emph{one \sltpp} (\cref{fig:analysis-comparison-sleepy}(a)). This was enough for the bounded-delay analysis because in the bounded-delay setting, every \sltpp is also a \sltcp~\cite[Fact 1]{DBLP:journals/iacr/BentovPS16}. 
However, in the bounded-capacity setting, not every \sltcp is a \sltpp, because not every \sltgood \timeslot results in 
growth of the longest processed chain of honest nodes (\cref{fig:analysis-comparison-sleepy}(b)).
Thus, existence of one \sltpp in every large interval is not enough to conclude existence of one \sltcp in every large interval. 
Instead,
we prove,
using a concentration bound on the number of \sltpps,
that long intervals of \iindices in fact contain
\emph{a number of \sltpps proportional
to the interval length} (\cref{lem:many-pps-informal}).
Then, in \cref{sec:proof-analysis-many-pps-one-cps}, we prove that out of those many \sltpps, at least one must also be a \sltcp, which allows us to continue with the safety and liveness proofs from~\cite{sleepy}.

\begin{figure}[tb]%
    \centering%
    \begin{tikzpicture}[x=58pt,y=0.4cm]
        \scriptsize
        \draw ([xshift=-10pt]0,0) -- (0,0);
        \draw [dashed] ([xshift=116pt]0,0) -- ([xshift=150pt]0,0);
        \draw ([xshift=208pt]0,0) -- ([xshift=218pt]0,0);
        \begin{scope}
            \draw (0,0) -- (0.35,0);
            \draw [densely dotted] (0.35,0) -- (0.75,0);
            \draw (0.75,0) -- (1,0);
            \node [circle,draw=none,fill=myParula01Blue,inner sep=2pt] at (0.05,0) {};
            \node [circle,draw=none,fill=myParula02Orange,inner sep=2pt] at (0.15,0) {};
            \node [circle,draw=none,fill=myParula03Yellow,inner sep=2pt] at (0.25,0) {};
            \node [circle,draw=none,fill=myParula04Purple,inner sep=2pt] at (0.35,0) {};
            \node [circle,draw=none,fill=myParula05Green,inner sep=2pt] at (0.75,0) {};
            \node [circle,draw=none,fill=myParula06LightBlue,inner sep=2pt] at (0.85,0) {};
            \node [circle,draw=none,fill=myParula07Red,inner sep=2pt] at (0.95,0) {};
        \end{scope}
        \begin{scope}[xshift=58pt]
            \draw (0,0) -- (0.35,0);
            \draw [densely dotted] (0.35,0) -- (0.75,0);
            \draw (0.75,0) -- (1,0);
            \node [circle,draw=none,fill=myParula01Blue,inner sep=2pt] at (0.05,0) {};
            \node [circle,draw=none,fill=myParula02Orange,inner sep=2pt] at (0.15,0) {};
            \node [circle,draw=none,fill=myParula03Yellow,inner sep=2pt] at (0.25,0) {};
            \node [circle,draw=none,fill=myParula04Purple,inner sep=2pt] at (0.35,0) {};
            \node [circle,draw=none,fill=myParula05Green,inner sep=2pt] at (0.75,0) {};
            \node [circle,draw=none,fill=myParula06LightBlue,inner sep=2pt] at (0.85,0) {};
            \node [circle,draw=none,fill=myParula07Red,inner sep=2pt] at (0.95,0) {};
        \end{scope}
        \begin{scope}[xshift=150pt]
            \draw (0,0) -- (0.35,0);
            \draw [densely dotted] (0.35,0) -- (0.75,0);
            \draw (0.75,0) -- (1,0);
            \node [circle,draw=none,fill=myParula01Blue,inner sep=2pt] at (0.05,0) {};
            \node [circle,draw=none,fill=myParula02Orange,inner sep=2pt] at (0.15,0) {};
            \node [circle,draw=none,fill=myParula03Yellow,inner sep=2pt] at (0.25,0) {};
            \node [circle,draw=none,fill=myParula04Purple,inner sep=2pt] at (0.35,0) {};
            \node [circle,draw=none,fill=myParula05Green,inner sep=2pt] at (0.75,0) {};
            \node [circle,draw=none,fill=myParula06LightBlue,inner sep=2pt] at (0.85,0) {};
            \node [circle,draw=none,fill=myParula07Red,inner sep=2pt] at (0.95,0) {};
        \end{scope}
        \draw ([yshift=3pt]0,0) -- (0,2);
        \draw ([yshift=3pt]1,0) -- (1,1);
        \draw ([yshift=3pt]2,0) -- (2,1);
        \draw ([yshift=3pt,xshift=34pt]2,0) -- ([xshift=34pt]2,1);
        \draw ([xshift=34pt,yshift=3pt]3,0) -- ([xshift=34pt]3,2);
        \draw [latex-latex] ([yshift=-3pt]0,1) -- ([yshift=-3pt]1,1) node [midway,fill=white] {$2 K_1$};
        \draw [latex-latex] ([yshift=-3pt]1,1) -- ([yshift=-3pt]2,1) node [midway,fill=white] {$2 K_1$};
        \draw [latex-latex] ([yshift=-3pt,xshift=34pt]2,1) -- ([yshift=-3pt,xshift=34pt]3,1) node [midway,fill=white] {$2 K_1$};
        \draw [latex-latex] ([yshift=-3pt]0,2) -- ([xshift=34pt,yshift=-3pt]3,2) node [midway,fill=white] {$2 K_1 K_2$};
    \end{tikzpicture}
    \caption{%
        An illustration for the proof of abundance of \sltpps (\appendixRef{\cref{prop:lower-tailbound-ppivots}}). Given a long interval of size $2K_1K_2$, we partition it into $K_2$ intervals of size $2K_1$ each, and we group the \iindices as indicated by different colors. All \iindices of the same color are at least $2K_1$ apart, so that intervals of size at most $K_1$ surrounding two \iindices from the same group are disjoint, and hence the corresponding \sltpp conditions are independent (conditioned on the fact that the \sltpp condition holds for all long intervals).
    }%
    \label{fig:ppivot-tailbound-illustration}%
\end{figure}%

The key challenge in proving that there are many \sltpps is that for two \iindices $k_1, k_2$, the events that $k_1$ is a \sltpp and that $k_2$ is a \sltpp are dependent, because both events depend on overlapping intervals. But a key observation is that since the \sltpp condition
(\cref{def:pp-informal})
already holds for large intervals with high probability (\appendixRef{\cref{prop:lower-tailbound-X}}), we only need to look at the small intervals. Then, for two \iindices $k_1,k_2$ that are sufficiently far apart, these short intervals are disjoint, and thus the corresponding \sltpp conditions are independent. Therefore, we decompose a long interval of \iindices into several groups of far-apart \iindices. This is illustrated in \cref{fig:ppivot-tailbound-illustration}, each group indicated by a different color. Within each group, by a concentration bound for \iid random variables, there are many \sltpps. Further, by a union bound, the concentration holds in all the groups simultaneously with high probability (\appendixRef{\cref{prop:lower-tailbound-ppivots}}). Using this, we show:

\begin{lemma}[\FormalVersion{\appendixRef{\cref{lem:many-pps}}}]
    \label{lem:many-pps-informal}
    For $\Kcp = \Omega(\kappa^2)$,
    with overwhelming probability,
    in every interval of size at least $\Kcp$,
    at least $(1-\delta)\probPP$ fraction of the \iindices in the interval are \sltpps.
\end{lemma}

\subsubsection{Many Probabilistic Pivots Imply One Combinatorial Pivot}
\label{sec:proof-analysis-many-pps-one-cps}

The \rulelc rule $\dlrulelong$
(\appendixRef{\cref{alg:longest-header-chain-rule}})
has a few useful properties.
Nodes using this rule 
\begin{enumerate}[(P1),leftmargin=3em]
    \item \label{item:good-download-rule-no-repeat}
    process a \BPO's block's content at most once,
    \item \label{item:good-download-rule-honest-block} either process the most recent honest block, or fully utilize their capacity to process other blocks (\ie, do not stay idle), and
    \item \label{item:good-download-rule-cutoff} prioritize blocks that were produced `recently'.
\end{enumerate}
\ref{item:good-download-rule-no-repeat} holds by construction.
\ref{item:good-download-rule-honest-block} holds because 
this rule
is never idle, and 
will always process towards an honest block
when it has processed all longer chains
and there is capacity remaining.
Moreover, we expect that in a secure execution,
\ref{item:good-download-rule-cutoff} holds because
a node's longest header chain cannot fork off too much from its longest processed chain.
More precisely, due to \cref{lem:cps-stabilize-informal},
any longest header chain in any honest node's view must extend the block produced in the most recent \sltcp, and therefore blocks with the highest process priority must have been produced after the most recent \sltcp.
Thus, if the adversary wants to prevent honest nodes from processing the block produced at a \sltgood \iindex $k$,
so that $G_k = 1$ but $D_k = 0$, 
then it can only ``distract'' them by providing $\goodsepbw$ blocks
produced after the most recent \sltcp
(\cref{prop:download-or-spend-budget}).

While we subsequently use 
$\dlrulelong$
as a concrete example,
the proofs only use \ref{item:good-download-rule-no-repeat}, \ref{item:good-download-rule-honest-block}, \ref{item:good-download-rule-cutoff},
and thus apply to several other simple scheduling policies, including the freshest-block rule 
of~\cite{bwlimitedposlc}.

\begin{restatable}{proposition}{RestatePropDownloadOrSpendBudget}
    \label{prop:download-or-spend-budget}
    If $G_k = 1$ and $D_k = 0$, then during \timeslots $[t_k, t_k + \goodsep]$, all honest nodes 
    using the \rulelc scheduling policy
    process content of at least $\goodsepbw$ blocks that are produced in $\intvl{i}{k}$, where $i<k$ is the largest index such that $\predCP{i}$ (if such an $i$ does not exist, $i=0$).
\end{restatable}

\begin{proof}
In \timeslot $t_k$, there is exactly one block $b$ produced by an honest node, 
the block header is made public at the beginning of the \timeslot,
and is seen by all honest nodes within $\DeltaHeader$ time.
Thereafter, each node has enough time to process $\goodsepbw$ blocks during \timeslots $[t_k, t_k + \goodsep]$.

Under the \rulelc scheduling policy,
if $\Dat{k} = 0$,
\ie an honest node did not process content for the block $b$ before the end of \timeslot $t_k + \goodsep$,
then
that honest node must process the content for at least $\goodsepbw$ blocks on chains longer than the height of the block $b$ or in the prefix of the block $b$.
Since honest nodes produce blocks extending their longest chain, $b$ extends the longest processed chain of some honest node at \timeslot $t_k - 1$.
Let $b^*$ be the block produced in \timeslot $t_i$ where $\predCP{i}$ (suppose $i$ exists).
$\predCP{i} \implies \Yat{i} = 1$, therefore this block is unique, and also $t_k > t_i + \goodsep$.
Due to \cref{lem:cps-stabilize-informal}, any valid header chain longer than $b$ (which is some node's longest processed chain) at time slot $t_k$ must contain $b^*$.
Therefore, the only blocks
that are processed by an honest node during \timeslots $[t_k, t_k + \goodsep]$
\begin{enumerate}
    \item must be produced after $t_i$ because they extend $b^*$, and
    \item must be produced no later than $t_k$ because there are no blocks produced in $\intvl{t_k}{t_k+\goodsep}$.
\end{enumerate}
In case a \sltcp $i<k$ does not exist, the claim is trivial.
\end{proof}

\begin{figure}[tb]
    \centering
    \begin{tikzpicture}[x=1.5em]
        \scriptsize
    
        \begin{scope}[]
            
            \node at (-2,0) {\textbf{(a)}};
            
            \node at (0,0) {$-$};
            \node at (1,0) {$+$};
            \node at (2,0) {$+$};
            \node at (3,0) {$+$};
            \node at (4,0) {$-$};
            \node at (5,0) {$+$};
            \node at (6,0) {$+$};
            \node at (7,0) {$+$};
            \node at (8,0) {$-$};
            \node at (9,0) {$+$};
            \node at (10,0) {$+$};
            \node at (11,0) {$+$};
            \node at (12,0) {$-$};

            \draw [densely dotted] (-0.5,0) ++ (0,-0.2) -- ++(0,0.4) node [above] {$0$};
            \draw [densely dotted] (12.5,0) ++ (0,-0.2) -- ++(0,0.4) node [above] {$\Kcp$};

            \node [circle,draw=myParula01Blue,ultra thick,inner sep=4pt] at (2,0) {};
            \node [circle,draw=myParula01Blue,ultra thick,inner sep=4pt] at (6,0) {};
            \node [circle,draw=myParula01Blue,ultra thick,inner sep=4pt] at (10,0) {};
        
        \end{scope}
    
        \begin{scope}[yshift=-2em]
            
            \node at (-2,0) {\textbf{(b)}};
            
            \node at (0,0) {$-$};
            \node at (1,0) {$+$};
            \node at (2,0) {$+$};
            \node at (3,0) {$+$};
            \node at (4,0) {$-$};
            \node at (5,0) {$+$};
            \node at (6,0) {$+$};
            \node at (7,0) {$+$};
            \node at (8,0) {$-$};
            \node at (9,0) {$+$};
            \node at (10,0) {$+$};
            \node at (11,0) {$+$};
            \node at (12,0) {$-$};

            \draw [densely dotted] (-0.5,0) ++ (0,-0.2) -- ++(0,0.4);
            \draw [densely dotted] (12.5,0) ++ (0,-0.2) -- ++(0,0.4);

            \node [circle,draw=myParula01Blue,ultra thick,inner sep=4pt] at (2,0) {};
            \node [circle,draw=myParula01Blue,ultra thick,inner sep=4pt] at (6,0) {};
            \node [circle,draw=myParula01Blue,ultra thick,inner sep=4pt] at (10,0) {};

            \draw [draw=myParula07Red,ultra thick] (5,0) ++(-0.3,-0.2) -- ++(0.6,0.4);
            \draw [draw=myParula07Red,ultra thick] (5,0) ++(-0.3,0.2) -- ++(0.6,-0.4);

            \draw [draw=myParula07Red,dashed,rounded corners] (1.4,-0.25) rectangle (5.4,0.25);
            \draw [draw=myParula07Red,dashed,rounded corners] (4.6,-0.35) rectangle (6.6,0.35);
        
        \end{scope}
        
    \end{tikzpicture}%
    \caption[]{%
    (a) Example realization of \BPOs
    in $\intvl{0}{\Kcp}$,
    with \sltgood ($+$) and \sltbad ($-$) \iindices,
    and resulting \sltpps
    (\tikz[x=0.75em,y=0.75em]{ \node [circle,draw=myParula01Blue,thick,inner sep=1.8pt] at (0,0) {}; }).
    (b) To prevent a \sltpp from being a \sltcp,
    the adversary needs to prevent timely processing
    (\tikz[x=0.6em,y=0.6em]{ \draw [draw=myParula07Red,thick] (0,0) ++(-0.5,-0.5) -- ++(1,1); \draw [draw=myParula07Red,thick] (0,0) ++(0.5,-0.5) -- ++(-1,1); })
    of some blocks produced at \sltgood \iindices,
    so that around the respective \sltpp
    there is an interval 
    (\tikz[x=0.7em,y=0.7em]{ \draw [draw=myParula07Red,dashed] (-0.5,-0.5) rectangle (0.5,0.5); })
    in which the \sltcp condition (\cf \cref{def:cp-informal}) is violated.
    Here, the first two \sltpps are not \sltcps.
    To prevent timely processing of a block from a \sltgood
    \iindex, the adversary must `spend' $\goodsepbw$ blocks.
    Once the adversary runs out of blocks, a \sltpp remains a \sltcp
    (here the third \sltpp).
    }
    \label{fig:ppivot-cpivot-intuitive}
\end{figure}

Given the above properties of the scheduling policy, we now want to show that \sltcps occur once in a while.
\Cref{fig:ppivot-cpivot-intuitive} illustrates the key argument for this.
To start, let us show that there is at least one \sltcp in $\intvl{0}{\Kcp}$.
From \cref{lem:many-pps-informal}, there are many \sltpps in $\intvl{0}{\Kcp}$.
If there were no \sltcps in $\intvl{0}{\Kcp}$, then the adversary must prevent each \sltpp from turning into a \sltcp.
We know that in any interval around a \sltpp, 
\sltgood \iindices outnumber \sltbad \iindices by a margin 
proportional to the interval size
(\appendixRef{\cref{prop:lower-tailbound-X}}, see top row in \cref{fig:ppivot-cpivot-intuitive}).
Therefore, for a \sltpp to not be a \sltcp, the adversary must prevent an honest node from processing the most recent honest block in several of these \sltgood \iindices (so that the corresponding $\Gat{k}=1$ \iindices have $\Dat{k}=0$).
\Cref{fig:ppivot-cpivot-intuitive} shows an example where the adversary prevented processing of the honest block in one \sltgood \iindex, and as a result, two of the \sltpps fail to become a \sltcp.
From \cref{prop:download-or-spend-budget}, for each such \iindex, the adversary must `spend' at least $\goodsepbw$ blocks that the honest node processes.
These blocks come from a `limited budget'. 
In \appendixRef{\cref{lem:one-cp-induction-base}}, through a combinatorial argument, we show that
this `budget' falls short of the number of blocks required to overthrow all \sltcps. Thus, there must be at least one \sltcp in $\intvl{0}{\Kcp}$.
Next, we would like to show that there is at least one \sltcp in $\intvl{m\Kcp}{(m+1)\Kcp}$ for all $m \geq 0$ (by induction, where 
we just saw the base case $m=0$).
Here,
one may be concerned that
the adversary could save up many blocks from the past and attempt to make honest nodes process these blocks at a particular target \timeslot $t_k$.
But, given that one \sltcp occurred in $\intvl{(m-1)\Kcp}{m\Kcp}$ (by induction hypothesis), \cref{prop:download-or-spend-budget} ensures that honest nodes will only process blocks that are produced after $(m-1)\Kcp$.
This allows us to bound the `budget' of blocks that the adversary can use to prevent \sltpps from becoming \sltcps, and 
we can complete the induction and conclude:
\begin{lemma}[\FormalVersion{\appendixRef{\cref{lem:many-pps-one-cps}}}]
\label{lem:many-pps-one-cps-informal}
    If honest nodes use the \rulelc scheduling policy,
    and in every interval of size at least $\Kcp$,
    at least a certain fraction of \BPOs are \sltpps
    (which holds for 
    $\blkrateslot, \blkratetime$ chosen as a function of the
    model and analysis parameters, as per \eqref{pow-max-tp}),
    then for all $m \geq 0$,
    the interval $\intvl{m\Kcp}{(m+1)\Kcp}$ has at least one \sltcp.
    It follows that any arbitrary interval of length $2\Kcp$ contains at least one \sltcp.
\end{lemma}

\subsection{Security of Proof-of-Work Nakamoto Consensus}
\label{sec:proof-pow}

From \cref{lem:many-pps-informal,lem:many-pps-one-cps-informal}, we conclude that for suitable $\blkratetime$,
with overwhelming probability, \sltcps occur
in every $2\Kcp$-interval.
This allows us, together with \cref{lem:cps-stabilize-informal} (\sltcps stabilize), to prove safety and liveness of the protocol for a confirmation depth $\confDepth=\Theta(\Kcp)$.
The key arguments are in the proof of the following lemma.

\begin{lemma}
    \label{lem:safety-and-liveness-comb-pow}
    If for some $\Kcp > 0$,
    \begin{IEEEeqnarray}{C}
        \label{eq:condition-one-cp-m-pow-safety}
        \forall k \colon \exists k^* \in \intvl{k}{k+2\Kcp} \colon \predCP{k^*},
        \IEEEeqnarraynumspace
    \end{IEEEeqnarray}
    then the PoW Nakamoto consensus protocol $\protocol$ with $\confDepth = 2\Kcp + 1$ satisfies safety.
    Further, if the environment is $((\frac{1}{2}-\beta)\blkrateslot, 2\Kcp/\blkrateslot)$-tx-limited, 
    then $\protocol$ also satisfies 
    liveness with 
    $\Tlive = \Theta(\Kcp)$.
\end{lemma}
\begin{proof}
    \emph{Safety:}
    Denote the longest processed chain of node $p$ at \timeslot $t$ as $\dC_p(t)$ and its $\confDepth$-deep prefix as $\dC_p(t)\trunc{\confDepth}$.
    For an arbitrary \timeslot $t$, let $k$ be the largest \iindex such that $t_k \leq t$.
    From \cref{lem:many-pps-one-cps-informal}, every interval of $2\Kcp$ \iindices contains at least one \sltcp. Therefore, there exists $k^* \in \intvl{k-2\Kcp-1}{k-1}$ such that $\predCP{k^*}$.
    Let $b^*$ be the block from \iindex $k^*$.
    Due to \cref{lem:cps-stabilize-informal}, for all honest nodes $p,q$ and $t' \geq t$,
    $b^* \in \dC_p(t)$ and $b^* \in \dC_q(t')$.
    But $k^* \geq k-\confDepth$, so the block $b^*$ cannot be $\confDepth$-deep in any chain at \timeslot $t$. Therefore, $\LOG{p}{t}$ is a prefix of $b^*$ which in turn is a prefix of $\dC_q(t')$.
    We can thus conclude that
    either 
    $\LOG{p}{t} \preceq \LOG{q}{t'}$ or $\LOG{q}{t'} \preceq \LOG{p}{t}$.
    Therefore, safety holds.

    \emph{Liveness (proof sketch, details in \appendixRef{\cref{sec:fullproof-pow}}):}
    Recall that since \iindices count \timeslots with block production, $T$ \timeslots corresponds to roughly $\blkrateslot T$ \iindices.
    Again let $k$ be the largest \iindex such that $t_k \leq t$.
    We will first prove that all transactions received between \iindices $k-2\Kcp$ and $k$, which are of total size at most $(1-2\beta)\Kcp$ as per the tx-limited environment, will be added to the longest processed chains of all nodes by \iindex $k + 2\Kcp$.
    We know that there exists $k^* \in (k,k+2\Kcp]$ such that $\predCP{k^*}$.
    Since $k^*$ is a \sltcp, there are more \iindices $j$ with $\Dat{j}=1$ than \iindices with $\Dat{j} = 0$ in the interval $\intvl{k}{k + 2\Kcp}$ (by \cref{def:cp-informal}).
    Since each \iindex with $\Dat{j}=1$ leads to chain growth, every honest node's longest processed chain grows by at least $\Kcp$ between \iindices $k$ and $k + 2\Kcp$.
    There are at most $\beta \cdot 2\Kcp$ adversary block productions in the interval $\intvl{k}{k+2\Kcp}$, hence every honest node's longest processed chain grows by at least $\Kcp - 2\beta\Kcp$ \emph{honest blocks}.
    These honest blocks will include pending transactions, whose size is at most $(1-2\beta)\Kcp$.
    Moreover, in the interval $\intvl{k+2\Kcp}{k+2\Kcp+2\confDepth}$, every honest node's longest processed chain grows by at least $\confDepth$.
    Thus, the newly added transactions are $\confDepth$-deep, hence confirmed, by all nodes by \iindex $k + 2\Kcp + 2\confDepth$, which is a latency of $\Tlive = \frac{6\Kcp+2}{\blkrateslot}$ \timeslots.   
\end{proof}

Subsequently, we take $\slotduration \to 0$ and $\blkratetime \triangleq \blkrateslot/\slotduration$ in order to model PoW accurately. 
Finally, since $\goodsepbw, \goodsep$ were analysis parameters chosen arbitrarily, we maximize over these parameters to find the best possible security--performance tradeoff (\cref{thm:safety-and-liveness-pow}).
The result is plotted for $\DeltaHeader \approx 0$ (reasonable approximation for large block content relative to headers) in \cref{fig:comparison-bddelay-bdbandwidth}.

\begin{restatable}{theorem}{RestateThmSafetyAndLivenessPow}
\label{thm:safety-and-liveness-pow}
For all $\beta < 1/2$,
$\blkratetime > 0$,
such that
\begin{IEEEeqnarray}{C}
\label{eq:pow-max-tp}
    \blkratetime < \max_{\goodsepbw} \frac{1}{\DeltaHeader + \goodsepbw/\bwtime} \ln\left( \frac{2(1-\beta)\goodsepbw}{\goodsepbw+4 + \sqrt{8\goodsepbw+16}} \right),
    \IEEEeqnarraynumspace
\end{IEEEeqnarray}
the PoW Nakamoto consensus protocol 
with 
the \rulelc scheduling policy,
$\slotduration \to 0$,
$\blkrateslot = \blkratetime \slotduration$,
and
$\confDepth = \Theta(\kappa^2)$
is secure
with 
transaction rate $(\frac{1}{2}-\beta)\blkratetime$,
confirmation latency
$\Theta(\kappa^2)$ 
over a time horizon of 
$\Thorizon = \poly(\kappa)$.
\end{restatable}

\Cref{thm:safety-and-liveness-pow} is proved in \appendixRef{\cref{sec:fullproof-pow}}.

\section{Proof-of-Stake Nakamoto Consensus}

Nakamoto consensus has been adapted to proof-of-stake in protocols of the Ouroboros~\cite{kiayias2017ouroboros,david2018ouroboros,badertscher2018ouroboros} and Sleepy Consensus~\cite{sleepy,snowwhite} families.
The protocol is identical to what was described in \cref{sec:modelprotocol} and formalized in \appendixRef{\cref{alg:generic-lc-protocol}}, except for a few key differences.
The block production oracle for proof-of-stake (idealized in \appendixRef{\cref{alg:hdrtree-pos}}) behaves differently.
As in PoW, each node can make one block production attempt per \timeslot that will be successful with probability $\blkrateslot/N$, independently of other nodes and \timeslots%
\footnote{There may be multiple blocks in one \timeslot, as in
the Ouroboros~\cite{kiayias2017ouroboros,david2018ouroboros,badertscher2018ouroboros} and Sleepy Consensus~\cite{sleepy,snowwhite} protocols.}%
, modeling uniform stake.
In PoS, however, (even past) block production opportunities can be `reused' to produce multiple blocks with different parents and/or content, \ie, to equivocate. 

\subsection{\ProtLong (\ProtShort)}
\label{sec:sapos}

In the classic bounded-delay analysis, the tradeoff between 
$\beta$ and $\lambda$
is the same for 
PoW and PoS NC~\cite{dem20,tight_bitcoin},
because, conceptually,
NC security depends only on 
a 
race between
the honest chain and adversary chains.
Even in PoS,
the adversary cannot 
use
equivocations
to boost
its chain growth rate,
because blocks within one chain must be from strictly increasing \timeslots, \ie, different \BPOs.
Under bounded capacity, however, as observed in~\cite{bwlimitedposlc}, honest nodes may waste their limited capacity processing 
equivocations
rather than staying up-to-date with the longest chain.
Thus, blocks they produce may not contribute to honest chain growth.
As a result,
the honest chain growth rate decreases, 
and with it
PoS NC security~\cite{bwlimitedposlc} (compared to PoW).
The key idea in \ProtShort is to modify the scheduling policy of PoS NC 
such that per \BPO at most one block is processed.
This restores the one assumption of the bounded-capacity PoW NC analysis (\cref{sec:proof-analysis-many-pps-one-cps}~(P1))
that was previously violated in PoS NC due to equivocations.
With the modification of \ProtShort,
the analysis from \cref{sec:proof} carries over to PoS.
One 
may consider 
this
alternative:
defer content processing until after consensus has been reached on a header chain.
This, however, requires
ensuring
that the contents belonging to headers
will be available for download.
Sampling-based approaches~\cite{DBLP:conf/fc/Al-BassamSBK21}
to check \emph{data availability}
come with various challenges~\cite{paradigmdas}
and
VID-based
approaches~\cite{dispersedledger,poar}   %
do not
scale to the large $N$ 
found in PoS NC.

\subsubsection{Protocol}
\label{sec:sapos_protocol}

\ProtShort is PoS NC (\cf\cite{kiayias2017ouroboros,sleepy}),
with the following modifications.

\myparagraph{The Scheduling Policy in \ProtShort}
\ProtShort uses
any scheduling policy that is secure for PoW NC (such as \rulelc),
modified as follows:
a node does not process content for a header
(denoted by the corresponding header chain $\Chain$)
if it has seen 
another equivocating header
from the same \BPO as $\Chain$.
Instead,
the node \emph{pretends} that content was ``processed'' and sets it to be empty.
The node can then continue processing content for headers that extend $\Chain$, and these blocks will be candidates for the node's longest processed chain.
With only the above scheduling policy,
one honest node may process
the real
content for a header while another 
may 
set it to be empty
(depending on when each node saw an equivocating header).
In order to output a consistent 
transaction ledger,
reaching consensus on the header chain is 
no longer
enough.
Instead, we ensure that honest nodes also agree on which blocks had an equivocation, through equivocation proofs, 
so that they can consistently \emph{blank} their contents.

\myparagraph{Equivocation Proofs}
An equivocation proof consists of two headers $\Chain, \Chain'$ from the same \BPO.
Whenever a node produces a new block header extending its longest processed chain,
it 
includes an equivocation proof
for any header $\Chain$ among the last $\keqproof$ headers
(on the new block's prefix)
for which it has seen an equivocating header $\Chain'$
but no equivocation proof 
was
recorded on chain yet.

\myparagraph{Equivocation Proof Deadline}
The deadline $\keqproof$ for adding equivocation proofs 
ensures
that 
the adversary cannot 
use equivocations or equivocation proofs
to make honest nodes blank the content of an old block
whose transactions they have already confirmed.
A header $\Chain$ is thus \emph{invalid}
if it contains an equivocation proof against 
a block 
that is 
not within $\keqproof$ blocks above $\Chain$.

\myparagraph{Ledger Construction in \ProtShort}
At the end of each \timeslot,
each node confirms 
all blocks on its longest processed chain that 
are $\confDepth$-deep,
except it blanks
the contents of 
blocks
against which there is an equivocation proof
on chain.

\subsubsection{Security Proof}
\label{sec:security}

The scheduling policy of \ProtShort ensures that,
just like in PoW NC,
honest nodes process at most one block per \BPO.
This eliminates additional block processing delays caused by equivocations, allowing the honest chain growth rate
to match that of PoW NC.
Given this, the security proof of PoW NC in \cref{sec:proof}
can be adapted to \ProtShort to 
show that
the $\confDepth$-deep \emph{header chains} of all nodes are consistent.

To ensure that their \emph{ledgers} are consistent, and complete the security proof, we need two more steps.
First, liveness of \ProtShort follows easily because the contents of blocks produced by honest nodes will never be blanked.
Second, for safety,
we show (a) honest nodes have processed the content for all blocks against which there is no equivocation proof on chain (these blocks must not be blanked), and (b) honest nodes blank content in their ledger consistently, that is, any honest node blanks the contents of a block in its ledger iff all honest nodes do so.
We prove (a) and (b) in \cref{thm:safety-and-liveness-pos} by choosing appropriate values for $\confDepth$ and $\keqproof$.

Since the analysis of PoW NC from 
\cref{sec:proof} (details in \appendixRef{\cref{sec:fullproof}}) applies to \ProtShort as well, \cref{lem:cps-stabilize-informal} (\sltcps stabilize the longest processed chains of all nodes) and \cref{lem:many-pps-one-cps-informal} (\sltcps recur) hold for \ProtShort.
Thus, \eqref{pow-max-tp} also determines the parameters under which \ProtShort is secure, \ie,
the security--throughput tradeoff of \ProtShort.%
\footnote{Technically, since PoS protocols run in \timeslots of fixed duration, unlike PoW, $\slotduration$ must match the \timeslot duration. If $\slotduration$ is small relative to the block production and processing times (such as $1$ second in Cardano), we can still use the approximation $\slotduration \to 0$, just like in PoW. We calculate the parameters for general $\slotduration$ in \appendixRef{\cref{sec:fullproof-analysis-many-pps-one-cps}}.}
In \cref{fig:comparison-bddelay-bdbandwidth-pos},
we plot the solutions of \eqref{pow-max-tp} with 
$\bwtime = 1$ and
$\DeltaHeader \approx 0$ (approximation for block content much larger than headers).
Since \eqref{pow-max-tp} does not depend on $\kappa$, for any given $\beta$, the block rate $\lambda$ is non-vanishing.
Only latency scales with $\kappa$, similar to PoW NC.

In both \cref{thm:safety-and-liveness-pos} 
(for \ProtShort) and \cref{thm:safety-and-liveness-pow} (for PoW NC), we prove an upper bound on the confirmation latency that scales with the security parameter $\kappa$ as $O(\kappa^2)$.
Concretely, our bound on \ProtShort's latency (\cref{thm:safety-and-liveness-pos}) is $3\times$ our bound for PoW NC (\cref{lem:safety-and-liveness-comb-pow}).

\begin{theorem}
\label{thm:safety-and-liveness-pos}
For all $\beta < 1/2$,
$\bwtime$, $\DeltaHeader$, $\blkrateslot$, $\slotduration$
satisfying \eqref{pow-max-tp},
there exists $\keqproof, \confDepth = \Theta(\kappa^2)$
such that the \emph{\ProtShort protocol
is secure}
with
transaction rate $(1 - 2\beta)\blkratetime$,
confirmation latency $\Tlive = \Theta(\kappa^2)$ \timeslots
over a time horizon of 
$\Thorizon = \poly(\kappa)$.
\end{theorem}

\begin{proof}
\begin{figure}[tb]
    \centering
    \begin{tikzpicture}[blockchainold,x=0.8cm,y=0.4cm]
        \footnotesize

        \coordinate (G) at (-0.75,0) {};

        \node [block-gray] (bi) at (0,0) {};
        \node [yshift=1.5em] (bli) at (bi) {$b$};
        \node [block-gray] (bj) at (1,0) {};
        \node [yshift=1.5em] (blj) at (bj) {$b_j$};
        \node [block-gray] (bk) at (2.5,0) {};
        \node [yshift=1.5em] (blk) at (bk) {$b_k$};
        \node [block-gray] (bl) at (4,0) {};
        \node [yshift=1.5em] (bll) at (bl) {$b_l$};

        \node [block-gray] (B1) at (5,0) {};
        
        \node [block-gray] (B21) at (5.5,-1) {};
        \node [block-gray] (B22) at (5.5,1) {};

        \node [block-gray] (B32) at (6.2,1) {};

        \node [align=left,anchor=west] (dCq) at (7,-1) {$\dC_q(t')$};
        \node [align=left,anchor=west] (dCp) at (7,1) {$\dC_p(t)$};

        \draw [->,dashed,shorten >=1em] (dCp) -- (B32);
        \draw [->,dashed,shorten >=1em] (dCq) -- (B21);

        \draw [link] (B32) -- (B22);
        \draw [link] (B22) -- (B1);

        \draw [link] (B21) -- (B1);

        \node [inner sep=0pt] (dots1) at (0.5,0) {...};
        \node [inner sep=0pt] (dots2) at (1.75,0) {......};
        \node [inner sep=0pt] (dots3) at (3.25,0) {......};
        \node [inner sep=0pt] (dots4) at (4.5,0) {...};
        \coordinate (dots4repl) at (4.5,0);

        \draw [link] (bi) -- (G);
        \draw [link] (dots1) -- (bi);
        \draw [link,-] (bj) -- (dots1);
        \draw [link] (dots2) -- (bj);
        \draw [link,-] (bk) -- (dots2);
        \draw [link] (dots3) -- (bk);
        \draw [link,-] (bl) -- (dots3);
        \draw [link] (dots4) -- (bl);
        \draw [link,-] (B1) -- (dots4);

        \coordinate (bot) at (-1,0);

        \coordinate (a) at (bi |- bot);
        \coordinate (b) at (dots3 |- bot);
        \draw [|-|] ([yshift=-1.5em]a) -- ([yshift=-1.5em]b) node [pos=0.5,below] {$\keqproof = 4\Kcp$ blocks};
        \coordinate (a) at (bi |- bot);
        \coordinate (b) at (dots4repl |- bot);
        \draw [|-|] ([yshift=-3.5em]a) -- ([yshift=-3.5em]b) node [pos=0.5,below] {$\confDepth = 6\Kcp+1$ blocks};

        \draw [<-,dashed] (bll.north) -- ([yshift=0.4cm]bll.north) node [anchor=south west,align=left,xshift=-3.5em] {Latest possible equivocation proof against $b$ stabilizes};
        \draw [<-,dashed] (blk.north) -- ([yshift=0.8cm]blk.north) node [anchor=south west,align=left,xshift=-3.5em] {Includes equivocation proof against $b$ (if any)};
        \draw [<-,dashed] (blj.north) -- ([yshift=1.2cm]blj.north) node [anchor=south west,align=left,xshift=-3.5em] {All honest nodes have processed $b$ or seen an equivocation};

    \end{tikzpicture}%
    \vspace{-1em}%
    \caption{%
        Illustration for the security proof of \ProtShort (\cref{thm:safety-and-liveness-pos}).
        Consider a block $b$ that is $\confDepth$-deep in the longest processed chain of a node.
        \Iindices $j,k,l$ are \sltcps.
        Since \sltcps stabilize (\cref{lem:cps-stabilize-informal}), the corresponding blocks $b_j, b_k, b_l$ are in all honest nodes' processed longest chains.
        At \sltcp $j$, we know for sure that all honest nodes 
        either processed $b$'s content or saw an equivocation for it, because they have processed $b_j$'s content.
        At \sltcp $k$, we know for sure that if $b$ had an equivocation, preventing processing of its content, then an equivocation proof against $b$ must have entered the chain.
        At \sltcp $l$, we know for sure that the last block that can add an equivocation proof against $b$ has stabilized (as the deadline of $\keqproof$ blocks has passed).
        Thus, a ledger formed from $\confDepth$-deep blocks (sufficient to obtain three \sltcps) will remain safe.
    }
    \label{fig:pos-safety-proof}
\end{figure}

Set $\confDepth \triangleq 6\Kcp + 1, \keqproof \triangleq 4\Kcp$.
Denote the longest processed chain of node $p$ at \timeslot $t$ as $\dC_p(t)$ and its $\confDepth$-deep prefix as $\dC_p(t)\trunc{\confDepth}$.
\emph{Safety} holds if the following three properties hold for all \timeslots $t \leq t'$ and for all honest nodes $p,q$:
(1) $\dC_p(t)\trunc{\confDepth} \preceq \dC_q(t')\trunc{\confDepth}$ or $\dC_q(t')\trunc{\confDepth} \preceq \dC_p(t)\trunc{\confDepth}$.
(2) If $b \in \dC_p(t)\trunc{\confDepth}$ and there is no equivocation proof in a block header following it, then node $p$ must have processed the content of $b$ before \timeslot $t$.
(3) 
If $b \in \dC_p(t)\trunc{\confDepth}$ and $b \in \dC_q(t')\trunc{\confDepth}$, then $p$ blanks the content of $b$ in $\LOG{p}{t}$ iff $q$ blanks it in $\LOG{q}{t'}$.

Consider an arbitrary block $b_i$ (produced at some \iindex $i$) that is confirmed by an honest node $p$ at \timeslot $t$, \ie, $b_i \in \dC_p(t)\trunc{\confDepth}$.
Since $b_i$ is $\confDepth$-deep, there must have been at least $6\Kcp + 1$ \iindices after $i$.
Due to \cref{lem:many-pps-one-cps-informal}, there must have been at least three \sltcps $j,k,l$ after \iindex $i$.
Due to \cref{lem:cps-stabilize-informal}, the blocks produced at these \iindices, $b_j,b_k,b_l$ are in $\dC_q(t')$ for all $t' \geq t$ and for all $q$ (see \cref{fig:pos-safety-proof}).
Therefore, $\dC_p(t)\trunc{\confDepth} \preceq \dC_q(t')$, and from this, we can prove (1).

To prove (2), suppose that node $p$ did not process the content of block $b_i$.
Since block $b_j$, and hence also $b_i$, is in every honest node's longest processed chain at \timeslot $t_{k+1} - 1$ (\cref{lem:cps-stabilize-informal}), it must have been that $p$ saw an equivocation for $b_i$ before \timeslot $t_{k+1} - 1$ (otherwise it must have actually processed the content of $b_i$).
Due to synchrony, all honest nodes see the headers of $b_i$ and its equivocation.
Since the block $b_k$ is produced by an honest node, and $k < i + 4\Kcp = i + \keqproof$, $b_k$ must contain an equivocation proof against $b_i$ (see \cref{fig:pos-safety-proof}).

To prove (3), we show that while confirming the block $b_i$, either all nodes see an equivocation proof against $b_i$ or none of them do.
The latest that an equivocation proof against $b_i$ can be included is $\keqproof$ blocks below $b_i$.
Since $\confDepth > \keqproof + 2\Kcp$, due to \cref{lem:many-pps-one-cps-informal}, the \sltcp $l$ must have occurred after $b_i$ became $\keqproof$-deep and before it became $\confDepth$-deep (see \cref{fig:pos-safety-proof}).
Thus, for all $p$ and $t$, if $b_i \in \dC_p(t)\trunc{\confDepth}$, then $b_l \in \dC_p(t)$, hence all nodes agree on whether or not an equivocation proof was included.

\emph{Liveness} follows similarly to PoW NC: Within $2\Kcp$ \iindices, there are enough honestly produced blocks to include new transactions, and their contents will never be blanked. In $\Theta(\Kcp)$ \timeslots, these blocks will become $\confDepth$-deep and will be confirmed.
\end{proof}

\subsection{Handling Loss of Predictable Validity}
\label{sec:predictablevalidity}

\subsubsection{Predictable Transaction Validity}
\label{sec:predictablevalidity-transaction}

In UTXO-based chains like Bitcoin (account-based like Ethereum), a transaction is \emph{valid} iff its inputs are unspent (its execution succeeds and fees are paid).

\begin{definition}
\label{def:predictablevalidity-transaction}
    A transaction has \emph{predictable validity} iff: 
    validity
    at the time an honest node adds it to a block
    implies validity when that block gets executed.
\end{definition}

The blanking in \ProtShort leads to a loss of \emph{predictable transaction validity}. An honest block $B$ may include a transaction that depends on the contents of a previous block $A$ whose equivocations were not known at the time. 
After block $B$ is produced, the adversary could release an equivocation for the block $A$, forcing honest nodes to blank block $A$'s contents, which may invalidate the transaction in block $B$. Such invalidated transactions take up free space in honest blocks and lower the effective throughput (valid confirmed transactions) of the ledger.

We propose a simple solution to recover predictable validity for \ProtShort:
If nodes limit transactions they include in a block to those that don't depend on any \emph{recently changed} state, then they can be sure that all equivocations that could affect the validity 
of a transaction already have a corresponding equivocation proof included on chain. 
This is because at the time of creating a block, honest nodes \emph{have seen all transactions which will be executed}, however, \emph{not all transactions nodes have seen will be executed}. The following lemma follows easily.

\begin{lemma}
    \label{lem:pred-valid-1}
    If a node produces a block whose transactions do not share state with any transaction included in the last $\keqproof$ blocks, then the block satisfies 
    \cref{def:predictablevalidity-transaction}.
\end{lemma}

\subsubsection{Predictable Fee Validity}
\label{sec:predictablevalidity-fee}

In practice, %
in popular DeFi-ecosystems, which consist of very interdependent transactions~%
\cite{guo2019graph,chen2020understanding}, it may not always be practical to limit the interaction between transactions. %
We propose instead %
preserving the minimum requirement that each transaction \textit{pays its fee}, regardless of the outcome of its execution. This guarantees that miners are compensated for space used in their blocks, and also makes it costly for the adversary to take up space with invalid transactions.

\begin{definition}
\label{def:predictablevalidity-fee}
    A transaction has \emph{predictable fees} 
    iff:
    ability to pay fees when an honest node adds it to a block
    implies ability to pay fees when the block executes.
\end{definition}

In systems like Ethereum, transactions have a \emph{max gas} value set by the sender which limits the computation allowed by the transaction and ultimately its fee. We consider a protocol with this gas mechanism, as well as a base transaction cost that covers the block space taken up by the transaction. 
We introduce a notion of \emph{gas deposit accounts} to \ProtShort that can only be used for transaction fees (transactions internally do not have access to these accounts).
When a miner includes a transaction, it checks that the account funding the transaction has enough funds to cover the maximum gas, even if all transactions in its recent ancestor blocks make it to the blanked ledger and consume their maximum gas. Users thus need to maintain a balance proportional to the complexity and frequency of the transactions they make. 
We also require that any deposit to the account is not considered in the balance until $\keqproof$ blocks after the deposit transaction. 
Withdrawals can take place immediately, as direct transactions.

\begin{lemma}
    \label{lem:pred-valid-2}
    If a node produces a block whose transactions are funded by gas deposit accounts with sufficient balance (balance before $\keqproof$ blocks minus any fees since),
    then all transactions in the block satisfy 
    \cref{def:predictablevalidity-fee}.
\end{lemma}

The solutions in \cref{sec:predictablevalidity-transaction,sec:predictablevalidity-fee} are complementary and could each be adopted as per-validator heuristics (\ie, not a consensus rule), or by the system based on the use-case (\eg, expected inter-dependency of transactions).

\section{Discussion}
\label{sec:conclusion}
\myparagraph{Tightening the Analysis}
Our \teaserattack and security analysis (\cf \cref{fig:comparison-bddelay-bdbandwidth}) serve as the first lower and upper bounds on nodes' minimum capacity required to ensure security in the bounded-capacity PoW setting. A question remains on how to tighten the gap.
One avenue for future work is whether the adversary has better strategies
than the \teaserattack, which we believe may be optimal in the bounded-capacity model.
\begin{conjecture}
    For the PoW NC protocol with the \rulelc scheduling policy,
    for all $\beta < 1/2, \blkratetime$ for which the \teaserattack is unsuccessful,
    the protocol is secure (against all attacks).
\end{conjecture}

Conversely, we expect that the security analysis can be improved in multiple ways. The current analysis uses only a few basic properties \ref{item:good-download-rule-no-repeat}, \ref{item:good-download-rule-honest-block}, \ref{item:good-download-rule-cutoff} of the scheduling policy. As a result, we assume that any valid block can be used by the adversary to spam honest nodes. However, when using the \rulelc policy, the adversary can only force honest nodes to process blocks that are on their longest header chain, which is already hard for the adversary given that the honest chain has been growing so far. An improved analysis should account not just for the number of block productions in the adversary's budget but also their blocktree structure. Further, \sltgood \timeslots are sufficient but not necessary for chain growth. Improved analysis of the chain growth rate using techniques such as blocktree partitioning \cite{dem20} can further tighten the analysis.

\myparagraph{Variable Difficulty}
In practice, PoW blockchains implement a difficulty adjustment algorithm (DAA) to maintain the target block rate as players join and leave the system.
This introduces new avenues for attack \cite{bahack2013theoretical}.
The variable difficulty protocol has so far been proven secure only in the lock-step synchronous 
model (\ie, messages delivered in exactly one round) 
\cite{garay2017bitcoin}. Security in the bounded-delay and bounded-capacity models remains an open problem.
We note, however, that the DAA seems to apply even more stress to limited capacity nodes,
as 
it would
lower the difficulty to compensate for chain growth rate lost due to congestion,
leading to an increase in the overall block rate of the system.
In turn, this would increase congestion, in particular
if honestly produced orphaned blocks are processed by honest parties, leading to a vicious cycle.
Under the \rulelc scheduling policy that we consider in this work,
honest nodes do not prioritize processing orphaned blocks,
but this appears to be the case for scheduling policies that 
allow for processing multiple blocks in parallel. The nuance in this analysis is left for future work.

\section*{Acknowledgment}
We thank
Lei Yang, Mohammad Alizadeh,
Sundararajan Renganathan, David Mazières,
Ertem Nusret Tas,
Ghassan Karame,
Florian Tschorsch,
and
George Danezis
for fruitful discussions.
The work of LK, JN, and AZ
was conducted in part
during Dagstuhl Seminar \#22421.
LK is supported by the
armasuisse Science and Technology CYD Distinguished Postdoctoral Fellowship.
JN is supported by the Protocol Labs PhD Fellowship.
JN and SS are supported by a gift from
the Ethereum Foundation and by a research hub funded by Input Output Global Inc.
SS is suported by NSF grant CCF-1563098.
AZ is supported by grant \#1443/21 from the Israel Science Foundation.
\bibliographystyle{ACM-Reference-Format}
\bibliography{references}

\ifshortVersion
    \appendix
    
    \section*{Appendices}
    Full version with appendices: \cite{full-version}.

    \nocite{doi:10.1080/01621459.1963.10500830,duchi-hoeffding,stackexchange-math-rwreturnto0,poisson_tail}
\else
    \nocite{full-version}
\appendix
\crefalias{section}{appendix}
\crefalias{subsection}{subappendix}
\crefalias{subsubsection}{subsubappendix}
\crefalias{subsubsubsection}{subsubsubappendix}

\section{Simulation Details}
\label{sec:attacks-details}

Nodes in our simulation\footnote{Source code: \gitSourceUrl}
generate blocks in a Poisson process with rate proportional to their mining power. We assume the mining difficulty is fixed, and do not include any adjustment by a difficulty adjustment algorithm. In fact, difficulty adjustment algorithms tend to worsen processing problems as they increase the block creation rate if the chain does not grow fast enough---which in turn requires more processing from nodes. 

Nodes process blocks one at a time according to the priority dictated by the processing policy, at a rate determined by their capacity. They are allowed to preempt their current task if new information (headers that are published, blocks that they mined) presents them with a higher priority target. Since queues can grow large if nodes do not manage to process all blocks in a timely manner, we maintain priority queues of bounded size (typically 100) and evict low priority tasks from the queue as needed.
As preemption of tasks may cause nodes to alternate between tasks, we 
allow nodes to retain partial work in an LRU cache of size 10.

Except where we note otherwise, headers are assumed to propagate instantly in the simulations.
To simulate an idealized bounded-delay network, where needed, we set the header propagation delay to $\Delta$ and the capacity of each node to be $\infty$.
Block headers %
contain the relevant lottery information which can be easily validated. We therefore assume the adversary never publishes headers it did not actually mine.

To remain close to the theoretical analysis, we model all processing tasks as dependent only on the resources available to the node itself. In reality, things are much more complex: nodes typically propagate blocks in a peer-to-peer network, which means both the overlay network topology and the underlying internet topology both greatly impact block download rates and performance. 
Our simplified setting allows us to focus more on the congestion effects in isolation from the effects of topology and other peer-to-peer related issues.

\section{Other Congestion-Based Attacks}
\label{sec:general-attacks}

\subsection{\GreedyAttack}
\label{sec:greedy-attack}
The \teaserattack relied on the fact that the adversary could entice nodes with a long header chain that is later discovered to be unavailable for processing. It is natural in this case to consider adjusting the scheduling policy to one that prefers the proverbial `bird in the hand over two birds in the bush', \ie, to extend the blocks we already processed over the illusive promise of a longer chain that the adversary may withhold from us. 

\emph{\ruleGreedy policy.}\;\;
This policy prioritizes processing blocks that extend the chain a node has already processed. If a header of a block at height $h$ is announced, and we already have $h_i$ blocks from that chain,
we set the priority of the block to be $(h_i,h)$ and compare between the two priorities lexicographically.

While the \rulegreedy policy performs well at high processing rates, we unfortunately find that it performs poorly in the low processing rate regime. Specifically, if a fork in the chain occurs, and nodes are split evenly between the two alternatives, the fork may never resolve. This is because nodes extend their own chain, and prioritize processing on their side of the split while having insufficient processing power to catch up with the other alternative chain. A fork in the chain can result from a deliberate attack by an adversary that releases blocks selectively to different nodes, by a network split, or worse, by an unlucky timing of honest node mining events. In this case, the blockchain fails even for small adversaries. 
Importantly, a fork that never resolves is either a safety or a liveness failure, as no transaction on either side of the split can be safely accepted.

\begin{figure}[tb]%
    \centering%
    \begin{tikzpicture}[]
        \scriptsize
        \begin{axis}[
                mysimpleplot,
                xmode=log, 
                xlabel={Capacity: $\bwtime$ [blocks per second]},
                ylabel={Growth of agreed chain\\ ($\mathrm{height}/\blkratetimeHon\cdot T$)},
                legend columns=2,
                xmin=1e-2, xmax=1e3,
                ymin=0, ymax=1.005,
                height=0.4\linewidth,
                width=\linewidth,
                yticklabel style={
                        /pgf/number format/fixed,
                        /pgf/number format/precision=2
                    },
                scaled y ticks=false,
                major grid style={solid,draw=gray!10},
            ]

            \addplot [myparula11, only marks,
                     mark size = 1pt] table [x=bandwidth,y=max_ancestor_height] {figures/fig-experiment-greedy-longestdl-data.txt};
            \addlegendentry{\ruleLc policy};
            \label{plt:experiment-greedy-longestdl};

            \addplot [myparula73, mark size=1pt,%
            dotted, 
            mark=o ] table [x=bandwidth,y=max_ancestor_height] {figures/fig-experiment-greedy-greedydl-data.txt};
            \addlegendentry{\ruleGreedy policy};
            \label{plt:experiment-greedy-greedydl};

        \end{axis}
    \end{tikzpicture}%
    \vspace{-0.5em}%
    \caption{%
        The rate nodes grow the agreed chain after the network splits into two sets of 50 nodes for 15 secs, when the scheduling policy is ``longest-header-chain'' (\ref{plt:experiment-greedy-longestdl}) or ``\rulegreedy'' (\ref{plt:experiment-greedy-greedydl}).
        Nodes using the \rulegreedy policy prioritize processes on their current chain.
        Under low capacity, they do not recover from the split, resulting in two chains forking at genesis, providing no growth of the agreed chain. Thus, longest chain is insecure \emph{without an adversary}
        (\cf \cref{fig:comparison-bddelay-bdbandwidth}(c)).
    }%
    \label{fig:experiment-greedy}%
\end{figure}%

To demonstrate this scheduling policy in action, we simulate a network of 100 nodes that are split evenly between two partitions for only 15 seconds, \ie, for an expected time required to produce 15 blocks.%
\footnote{Such short splits are relatively easy to induce in reality (transient problems with Internet routing, denial-of-service on the network, etc.) and thus a practical scheduling rule must recover from such splits.}
Once the network split ends, the simulation continues for another 4000 seconds, allowing nodes the opportunity to 
converge on a chain.
We 
measure the height of the latest block all nodes agree upon. If nodes do not recover from the partition, 
this block will be the genesis and the liveness of the protocol has failed. Otherwise, nodes quickly agree on the main chain and the height of the latest agreed block is 
just a little behind the longest tip of the chain.

We simulate the evolution after a brief partition for both the \rulelc policy as well as for the \rulegreedy policy. Our results (\cref{fig:experiment-greedy}) show that in settings where capacity is greater than $1/2$, nodes manage to catch up with the chain and the rate of growth matches for both scheduling policies. In lower capacity settings, however, nodes never catch up.
Note that this attack requires no adversary mining,
yet the protocol is insecure (\cf \cref{fig:comparison-bddelay-bdbandwidth}(c)).
This is in stark contrast to the bounded-delay analysis
which suggests that the protocol retains security
against a non-mining adversary
at any capacity (\cf \cref{fig:comparison-bddelay-bdbandwidth}(a)),
and highlights again the need to study the security of blockchains at capacity.

\subsection{The \PoSTeaserAttack (PoS)}
\label{sec:pos-teaser-attack}

\begin{figure}[tb]%
    \centering%
    \begin{tikzpicture}[]
        \footnotesize
        \begin{axis}[
                mysimpleplot,
                xlabel={Capacity: $\bwtime$ [blocks per second]},
                ylabel={Honest chain growth\\rate ($\blkratetimeGrowth/\blkratetimeHon$)},
                xmin=0, xmax=2,
                ymin=0, ymax=0.7,
                height=0.5\linewidth,
                width=\linewidth,
                yticklabel style={
                        /pgf/number format/fixed,
                        /pgf/number format/precision=2
                },
                scaled y ticks=false,
                legend columns=2,
            ]

            \addplot [myparula11, %
                    only marks, mark size=1.5pt] table [x=bandwidth,y=chain_growth] {figures/fig-experiment-teaser-noattacker-data.txt};
            \addlegendentry{No attack or private attack};
            \label{plt:experiment-teaser-noattacker};

            \addplot [myparula73, mark size=1.5pt,%
            only marks] table [x=bandwidth,y=chain_growth] {figures/fig-experiment-teaser-activeattacker-data.txt};
            \addlegendentry{\Teaserattack};
            \label{plt:experiment-teaser-attacker};

            \addplot [myparula54, mark size=1.5pt,%
            only marks] table [x=bandwidth,y=chain_growth] {figures/fig-experiment-teaserequiv-equivteasingattacker-data.txt};
            \addlegendentry{\PoSTeaserattack (PoS)};
            \label{plt:experiment-teaser-equivattacker};

        \end{axis}
    \end{tikzpicture}%
    \vspace{-0.5em}%
    \caption{%
    Results of a simulation of the \PoSTeaserattack comparing the rate of chain growth relative to honest block production rate, when nodes prioritize processing towards the longest header chain, for various capacity limits. The honest chain growth rate falls to almost zero as the adversary spams honest nodes with longer chains. The \teaserattack is shown for comparison.
    }%
    
    \label{fig:experiment-teaser-pos}%
\end{figure}%

\import{./figures/}{fig-attack-pos-teaser.tex}

In PoS, 
the adversary can greatly increase the network's processing load using equivocations. The \PoSteaserattack, described in \cref{fig:attack-pos-teaser}, uses equivocations to announce a whole new chain at every instance when the \teaserattack would have announced a single new block.
As the attack goes on, the length of the new announced chain increases. This increases the time honest nodes spend processing this chain, and \emph{decelerates} the honest chain growth until it comes to a halt. As a result,
in \cref{fig:experiment-teaser-pos}, 
the chain growth rate under the \PoSteaserattack is nearly zero.

As in the \teaserattack, the adversary starts by producing a private chain. Assuming the adversary's block production rate $\blkratetimeAdv$ is less than the honest chain growth rate before the attack ($\blkratetimeGrowthSilent$), the probability that the adversary produces a chain of length $k$ before the honest chain reaches length $k$ is $e^{-O(k)}$ \cite{nakamoto_paper,dem20}. This means that with probability $\epsilon$, the adversary eventually produces a private chain of length $k = O(\log(1/\epsilon))$, of which it can announce equivocations during the attack.
Since this chain is longer than the honest chain, it has higher scheduling priority.
It takes honest nodes $k/\bwtime$ time to process such a chain, during which time, honest nodes do not process blocks on the honest chain. So, any honest blocks produced within $k/\bwtime$ time after the first honest block at height $h$ do not grow the honest chain (\cref{fig:attack-teaser}(e)). If $\blkratetimeHon k/\bwtime$ is large, then there are many honest blocks that do not lead to chain growth, causing the chain growth rate $\blkratetimeGrowth$ to drop
(\cref{fig:experiment-teaser-pos}). 
As in the \teaserattack, if the adversary's block production rate $\blkratetimeAdv$ exceeds $\blkratetimeGrowth$, then the adversary succeeds in maintaining the number of block productions required for the attack to go on forever. This eventually slows honest chain growth to a halt. Thus, if $\blkratetimeHon k/\bwtime$ is large, \ie, $\blkratetimeHon = \Omega(1/k) = \Omega\left(\frac{1}{\log(1/\epsilon)}\right)$, then the attack succeeds with probability $\epsilon$.
\section{Protocol \& Model Details}
\label{sec:algos-reference}

\subsection{Nakamoto Consensus Pseudocode}
\label{sec:algos-reference-pseudocode}

\begin{algorithm}[tb]%
    \caption{%
        Idealized NC protocol $\protocol$ with scheduling policy
        (helper functions: \cref{sec:algos-reference-helperfunctions},
        environment $\Env$: \appendixRef{\cref{sec:algos-reference-environment}},
        functionality $\FtreePoW$: \appendixRef{\cref{alg:hdrtree-pow}})%
    }%
    \label{alg:generic-lc-protocol}%
    \begin{algorithmic}[1]%
        \scriptsize%

        \LineComment{Global counter of \timeslots $t \gets 1,2,...$ of duration $\tau$ ($\tau \to 0$ for PoW)}
        
        \On{$\Call{init}{\genesisHeaderChain, \mathsf{genesisTxs}}$}
                \label{loc:generic-lc-protocol-init}
            \LineComment{Initialize header tree $\hT$, longest processed chain $\dC$, and mappings from block header to content $\TxsMap$}
            \State $\hT, \dC \gets \{\genesisHeaderChain\}, \genesisHeaderChain$
            \State $\TxsMap[\genesisHeaderChain] \gets \mathsf{genesisTxs}$
                \Comment{Unset entries of $\TxsMap$ are $\BLOCKUNKNOWN$}
        \EndOn
        
        \On{$\Call{receivedHeaderChain}{\Chain}$}
                \label{loc:generic-lc-protocol-receiveheader}
                \Comment{Called by $\Env$ or $\Adv$}
            \State \Assert{$\FtreePoW.\Call{verify}{\Chain}$}
                \Comment{Validate header chain}
            \State $\hT \gets \hT \cup \operatorname{prefixChainsOf}(\Chain)$
                \Comment{Add $\Chain$ and its prefixes to $\hT$}
            \State $\Env.\Call{broadcastHeaderChain}{\Chain}$
        \EndOn
        
        \On{$\Call{receivedContent}{\Chain, \txs}$}
                \label{loc:generic-lc-protocol-receivecontent}
                \Comment{Called by $\Env$ or $\Adv$}
            \LineComment{Defer processing the content until all prefixes' contents are processed}
            \State \textbf{defer until} $\forall \Chain' \prec \Chain\colon \TxsMap[\Chain'] \neq \BLOCKUNKNOWN$
            \State \Assert{$\Chain.\mathsf{txsHash} = \operatorname{Hash}(\txs)$}
            \State $\Call{receivedHeaderChain}{\Chain}$
                \Comment{Validate header chain}
            \State \Assert{$\operatorname{AreTxsValid}(\txs)$}
                \Comment{Validate content of chain}
            \State $\TxsMap[\Chain] \gets \txs$
            \State $\Env.\Call{uploadContent}{\Chain, \txs}$
            \LineComment{Update the longest processed chain}
            \State $\Tree' \gets \{ \Chain' \in \hT \mid \TxsMap[\Chain'] \neq \BLOCKUNKNOWN \}$
            \State $\dC \gets \argmax_{\Chain' \in \Tree'} \len{\Chain'}$
        \EndOn
        
        \At{\timeslot $t \gets 1,2,...$}
                \label{loc:generic-lc-protocol-mainloop}
                \Comment{NC protocol main loop}
            \State $\txs \gets \Env.\Call{receivePendingTxs}{\null}$
            \LineComment{Produce and disseminate a new block if eligible}
            \If{$\Chain' \neq \bot$ \textbf{with} $\Chain' \gets \FtreePoW.\Call{extend}{\dC, \txs}$}
                \State $\Env.\Call{broadcastHeaderChain}{\Chain'}$
                \State $\Env.\Call{uploadContent}{\Chain', \txs}$
            \EndIf
                \LineComment{Confirm all but the last $\confDepth$ blocks on the longest processed chain}
            \State $\LOG{}{t} \gets \operatorname{txsLedger}(\TxsMap, \dC\trunc{\confDepth})$
                \label{loc:generic-lc-protocol-confirmation}
                \label{loc:generic-lc-protocol-ledger}
                \Comment{Ledger of node $p$ at $t$: $\LOG{p}{t}$}
        \EndAt

        \Throughout
            \State{Download content for some $\Chain$ chosen by scheduling policy (\eg \appendixRef{\cref{alg:longest-header-chain-rule}})}
        \EndThroughout
        \label{loc:generic-lc-protocol-downloadrule}
    \end{algorithmic}%
\end{algorithm}%
\begin{algorithm}[tb]%
    \caption{%
        Idealized functionality $\FtreePoW$:
        block production lottery and header chain structure
        for PoW
        (helper functions: \cref{sec:algos-reference-helperfunctions})%
    }%
    \label{alg:hdrtree-pow}%
    \begin{algorithmic}[1]%
        \scriptsize%
        
        \On{$\Call{init}{\genesisHeaderChain, \mathsf{numNodes}}$}
            \State $N \gets \mathsf{numNodes}$
            \State $\Tree \gets \{\genesisHeaderChain\}$
                \Comment{Global set of valid header chains}
        \EndOn

        \On{$\Call{extend}{\Chain, \txs}$ \textbf{from} node $P$ (possibly adversary) \textbf{at} \timeslot $t$}
                \label{loc:hdrtree-pow-binding}
            \LineComment{Abstraction of proof-of-work lottery: each node can call this once per slot and produces a block with probability $\blkrateslot/N$ independently of other nodes and slots}
            \If{$\Lottery[P,t] \neq \bot$}
                \Return{$\bot$} 
                \Comment{Only one ticket per node and slot}
            \EndIf
            \State $\Lottery[P,t] \overset{\$}{\gets} (\TRUE \text{ with probability $\blkrateslot/N$, else } \FALSE)$
            \label{loc:hdrtree-pow-blockproductionlottery}
            \If{$\Chain\in\Tree \land \Lottery[P,t]$}
                    \Comment{Parent chain $\Chain$ is valid and lottery was won?}
                \LineComment{Produce a new block header extending $\Chain$}
                \State $\Chain' \gets \Chain\|\operatorname{newBlock}(\mathsf{txsHash}\colon \operatorname{Hash}(\txs))$
                \State $\Tree \gets \Tree \cup \{\Chain'\}$
                    \Comment{Register new header chain in header tree}
                \State \Return{$\Chain'$}
            \EndIf
            \State \Return{$\bot$}
        \EndOn
        
        \On{$\Call{verify}{\Chain}$}
            \State \Return{$\Chain \in \Tree$}
                \Comment{Header chain is valid if previously added to header tree}
        \EndOn
    \end{algorithmic}%
\end{algorithm}%
\begin{algorithm}[tb]
    \caption[]{\ruleLc rule $\dlrulelong$
    }
    \label{alg:longest-header-chain-rule}
    \begin{algorithmic}[1]
        \scriptsize
        \RealFunction{$\operatorname{dlLongestHdrChain}(\hT, \TxsMap)$}
            \State $\Tree' \gets \{ \Chain \in \Tree' \mid \TxsMap[\Chain] = \BLOCKUNKNOWN \}$ \Comment{Ignore processed chains}
            \State $\Chain \gets \argmax_{\Chain'\in\Tree'} \len{\Chain'}$   \Comment{Select the longest chain} \label{loc:pseudocode-longest-headear-chain-decision-line}
            \State $\Chain' \gets \argmin_{\Chain''\preceq\Chain\colon  \TxsMap[\Chain'']=\BLOCKUNKNOWN} \len{\Chain''}$ 
            \Comment{First unknown block on that chain
            (if non-existent: $\bot$)
            }
            \State \Return{$\Chain'$}
        \EndRealFunction
    \end{algorithmic}
\end{algorithm}

Pseudocode of an idealized NC protocol $\protocol$ is provided in \appendixRef{\cref{alg:generic-lc-protocol}}.
Details of the PoW-based block production lottery, \ie, of production and verification of blocks, are abstracted through an idealized functionality $\FtreePoW$ whose pseudocode is provided in \appendixRef{\cref{alg:hdrtree-pow}} (\cf~\cite[Fig.~2]{sleepy}, \cite[Alg.~3]{bwlimitedposlc}).
Pseudocode for the \rulelc scheduling policy $\dlrulelong$ is provided in \appendixRef{\cref{alg:longest-header-chain-rule}}.
Helper functions used in the pseudocode are detailed in \cref{sec:algos-reference-helperfunctions}.

\subsection{Helper Functions for Nakamoto Consensus Pseudocode}
\label{sec:algos-reference-helperfunctions}

\begin{itemize}
      \item $\operatorname{Hash}(\txs)$:\;\;
            Cryptographic hash function to produce
            a binding commitment to $\txs$
            (modelled as a random oracle)

      \item $\Chain' \preceq \Chain$, $\Chain \succeq \Chain'$:\;\;
            Relation that $\Chain'$ is a prefix of $\Chain$

      \item $\Chain \| \Chain'$:\;\;
            Concatenation of $\Chain$ and $\Chain'$

      \item $\len{\Chain}$:\;\;
            Length of $\Chain$

      \item $(\TRUE \text{ with probability $x$, else } \FALSE)$:\;\;
            Bernoulli random variable with success probability $x$

      \item $\operatorname{prefixChainsOf}(\Chain)$:\;\;
            Set of prefixes of $\Chain$, \ie, all $\Chain'$ with $\Chain' \preceq \Chain$

      \item $\operatorname{newBlock}(\mathsf{txsHash}\colon \operatorname{Hash}(\txs))$ and
      \\
      $\operatorname{newBlock}(\mathsf{time}\colon t, \mathsf{node}\colon P, \mathsf{txsHash}\colon \operatorname{Hash}(\txs))$:\;\;
            Produce a new PoW and PoS block header with
            given parameters, respectively

      \item $\operatorname{txsLedger}(\TxsMap, \Chain)$:\;\;
            Concatenates the block contents stored in $\TxsMap$ for the blocks along the chain $\Chain$, to obtain the corresponding transaction ledger
\end{itemize}

\subsection{Bounded-Capacity Model Environment $\Env$}
\label{sec:algos-reference-environment}

\import{./figures/}{fig-model.tex}

We study PoW NC (\appendixRef{\cref{sec:algos-reference-pseudocode}})
using the following model
for a network $\Env$ with finite capacity (\appendixRef{\cref{fig:model}}), and for the powers and limits of an adversary $\Adv$.

The environment $\Env$ initializes $N$ nodes and lets $\Adv$ corrupt up to $\beta N$ nodes at the beginning of the execution. Corrupted nodes are controlled by the adversary. Honest nodes run $\protocol$.
The environment maintains a mapping $\Env.\TxsMap$ from block headers to the block content (transactions). This mapping is referred to as the `cloud' in 
\appendixRef{\cref{fig:model}}.
$\Env$ also maintains for each node a queue of pending block headers
to be delivered after a delay determined by the adversary.
If $\Adv$ has not instructed $\Env$ to deliver a header $\DeltaHeader$ real time after it was added to the queue of pending block headers,
then $\Env$ delivers it to the node.

Honest nodes and $\Adv$ interact with $\Env$ via the following functions:
\begin{itemize}
      \item $\Env.\Call{broadcastHeaderChain}{\Chain}$:

            \noindent If called by an honest node, $\Env$
            enqueues $\Chain$ in the queue of pending block headers for each node, and notifies $\Adv$.
            Then, for each node $P$, on receiving $\Call{deliver}{\Chain,P}$ from $\Adv$,
            or when $\DeltaHeader$ time has passed since $\Chain$
            was added to the queue of pending headers, $\Env$ triggers $P.\Call{receivedHeaderChain}{\Chain}$.

      \item $\Env.\Call{uploadContent}{\Chain, \txs}$:

            \noindent $\Env$ stores a mapping from the header chain $\Chain$ to the content $\txs$ of its last block by setting $\Env.\TxsMap[\Chain] = \txs$.
            $\Env$ only stores the content $\txs$
            if $\mathrm{Hash}(\txs) = \Chain.\mathsf{txsHash}$.

      \item $\Env.\Call{receivePendingTxs}{\null}$:

            \noindent $\Env$ generates a set of pending
            transactions and returns them.
    
    \item
        If node $P$ at \timeslot $t$ requests the content associated with a block header $\Chain$, $\Env$ acts as follows.
If $\Env.\TxsMap[\Chain]$ is set, then let $\txs = \Env.\TxsMap[\Chain]$ (if not, $\Env$ ignores the request).
If the request was received from an honest node $P$,
if $\Env$ has recently triggered $P.\Call{receivedContent}{.}$ at a rate below $\bwtime$,
then $\Env$ triggers  $P.\Call{receivedContent}{\Chain, \txs}$ (else, $\Env$ ignores the request).
If the request was received from $\Adv$, $\Env$ sends $(\Chain, \txs)$ to $\Adv$.
\end{itemize}

At all times, $\Adv$ can trigger $P.\Call{receivedHeaderChain}{\Chain}$
and $P.\Call{receivedContent}{\Chain, \txs}$
for honest nodes $P$ (bypassing header delay and capacity constraint in an adversarially chosen way).

\section{Full Security Proof}
\label{sec:fullproof}

This section provides a self-contained
proof of the argument
developed in \cref{sec:proof}.

Nodes are identified using $p, q$.
We distinguish between three notions of `time':
\emph{\Timeslots} of $\protocol$ are indicated by $r, s, t$.
\Timeslots in which one or more blocks are produced form a sub-sequence $\{t_k\}$, defined in \cref{sec:fullproof-definitions}.
\emph{\Iindices} into this sub-sequence are denoted by $i, j, k$.
The physical parameters of our model,
header propagation delay $\DeltaHeader$ and capacity $\bwtime$, as well as the mining rate $\blkratetime$, are specified in units of \emph{real time}.

We denote by $\dC_p(t)$ the longest fully processed chain of an honest node $p$ at the end of \timeslot $t$, and let $\len{b}$ denote the height of a block $b$. We use the same notation $\len{\Chain}$ to denote the length of a chain $\Chain$, define $L_p(t)=\len{\dC_p(t)}$ and $L_{\min}(t) = \min_p L_p(t)$
(where ``$\min_p$'' ranges only over honest nodes).

We denote intervals of \iindices (or \timeslots) as $\intvl{i}{j} \triangleq \{i+1,...,j\}$, with the convention that $\intvl{i}{j} \triangleq \emptyset$ for $j \leq i$.
We study executions over a finite horizon of $\Thorizon$ \timeslots (or $\Khorizon$ \iindices), and any interval $\intvl{i}{j}$ with $i < 0$ or $j > \Khorizon$ considered truncated accordingly.
The notation $\intvl{i}{j} \intvlg K$ (resp.\ $\intvlgeq, \intvll, \intvlleq, \intvleq$) is short for $j-i > K$ (resp.\ $\geq, <, \leq, =$).
In the analysis, we denote with upper-case Latin letters several random processes over \iindices (\eg, $\Xat{k}$) or \timeslots (\eg, $\Hat{t}$).
For any set $I$ of \iindices (analogously for \timeslots), we define $\Xat{I} \triangleq \sum_{k \in I} X_k$.

We denote by $\kappa$ the security parameter. An event $\Event_{\kappa}$ occurs \emph{with overwhelming probability} if $\Prob{\Event_{\kappa}} \geq 1 - \negl(\kappa)$.
Here, a function $f(\kappa)$ is \emph{negligible} $\negl(\kappa)$, if for all $n>0$, there exists $\kappa_n^*$ such that for all $\kappa > \kappa_n^*$, $f(\kappa) < \frac{1}{\kappa^n}$.

\subsection{Probabilistic Model for PoW NC Executions}
\label{sec:fullproof-probmodel}

Recall that the protocol runs in \timeslots of duration $\slotduration$.
A \emph{block production opportunity} (\BPO) is a pair $(p,t)$ where according
to the PoW block production lottery,
node $p$ is eligible to produce a block in \timeslot $t$.
A \BPO is called \emph{honest} (resp.\ \emph{adversary}) if node $p$ is honest (resp.\ adversary).
The random variables $H_t$ and $A_t$ 
denote
the number of honest and adversary 
\BPOs
in \timeslot $t$, respectively.
When the number of nodes $N\to\infty$ and each node holds an equal rate of block production, by the Poisson approximation of a binomial random variable,
we have $\Hat{t}\overset{\text{i.i.d.}}{\sim}\mathrm{Poisson}((1-\beta)\rho)$
and $\Aat{t}\overset{\text{i.i.d.}}{\sim}\mathrm{Poisson}(\beta\rho)$, independent of each other and across \timeslots.
The total number of \BPOs per \timeslot is $\Qat{t} \triangleq \Hat{t} + \Aat{t}$.
An \emph{execution} refers to a particular realization of the random process $\{(\Hat{t}, \Aat{t})\}$.

\subsection{Definitions}
\label{sec:fullproof-definitions}

\myparagraph{\sltGood, \sltBad, and \sltEmpty \Timeslots}
\Timeslots without a \BPO are called \emph{`\sltempty'}.
A \timeslot is \emph{`\sltgood'} iff
it has 
exactly one honest \BPO and no adversary \BPOs,
and is followed by $\goodsep$ \sltempty \timeslots.
This definition is inspired by convergence opportunities \cite{pss16,sleepy,kiffer2018better}, loners \cite{dem20}, and laggers \cite{ren}.
Here, $\goodsep$ is an analysis parameter.
We define another analysis parameter $\goodsepbw$
which is related to $\goodsep$ as 
\begin{IEEEeqnarray}{C}
    \label{eq:goodsep-bw-equation}
    (\goodsep+1)\slotduration \triangleq \DeltaHeader + \goodsepbw / \bwtime.
\end{IEEEeqnarray}
Thus, $\goodsep, \goodsepbw$ are chosen such
that for a \sltgood \timeslot, every honest node can 
receive the block header for the honest \BPO, and
process content for $\goodsepbw$ blocks, before the next \BPO. 
Any non-\sltempty \timeslot which is not \sltgood is called \emph{`\sltbad'}.
\begin{definition}
    \label{def:slots}
    We call a \timeslot $t$ \emph{\sltgood}, \emph{\sltbad}, \emph{\sltempty},
    respectively,
    denoted as $\predGood{t}$, $\predBad{t}$, $\predEmpty{t}$, respectively, iff:
    \begin{IEEEeqnarray}{rCl}
        \predGood{t} &\;\triangleq\;& (\Hat{t} = 1) \land (\Aat{t} = 0) \nonumber \\ && \quad {}\land{} (\Hin{t}{t+\goodsep} + \Ain{t}{t+\goodsep} = 0)
        \IEEEeqnarraynumspace\\
        \predBad{t} &\;\triangleq\;& (\Hat{t} + \Aat{t} > 0) \land \lnot\predGood{t}
        \IEEEeqnarraynumspace\\
        \predEmpty{t} &\;\triangleq\;& (\Hat{t} + \Aat{t} = 0).
        \IEEEeqnarraynumspace
    \end{IEEEeqnarray}
\end{definition}
Note that $\predEmpty{t} = \lnot\predGood{t} \land \lnot\predBad{t}$.

We denote by $t_k$ the $k$-th non-\sltempty \timeslot.
Then, we can introduce random processes over \emph{\iindices},
with \iindex $k$ corresponding
to the $k$-th non-\sltempty \timeslot $t_k$.
Considering only \iindices simplifies analysis by not having to deal with \sltempty \timeslots.
The process $\{\Gat{k}\}$ counts good \timeslots,
with $\Gat{k} \triangleq \Ind{ \predGood{t_k} }$.
Correspondingly, $\{\Bat{k}\}$ counts bad \timeslots,
$\Bat{k} \triangleq 1 - \Gat{k}$.

The following fact shows the distribution of \sltgood \iindices.
\RestatePropXiIsIid*
\begin{proof}
    First, for any $k$,
    \begin{IEEEeqnarray}{rCl}
        \Prob{\Gat{k} = 1} &=& \Prob{ \predGood{t_k} \mid \lnot \predEmpty{t_k} } \\
        &=& \frac{\Prob{\predGood{t_k}}}{\Prob{\lnot \predEmpty{t_k}}}
        = \frac{(1-\beta)\blkrateslot e^{-\rho(\goodsep+1)}}{1-e^{-\blkrateslot}}. \IEEEeqnarraynumspace
    \end{IEEEeqnarray}
    Let $\probEmpty \triangleq \Prob{\Hat{t}+\Aat{t}=0}$.
    Take an \iid random process $\{T_k\}$ with $\Prob{T_k = t} = (1-\probEmpty)\probEmpty^t$ for $t \geq 0$.
    The random variables $\{T_k\}$ describe the inter-arrival times between non-empty slots.
    Take another \iid random process $\{\Gat{k}'\}$, independent of $\{T_k\}$, such that $\Gat{k}' = 1$ with probability $\Prob{\Hat{t} = 1 \land \Aat{t} = 0 \mid \Hat{t}+\Aat{t}>0}$ and $\Gat{k}' = 0$ otherwise.
    The random process $\{\Gat{k}\}$ can be equivalently defined as $G_k = 1$ iff $G_k' = 1$ and $T_k \geq \goodsep$.
    The independence of the random variables $\{\Gat{k}\}$ then follows from the independence of the random variables $\{(T_k, \Gat{k}')\}$.
\end{proof}
Throughout the analysis, we assume
$\probGood > \frac{1}{2}$ (`honest majority' assumption).

\myparagraph{Some \sltGood \Timeslots Imply Growth}
A special role is played by \sltgood \timeslots $t_k$
with the additional property that 
the block produced
at $t_k$ is `soon' processed by all honest nodes.
Intuitively, these lead to \emph{chain growth},
the cornerstone of NC security~\cite{sleepy,dem20}.
We count these \timeslots with $\{\Dat{k}\}$,
and all other non-\sltempty \timeslots with $\{\Nat{k}\}$.
Specifically,
$\Dat{k} \triangleq 1$ if $\predGood{t_k}$
\emph{and} the block produced at $t_k$
has been processed by all honest nodes by the end
of \timeslot $t_k + \goodsep$,
$\Dat{k} \triangleq 0$ otherwise,
and $\Nat{k} \triangleq 1 - \Dat{k}$.
Finally, we define two random walks
on \iindices of non-\sltempty \timeslots
with increments
$\{\Xat{k}\}$ and $\{\Yat{k}\}$
that are handy for the definition of probabilistic
and combinatorial pivots:
\begin{IEEEeqnarray}{rClCrCl}
    \label{eq:random_walks_X_and_Y}
    \Xat{k} &\triangleq& \Gat{k} - \Bat{k}
    &\qquad\qquad&
    \Yat{k} &\triangleq& \Dat{k} - \Nat{k}
    \IEEEeqnarraynumspace
\end{IEEEeqnarray}
Note that the increments $\{\Xat{k}\}$
are \iid, and not affected by adversary action,
while the increments $\{\Yat{k}\}$ \emph{do depend}
on the adversary action and are thus in particular
\emph{not} \iidPERIOD.
Also note that $\forall k\colon Y_k \leq X_k$ since $D_k = 1 \implies G_k = 1$.

\myparagraph{Probabilistic and Combinatorial Pivots}
\begin{definition}
    \label{def:pp}
    We call an \iindex $k$ a \emph{\sltpp} (short for \emph{probabilistic pivot}),
    denoted as $\predPP{k}$, iff:
    \begin{IEEEeqnarray}{rCl}
        \predPP{k} &\;\triangleq\;&  ( \forall \intvl{i}{j} \ni k\colon  \Xin{0}{i} < \Xin{0}{k} \leq \Xin{0}{j} )
        \IEEEeqnarraynumspace
    \end{IEEEeqnarray}
\end{definition}
\begin{definition}
    \label{def:cp}
    We call an \iindex $k$ a \emph{\sltcp} (short for \emph{combinatorial pivot}),
    denoted as $\predCP{k}$, iff:
    \begin{IEEEeqnarray}{rCl}
        \predCP{k} &\;\triangleq\;&  ( \forall \intvl{i}{j} \ni k\colon  \Yin{0}{i} < \Yin{0}{k} \leq \Yin{0}{j} )
        \IEEEeqnarraynumspace
    \end{IEEEeqnarray}
\end{definition}
This definition of \sltpps and \sltcps decouples \cite[Def.~5]{sleepy} into its \emph{probabilistic} aspects~\cite[Sec.~5.6.3]{sleepy} and \emph{combinatorial} aspects~\cite[Sec.~5.6.2]{sleepy},
and casts them as conditions
on a random walk,
inspired by~\cite{dem20,close-latency-security-ling-ren}, to simplify the analysis.
The decoupling is one of the key differences from the analysis in \cite{sleepy} (see \cref{fig:analysis-comparison-sleepy}).
Note that a \sltcp is also a \sltpp because $Y_i \leq X_i$.
Also, \appendixRef{\cref{def:pp}} is equivalent to \cref{def:pp-informal},
and \appendixRef{\cref{def:cp}} is equivalent to \cref{def:cp-informal}
(\cf proof of \cref{prop:pivot-conditions-equivalence}).

\subsection{Analysis in the Probabilistic Model}
\label{sec:fullproof-analysis-overview}

We now develop the tools needed to prove safety and liveness of PoW NC in the bounded-capacity model, following \cref{fig:analysis-comparison-sleepy}.
First, analogously to the combinatorial argument of~\cite{sleepy},
we show (\cref{sec:fullproof-analysis-cps-stabilize})
that blocks from \sltcps \emph{stabilize}, \ie,
they remain in the longest processed chain of all honest nodes forever.
This is useful because
\emph{if we know that \sltcps occur frequently},
then honest nodes can confirm transactions that
must lie in 
the prefix of
a \sltcp's block (\emph{safety}),
and \sltcps' blocks (being produced by honest nodes)
bring any outstanding transactions onto chain (\emph{liveness}).
We then show that \sltcps do indeed occur frequently:
We show with a new probabilistic argument
(\cref{sec:fullproof-analysis-many-pps})
that \emph{\sltpps are abundant}, \ie,
in every `sufficiently long' interval (\ie, of length $\Omega(\kappa^2)$),
a constant fraction of the \timeslots are \sltpps (\appendixRef{\cref{lem:many-pps}}).
Then, we show with a new combinatorial argument
(\appendixRef{\cref{sec:fullproof-analysis-many-pps-one-cps}})
that \emph{the adversary cannot prevent all \sltpps from becoming \sltcps},
\ie, in every `sufficiently long' interval,
there is at least one \sltcp (\appendixRef{\cref{lem:many-pps-one-cps}}).
As a result,
if honest nodes confirm transactions that are still on their longest processed chain
after `sufficiently long' time (\ie, confirmation latency $\Omega(\kappa^2)$),
then PoW NC $\protocol$ is safe and live under bounded capacity.

\subsubsection{Combinatorial Pivots Stabilize}
\label{sec:fullproof-analysis-cps-stabilize}

We now show that the honest block produced in a \timeslot corresponding to a \sltcp persists in the longest processed chain of all honest nodes forever after $\goodsep$ \timeslots after it was produced.
Towards this, we first show that if $\Dat{k}=1$, \ie, if all honest nodes process the block produced in the \sltgood \timeslot $t_k$,
then the length of the longest processed chain of honest nodes increases, \ie, a \emph{chain growth event} (made precise in \appendixRef{\cref{prop:chain-growth}}).
Due to this, since, by \appendixRef{\cref{def:cp}}, all intervals around a \sltcp contain more \iindices with $\Dat{k}=1$ than those with $\Dat{k}=0$,
there can never be some honest node with a longest processed chain that does not contain the block corresponding to the \sltcp (\appendixRef{\cref{lem:cps-stabilize}}).
This is because there are not enough blocks for any other chain to outnumber the \emph{chain growth events} that contributed to the growth of the processed chain containing the \sltcp's block.
Thus, the block corresponding to the \sltcp remains in all honest nodes' longest processed chains forever.
\appendixRef{\cref{lem:cps-stabilize}} is proven
analogously to the combinatorial argument of~\cite{sleepy}.

Recall that $\dC_p(t)$ is the longest processed chain of node $p$ at the end of \timeslot $t$, $\len{\Chain}$ denotes the length of chain $\Chain$, $L_p(t) = \len{\dC_p(t)}$ and the length of the ``shortest (across honest nodes) longest processed chain'' is $L_{\min}(t) = \min_p L_p(t)$ (where ``$\min_p$'' ranges only over honest nodes).
The following proposition says that 
$L_{\min}(t)$
grows for every \iindex $k$ with $\Dat{k}=1$,
\ie, these are ``chain growth events''.
\begin{proposition}
\label{prop:chain-growth}
If $\Dat{k} = 1$, then $L_{\min}(t_k + \goodsep) \geq L_{\min}(t_k-1) + 1$.
\end{proposition}
\begin{proof}
Since $\Dat{k} = 1$, \timeslot $t_k$ is a \sltgood \timeslot.
Let $b$ be the unique honest block produced in \timeslot $t_k$, and let honest node $p$ be its producer.
Since honest nodes produce blocks on their longest processed chain, $\len{b} = L_p(t_k-1) + 1 \geq L_{\min}(t_k-1) + 1$.
Further, $\Dat{k} = 1$ means that the block $b$ is processed by all honest nodes by the end of \timeslot $t_k + \goodsep$. Therefore, $L_{\min}(t_k + \goodsep) \geq \len{b}$.
\end{proof}

\begin{lemma}
\label{lem:cps-stabilize}
Let $b^*$ be the block produced in a non-\sltempty \timeslot $t_k$ such that $\predCP{k}$. 
Then, for all header chains $\Chain'$ that are valid at \timeslot $t \geq t_k + \goodsep$ and $\len{\Chain'} \geq L_{\min}(t)$: $b^* \in \Chain'$.
Also then, for all honest nodes $p$ and for all \timeslots $t \geq t_k + \goodsep$: $b^* \in \dC_p(t)$.
\end{lemma}

The following proposition is helpful for proving \appendixRef{\cref{lem:cps-stabilize}}.
\begin{proposition}
    \label{prop:chain-growth-interval}
    For any $i < j$,
    \begin{IEEEeqnarray}{C}
        L_{\min}(t_j + \goodsep) \geq L_{\min}(t_{i+1} - 1) + \Din{i}{j}.
    \end{IEEEeqnarray}
\end{proposition}
\begin{proof}
By noting that if $\Dat{k} = 1$, then $t_{k+1} > t_k + \goodsep$, and adding the result of \appendixRef{\cref{prop:chain-growth}} for each \iindex with $\Dat{k}=1$.
\end{proof}

\begin{proof}[Proof of \appendixRef{\cref{lem:cps-stabilize}}]
    Note that $\dC_p(t)$ is a valid chain at \timeslot $t$ and $\len{\dC_p(t)} = L_p(t) \geq L_{\min}(t)$. Therefore, it suffices to show the first claim of the lemma.
    
    For contradiction, let $s \geq t_k + \goodsep$ be the first \timeslot in which 
    there is a valid header chain $\Chain'$ such that 
    $\len{\Chain'} \geq L_{\min}(s)$ and $b^* \not\in \Chain'$.
    
    Let $b'$ be the block with maximum height on the chain $\Chain'$, such that $b'$ was produced in a \timeslot $t_i$ with $D_i = 1$.
    For $\Chain'$ to be a valid chain at \timeslot $s$, we need $t_i \leq s$.
    Since the block $b'$ is produced by an honest node, $b'$ extends $\dC_q(t_i-1)$ for some honest node $q$.
    Therefore, $\dC_q(t_i-1)$ is a prefix of $\Chain'$.
    This means that $b^* \not\in \dC_q(t_i-1)$.
    Moreover, $\len{\dC_q(t_i-1)} = L_q(t_i-1) \geq L_{\min}(t_i-1)$.
    If $i > k$, then $t_i-1 \geq t_k + \goodsep$ (since $D_k = 1$) and $t_i - 1 < s$ (shown above). 
    This is a contradiction because we assumed that $s$ is the first \timeslot such that $s \geq t_k + \goodsep$ and 
    $b^* \notin \Chain'$ and $\len{\Chain'} \geq L_{\min}(s)$ for some valid chain $\Chain'$.
    Since $b^*$ is the only block produced in slot $t_k$, $i=k$ is also not possible.
    We conclude that $i < k$.
    
    Since $D_i = 1$ and $b'$ is produced in \timeslot $t_i$,
    \begin{IEEEeqnarray}{C}
    \label{eq:block-i-download}
        L_{\min}(t_i + \goodsep) \geq \len{b'}.
    \end{IEEEeqnarray}
    By assumption,
    \begin{IEEEeqnarray}{C}
    \label{eq:block-j-switch}
        \len{\Chain'} \geq L_{\min}(s).
    \end{IEEEeqnarray}
    
    Let $t_j$ be the last non-\sltempty \timeslot such that $t_j \leq s$. Note that $j \geq k > i$. 
    We must consider two cases:
    \begin{enumerate}
    \item Case 1: $s \geq t_j + \goodsep$ or $\Dat{j}=0$.
    If $\Dat{j}=0$, we don't have to worry about whether the block from slot $t_j$ was processed by all honest nodes.
    If $\Dat{j} = 1$ but $s \geq t_j + \goodsep$, then we know that all honest nodes have processed the block from slot $t_j$ before the end of \timeslot $s$. That is,
    \begin{IEEEeqnarray}{rCl}
        L_{\min}(s) 
        &\geq& L_{\min}(t_j + \goodsep)  \IEEEeqnarraynumspace \\
        &\geq& L_{\min}(t_{i+1}-1) + \Din{i}{j} \quad \text{(from \cref{prop:chain-growth-interval})}  \IEEEeqnarraynumspace \\
        \label{eq:chain-growth-case1}
        &\geq& L_{\min}(t_{i} + \goodsep) + \Din{i}{j}. \IEEEeqnarraynumspace 
    \end{IEEEeqnarray}
    By definition of $b'$, all blocks in $\Chain'$ appearing after $b'$ correspond to \iindices $l$ with $\Dat{l}=0$. These blocks must be from distinct \iindices greater than $i$ but at most $j$. So,
    \begin{IEEEeqnarray}{C}
    \label{eq:adv-chain-case1}
        \len{\Chain'} \leq \len{b'} + \Nin{i}{j}.
    \end{IEEEeqnarray}
    From \eqref{block-i-download,block-j-switch,chain-growth-case1,adv-chain-case1}, we derive
    \begin{IEEEeqnarray}{rCl}
    \label{eq:pivot-contra-case1}
        \Din{i}{j} \leq \Nin{i}{j} \implies \Yin{i}{j} \leq 0 \implies \Yin{0}{i} < \Yin{0}{j} \IEEEeqnarraynumspace
    \end{IEEEeqnarray}
    where $i < k \leq j$.
    
    \item Case 2: $t_j \leq s < t_j + \goodsep$ and $\Dat{j} = 1$.
    In this case, the block from slot $t_j$ may not have enough time to be processed by all honest nodes before the end of slot $s$.
    However, for any $l < j$ such that $\Dat{l} = 1$, $t_l + \goodsep < t_j \leq s$, so there is enough time to process the block from \timeslot $t_l$.
    Let $l \in\intvl{i}{j-1}$ be the greatest index such that $\Dat{l} = 1$. Then, $t_j > t_l + \goodsep$, and $\Din{i}{l} = \Din{i}{j-1}$.
    \begin{IEEEeqnarray}{rCl}
        \label{eq:chain-growth-case2}
        L_{\min}(s) 
        &\geq& L_{\min}(t_j) \\
        &\geq& L_{\min}(t_l + \goodsep) \\
        &\geq& L_{\min}(t_{i+1} - 1) + \Din{i}{l} \quad \text{(from \cref{prop:chain-growth-interval})} \IEEEeqnarraynumspace \\
        &\geq& L_{\min}(t_{i} + \goodsep) + \Din{i}{j-1}.
    \end{IEEEeqnarray}
    Note that since $\Dat{j}=1$, $\Nin{i}{j} = \Nin{i}{j-1}$. Therefore, as in the previous case,
    \begin{IEEEeqnarray}{C}
    \label{eq:adv-chain-case2}
        \len{\Chain'} \leq \len{b'} + \Nin{i}{j-1}.
    \end{IEEEeqnarray}
    From \eqref{block-i-download,block-j-switch,chain-growth-case2,adv-chain-case2},
    \begin{IEEEeqnarray}{rCl}
    \label{eq:pivot-contra-case2}
        \Din{i}{j-1} \leq \Nin{i}{j-1} &\implies& \Yin{i}{j-1} \leq 0 \nonumber \\ 
        &\implies& \Yin{0}{i} < \Yin{0}{j-1}. \IEEEeqnarraynumspace
    \end{IEEEeqnarray}
    Note that since we assumed $s \geq t_k + \goodsep$ and $s < t_j + \goodsep$, we know that $j > k$. Therefore, $i < k \leq j-1$.
    \end{enumerate}
    In either case, \eqref{pivot-contra-case1} or \eqref{pivot-contra-case2} contradict the assumption $\predCP{k}$ (\appendixRef{\cref{def:cp}}).
\end{proof}

\subsubsection{Probabilistic Pivots Are Abundant}
\label{sec:fullproof-analysis-many-pps}

Previous analyses of NC~\cite{sleepy,dem20} show that sufficiently long intervals contain at least \emph{one \sltpp} (\cref{fig:analysis-comparison-sleepy}(a)). This was enough for the bounded-delay analysis because in the bounded-delay setting, every \sltpp is also a \sltcp. 
However, in the bounded-capacity setting, not every \sltcp is a \sltpp, because not every \sltgood \timeslot results in 
growth of the longest processed chain of honest nodes (\cref{fig:analysis-comparison-sleepy}(b)).
Thus, existence of one \sltpp in every large interval is not enough to conclude existence of one \sltcp in every large interval. 
Instead,
we prove,
using a concentration bound on the number of \sltpps (\appendixRef{\cref{prop:lower-tailbound-ppivots}}),
that long intervals of \iindices in fact contain
\emph{a number of \sltpps proportional
to the interval length} (\appendixRef{\cref{lem:many-pps}}).
Then, in \appendixRef{\cref{sec:fullproof-analysis-many-pps-one-cps}}, we prove that out of those many \sltpps, at least one must also be a \sltcp, which allows us to continue with the safety and liveness proofs from~\cite{sleepy}.

The key challenge in proving that there are many \sltpps is that for two \iindices $k_1, k_2$, the events that $k_1$ is a \sltpp and that $k_2$ is a \sltpp are dependent, because both events depend on overlapping intervals. But a key observation is that since the \sltpp condition
(\appendixRef{\cref{def:pp}})
already holds for large intervals with high probability (\appendixRef{\cref{prop:lower-tailbound-X}}), we only need to look at the small intervals. Then, for two \iindices $k_1,k_2$ that are sufficiently far apart, these short intervals are disjoint, and thus the corresponding \sltpp conditions are independent. Therefore, we decompose a long interval of \iindices into several groups of far-apart \iindices. This is illustrated in \cref{fig:ppivot-tailbound-illustration}, each group indicated by a different color. Within each group, by a concentration bound for \iid random variables, there are many \sltpps. Further, by a union bound, the concentration holds in all the groups simultaneously with high probability. This summarizes the proof of \appendixRef{\cref{prop:lower-tailbound-ppivots}}, which culminates in \appendixRef{\cref{lem:many-pps}} showing that with overwhelming probability, there are many \sltpps in every long enough interval.

\begin{figure}[tb]
    \centering
    \begin{tikzpicture}[x=0.7cm,y=0.4cm]
        \footnotesize

        \draw [draw=none,fill=myParula05Green,opacity=0.3] (0,0) rectangle (5,3);
        \draw [draw=none,fill=myParula07Red,opacity=0.3] (0,3) rectangle (5,6.5);
        \draw [draw=none,fill=myParula07Red,opacity=0.3] (6,0) rectangle (10.5,4);
        \draw [draw=none,fill=myParula05Green,opacity=0.3] (6,4) rectangle (10.5,6.5);
        \draw [draw=none,fill=myParula01Blue,opacity=0.3] (5,0) rectangle (6,6.5);

        \draw [-Latex] (0,0) -- (10.75,0) node [below] {$k$};
        \draw [-Latex] (0,0) -- (0,6.75) node [left] {$\Xin{0}{k}$};

        \foreach \x in {0,...,10} {
                \draw (\x,0) -- ++(0,0.1) -- ++(0,-0.2) node [below] {$\x$};
            }
        \foreach \x in {1,...,10} {
                \draw (\x,0) ++(-0.5,0) node [below=1em] {$\Xat{\x}$};
            }
        \foreach \y in {0,...,6} {
                \draw (0,\y) -- ++(0.1,0) -- ++(-0.2,0) node [left] {$\y$};
            }

        \pgfmathsetmacro{\sumY}{0}
        \foreach \d [count=\di from 0] in {1,1,1,-1,1,1,1,-1,1,1} {
                \draw [-latex,thick,shorten >=2pt,shorten <=2pt] (\di,\sumY) -- ++(1,\d);
                \pgfmathsetmacro{\tmp}{\sumY+\d}
                \global\let\sumY\tmp
            }

        \draw [thick,densely dotted] (10,6) -- ++(0.5,0.5);

    \end{tikzpicture}%
    \vspace{-0.5em}%
    \caption[]{Illustration of \sltpp
        (\eqref{pivot-conditions-equivalence-randomwalks}):
        A \sltpp as an \iindex $k$
        so that
        $\Xat{k} = 1$ (\tikz[x=0.75em,y=0.75em]{ \draw [draw=none,fill=myParula01Blue,opacity=0.3] (0,0) rectangle (1,1); })
        and
        $\Xin{0}{.}$ is strictly below
        $\Xin{0}{k}$ left of $k$
        and
        weakly above
        $\Xin{0}{k}$ right of $k$
        (\tikz[x=0.75em,y=0.75em]{ \draw [draw=none,fill=myParula05Green,opacity=0.3] (0,0) rectangle (1,1); }, \tikz[x=0.75em,y=0.75em]{ \draw [draw=none,fill=myParula07Red,opacity=0.3] (0,0) rectangle (1,1); }).%
    }
    \label{fig:pivot-randomwalk}
\end{figure}

We first identify insightful
alternative characterizations
of \sltpps, and a few propositions to help prove \appendixRef{\cref{prop:lower-tailbound-ppivots}}. \appendixRef{\cref{lem:many-pps}} follows from there.
\begin{proposition}
    \label{prop:pivot-conditions-equivalence}
    \begin{IEEEeqnarray}{rCl}
        \predPP{k}
        &\iff&  (\forall \intvl{i}{j} \ni k\colon  \Xin{i}{j} > 0)
        \label{eq:pivot-conditions-equivalence-intervals}
        \IEEEeqnarraynumspace\\
        &\iff&  (\forall \intvl{i}{j} \ni k\colon  \Gin{i}{j} > \Bin{i}{j})
        \IEEEeqnarraynumspace\\
        &\iff&  (\Xat{k} = 1) \land (\forall j\geq k: \Xin{k}{j} \geq 0)
        \IEEEeqnarraynumspace\nonumber\\
        && \quad {}\land{} (\forall i<(k-1): \Xin{i}{k-1} \geq 0)
        \label{eq:pivot-conditions-equivalence-randomwalks}
        \IEEEeqnarraynumspace
    \end{IEEEeqnarray}
\end{proposition}
\begin{proof}
    Elementary, using $\Xin{i}{j} = \Xin{0}{j} - \Xin{0}{i}$.
\end{proof}
In particular, \eqref{pivot-conditions-equivalence-randomwalks}
characterizes a \sltpp as an \iindex $k$
such that $\Gat{k} = 1$
and the simple random walks
$\ell\mapsto \Xin{k}{k+\ell}$
and
$\ell\mapsto \Xin{k-1-\ell}{k-1}$
starting at
$0$
remain non-negative forever
(\cref{fig:pivot-randomwalk}).
Due to this, we easily see that the probability that any given \iindex is a \sltpp is the probability that the \iindex is \sltgood and the two random walks never return to zero (\cref{prop:ppivot-randomwalk}).
In \appendixRef{\cref{prop:lower-tailbound-X}} by a simple concentration bound over \iid random variables, we show that in all large intervals, with high probability, the random walk $\Xat{k}$ advances proportionally to the interval length (due to its positive drift).

Throughout this section, assume that
$\probGood = \frac{1}{2} + \epsGood$
with $\epsGood \in (0,1/2]$.

\begin{proposition}
    \label{prop:lower-tailbound-X}
    With $\alphaLowerTailX \triangleq 2 \epsGood^2$, $\forall \intvl{i}{j}$, $\forall \delta \geq 0$:
    \begin{IEEEeqnarray}{C}
        \Prob{\Xin{i}{j} \leq (1-\delta) 2 \epsGood (j-i)}
        \leq \exp( - \alphaLowerTailX \delta^2 (j-i)).
        \IEEEeqnarraynumspace
    \end{IEEEeqnarray}
\end{proposition}
\begin{proof}
By Hoeffding's inequality~\cite{doi:10.1080/01621459.1963.10500830}~\cite[Thm.~4]{duchi-hoeffding}.
\end{proof}
\begin{proposition}[Hoeffding's inequality~{\cite{doi:10.1080/01621459.1963.10500830} \cite[Thm.~4]{duchi-hoeffding}}]
    \label{prop:hoeffding}
    Let $Z_1, ..., Z_n$ be independent bounded random variables with
    $\forall i: Z_i \in [a,b]$, where $-\infty < a \leq b < \infty$.
    Then, $\forall t\geq0$:
    \begin{IEEEeqnarray}{rCl}
        \Prob{\left(\sum_{i=1}^n Z_i\right) \geq \Exp{\sum_{i=1}^n Z_i} + t n}
        &\leq&
        \exp\left(\frac{-2 n t^2}{(b-a)^2} \right) \IEEEeqnarraynumspace
        \\
        \Prob{\left(\sum_{i=1}^n Z_i\right) \leq \Exp{\sum_{i=1}^n Z_i} - t n}
        &\leq&
        \exp\left(\frac{-2 n t^2}{(b-a)^2} \right) \IEEEeqnarraynumspace
    \end{IEEEeqnarray}
\end{proposition}

\begin{proposition}
    \label{prop:ppivot-randomwalk}
    \begin{IEEEeqnarray}{C}
        \forall k\colon\quad
        \Prob{\predPP{k}}
        \geq \probPPFormula
        \triangleq \probPP
    \end{IEEEeqnarray}
\end{proposition}
\begin{proof}
    In \eqref{pivot-conditions-equivalence-randomwalks},
    $\predPP{k}$ is characterized
    as the intersection of three independent events:
    \begin{IEEEeqnarray}{rCl}
        \Event_1
        &\triangleq&
        \{ \Xat{k} = 1 \}
        \\
        \Event_2
        &\triangleq&
        \{ \forall\ell\colon \Xin{k}{k+\ell} \geq 0 \}
        \\
        \Event_3
        &\triangleq&
        \{ \forall\ell\colon \Xin{k-1-\ell}{k-1} \geq 0 \}
    \end{IEEEeqnarray}
    Their probabilities are easily calculated~\cite{stackexchange-math-rwreturnto0}:
    \begin{IEEEeqnarray}{C}
        \Prob{\Event_1}
        = \probGood
        \qquad
        \Prob{\Event_2} = \Prob{\Event_3}
        = (2\probGood - 1) / \probGood
        \IEEEeqnarraynumspace
    \end{IEEEeqnarray}
\end{proof}

The process $\{\Pat{k}\}$ counts \sltpps,
with $\Pat{k} \triangleq \Ind{\predPP{k}}$.
\begin{proposition}
    \label{prop:lower-tailbound-ppivots}
    With $\alphaLowerTailPP \triangleq 2 \probPP^2$,
    \begin{IEEEeqnarray}{l}
        \forall \intvl{i}{j} \intvleq 2 K_1 K_2\colon\quad
        \Prob{\Pin{i}{j} \leq (1-\delta) \probPP 2 K_1 K_2}
        \nonumber
        \\
        \qquad\qquad\qquad\qquad\qquad {}\leq{} 2 K_1 \exp(- \alphaLowerTailPP \delta^2 K_2) + \Khorizon^2 \exp(-\alphaLowerTailX K_1).
        \IEEEeqnarraynumspace
    \end{IEEEeqnarray}
\end{proposition}
\begin{proof}
    Let
    $\Event \triangleq \{\forall \intvl{i}{j} \intvlgeq K_1\colon \Xin{i}{j} > 0\}$.
    From \appendixRef{\cref{prop:lower-tailbound-X}} with $\delta=1$,
    and a union bound over all intervals
    ($\leq \Khorizon^2$ many),
    we get
    \begin{IEEEeqnarray}{C}
        \Prob{\lnot\Event} \leq \Khorizon^2 \exp(-\alphaLowerTailX K_1).
    \end{IEEEeqnarray}

    For any given index $k$, we can
    partition
    the intervals of
    \eqref{pivot-conditions-equivalence-intervals}
    into `long'
    and `short'
    intervals (length at least vs.\ less than $K_1$):
    \begin{IEEEeqnarray}{rCl}
        \Event_k
        &\triangleq&
        \{ \predPP{k} \}
        = \Event_k^{\mathrm{L}} \land \Event_k^{\mathrm{S}}
        \\
        \Event_k^{\mathrm{L}}
        &\triangleq&
        \{\forall \intvl{i}{j} \ni k, \intvl{i}{j} \intvlgeq K_1\colon \Xin{i}{j} > 0\}
        \IEEEeqnarraynumspace
        \\
        \Event_k^{\mathrm{S}}
        &\triangleq&
        \{\forall \intvl{i}{j} \ni k, \intvl{i}{j} \intvll K_1\colon \Xin{i}{j} > 0\}.
    \end{IEEEeqnarray}
    Note that $\Event_k^{\mathrm{L}} \supseteq \Event$.
    Also, 
    for any two given \iindices $k_1, k_2$ that are `far apart',
    \ie,
    if $\abs{k_1 - k_2} \geq 2 K_1$,
    then
    $\Event_{k_1}$ and $\Event_{k_2}$ are conditionally independent
    given $\Event$
    (since $\Event_{k_1}^{\mathrm{S}}$ and $\Event_{k_2}^{\mathrm{S}}$ are).

    We decompose $I^* \triangleq \intvl{i}{j} = \intvl{i}{i + 2 K_1 K_2} = \bigcup_{\ell=1}^{2K_1} I_{\ell}$:
    \begin{IEEEeqnarray}{rCl}
        \forall\ell\in\{1,...,2K_1\}\colon\quad\
        I_{\ell}
        &\triangleq&
        \{ i+0\cdot 2K_1+\ell, ...
        \nonumber\\
        && \quad{} ..., i+(K_2-1)\cdot 2K_1+\ell \}.
        \IEEEeqnarraynumspace
    \end{IEEEeqnarray}
    See \cref{fig:ppivot-tailbound-illustration}
    for illustration.
    We define corresponding events, $\forall\ell\in\{1,...,2K_1\}$:
    \begin{IEEEeqnarray}{rCl}
        \Event^*
        &\triangleq&
        \left\{ \Pat{I^*} \leq (1-\delta) \probPP 2 K_1 K_2 \right\}
        \\
        \Event_{\ell}
        &\triangleq&
        \left\{ \Pat{I_{\ell}} \leq (1-\delta) \probPP K_2 \right\}.
    \end{IEEEeqnarray}
    Clearly, $\Event^* \subseteq \bigcup_{\ell=1}^{2 K_1} \Event_\ell$.
    Thus, by a union bound,
    \begin{IEEEeqnarray}{rCl}
        \Prob{ \Event^* \cond \Event }
        &\leq&
        \sum_{\ell=1}^{2 K_1} \Prob{ \Event_\ell \cond \Event }.
        \IEEEeqnarraynumspace
    \end{IEEEeqnarray}
    Furthermore, $\forall\ell\in\{1,...,2K_1\}$,
    and with $\mu_\ell \triangleq \Exp{ \Pat{I_{\ell}} \cond \Event}$:
    \begin{IEEEeqnarray}{rCl}
        \IEEEeqnarraymulticol{3}{l}{
            \Prob{ \Event_\ell \cond \Event }
            =
            \Prob{ \Pat{I_{\ell}} \leq (1-\delta) \probPP K_2 \cond \Event }
        }
        \IEEEeqnarraynumspace
        \\\quad
        &\leqA&
        \Prob{ \Pat{I_{\ell}} \leq (1-\delta) \mu_\ell \cond \Event }
        \IEEEeqnarraynumspace
        \\
        &\leqB&
        \exp(-2 \delta^2 \mu_\ell^2 / K_2)
        \leqC
        \exp(-2 \probPP^2 \delta^2 K_2),
        \IEEEeqnarraynumspace
    \end{IEEEeqnarray}
    where
    (a) and (c)~use 
    \begin{IEEEeqnarray}{rCl}
        \mu_\ell = K_2 \Exp{\Ind{\predPP{k}} \cond \Event} &\geq& K_2 \Exp{\Ind{\predPP{k}}} \nonumber \\ &\geq& K_2 \probPP \IEEEeqnarraynumspace
    \end{IEEEeqnarray}
    (\cref{prop:ppivot-randomwalk}),
    and
    (b)~uses that
    $\{\predPP{k_1}\}$ and
    $\{\predPP{k_2}\}$
    are conditionally independent given $\Event$
    for $k_1, k_2 \in I_\ell$,
    and
    Hoeffding's inequality (\cref{prop:hoeffding}).

    To complete the proof, with $\alphaLowerTailPP = 2 \probPP^2$,
    \begin{IEEEeqnarray}{rCl}
        \Prob{ \Event^* }
        &=&
        \Prob{ \Event^* \cap \Event } + \Prob{ \Event^* \cap \lnot\Event }
        \\
        &\leq&
        \Prob{ \Event^* \cond \Event } + \Prob{ \lnot\Event }
        \\
        &\leq&
        2 K_1 \exp(-\alphaLowerTailPP \delta^2 K_2)
        + \Khorizon^2 \exp(-\alphaLowerTailX K_1).
        \IEEEeqnarraynumspace
    \end{IEEEeqnarray}
\end{proof}

\begin{lemma}
    \label{lem:many-pps}
    For
    $\Kcp = \Omega(\kappa^2)$,
    and $\Khorizon = \poly(\kappa)$,
    \begin{IEEEeqnarray}{C}
        \Prob{
            \forall \intvl{i}{j} \intvlgeq \Kcp\colon
            \Pin{i}{j} \geq (1-\delta) \probPP \Kcp}
        \nonumber
        \\
        \qquad{}\geq{} 1 - \exp(- \Omega(\kappa)) = 1 - \negl(\kappa).
    \end{IEEEeqnarray}
\end{lemma}
\begin{proof}
From \appendixRef{\cref{prop:lower-tailbound-ppivots}} by setting $K_1,K_2 = \Omega(\kappa)$ and $\Kcp = 2K_1K_2$.
\end{proof}

\subsubsection{Many Probabilistic Pivots Imply One Combinatorial Pivot}
\label{sec:fullproof-analysis-many-pps-one-cps}

The \rulelc rule $\dlrulelong$
(\appendixRef{\cref{alg:longest-header-chain-rule}})
has a few useful properties.
Intuitively, nodes using this rule 
\begin{enumerate}[(P1),leftmargin=3em]
    \item \label{item:good-download-rule-no-repeat}
    process a \BPO's block's content at most once,
    \item \label{item:good-download-rule-honest-block} either process the most recent honest block, or fully utilize their capacity to process other blocks (\ie, do not stay idle), and
    \item \label{item:good-download-rule-cutoff} prioritize blocks that were produced `recently'.
\end{enumerate}

\ref{item:good-download-rule-no-repeat} holds by construction.
\ref{item:good-download-rule-honest-block} holds because the scheduling policy $\dlrulelong$
is never idle, and 
will always process towards an honest block
when it has processed all longer chains
and there is capacity remaining.
Moreover, we expect that in a secure execution,
\ref{item:good-download-rule-cutoff} holds because the longest header chain cannot fork off too much from the longest processed chain of an honest node, otherwise it would imply a safety violation.
More precisely, due to \appendixRef{\cref{lem:cps-stabilize}},
any longest header chain in any honest node's view must extend the block produced in the most recent \sltcp, and therefore blocks with the highest processing priority must have been produced after the most recent \sltcp.
If the adversary wants to prevent honest nodes from processing the block produced at a \sltgood \iindex $k$,
so that $G_k = 1$ but $D_k = 0$, then it can only ``distract'' them by providing $\goodsepbw$ blocks
produced after the most recent \sltcp
(\cref{prop:download-or-spend-budget}).

\RestatePropDownloadOrSpendBudget*

\begin{proof}
In \timeslot $t_k$, there is exactly one block $b$ produced by an honest node, 
the block header is made public at the beginning of the \timeslot,
and is seen by all honest nodes within $\DeltaHeader$ time.
Thereafter, each node has enough time to process $\goodsepbw$ blocks during \timeslots $[t_k, t_k + \goodsep]$.

Under the scheduling policy $\dlrulelong$,
if $\Dat{k} = 0$,
\ie an honest node did not process content for the block $b$ before the end of \timeslot $t_k + \goodsep$,
then
that honest node must process the content for at least $\goodsepbw$ blocks on chains longer than the height of the block $b$ or in the prefix of the block $b$.
Since honest nodes produce blocks extending their longest chain, $b$ extends $\dC_p(t_k-1)$ for some $p$.
Let $b^*$ be the block produced in \timeslot $t_i$ where $\predCP{i}$ (suppose $i$ exists).
$\predCP{i} \implies \Yat{i} = 1$, therefore this block is unique, and also $t_k > t_i + \goodsep$.
Due to \appendixRef{\cref{lem:cps-stabilize}}, any valid header chain longer than $b$ at time slot $t_k$ must contain $b^*$.
Therefore, the only blocks
that are processed by an honest node during \timeslots $[t_k, t_k + \goodsep]$
\begin{enumerate}
    \item must be produced after $t_i$ because they extend $b^*$, and
    \item must be produced no later than $t_k$ because there are no blocks produced in $\intvl{t_k}{t_k+\goodsep}$.
\end{enumerate}
In case a \sltcp $i<k$ does not exist, the claim is trivial.
\end{proof}

Given the above properties of the scheduling policy, we now want to show that \sltcps occur once in a while.
\Cref{fig:ppivot-cpivot-intuitive} illustrates the key argument for this.
To start, let us show that there is at least one \sltcp in $\intvl{0}{\Kcp}$.
From \appendixRef{\cref{lem:many-pps}}, there are many \sltpps in $\intvl{0}{\Kcp}$.
If there were no \sltcps in $\intvl{0}{\Kcp}$, then the adversary must prevent each \sltpp from turning into a \sltcp.
We know that in any interval around a \sltpp, there are more \sltgood \iindices than \sltbad \iindices (see top row in \cref{fig:ppivot-cpivot-intuitive}). In fact, \sltgood \iindices outnumber \sltbad \iindices by a margin that increases linearly with the size of the interval (\appendixRef{\cref{prop:lower-tailbound-X}}).
Therefore, for a \sltpp to not be a \sltcp, the adversary must prevent an honest node from processing the most recent honest block in several of these \sltgood \iindices (so that the corresponding $\Gat{k}=1$ \iindices have $\Dat{k}=0$).
\Cref{fig:ppivot-cpivot-intuitive} shows an example where the adversary prevented processing of the honest block in one \sltgood \iindex, and as a result, two of the \sltpps fail to become a \sltcp.
In \appendixRef{\cref{lem:one-cp-induction-base}}, through a combinatorial argument, we show that to prevent all of $n$ \sltpps in $\intvl{0}{\Kcp}$ from becoming \sltcps, the adversary must prevent processing of the honest block in at least $n/4$ \sltgood \iindices in $\intvl{0}{2\Kcp}$.
From \cref{prop:download-or-spend-budget}, for each such \iindex, the adversary must `spend' at least $\goodsepbw$ blocks that the honest node processs.
These blocks must come from a `budget' that can contain at most all blocks mined during $\intvl{0}{2\Kcp}$.
If this `budget' falls short of the number of blocks required to overthrow all \sltcps, then there must be at least one \sltcp in $\intvl{0}{\Kcp}$.

Next, we would like to show that there is at least one \sltcp in $\intvl{m\Kcp}{(m+1)\Kcp}$ for all $m \geq 0$ (where 
we just saw the base case $m=0$).
Here,
one may be concerned that
the adversary could save up many blocks from the past and attempt to make honest nodes process these blocks at a particular target \timeslot $t_k$.
But given that one \sltcp occurred in $\intvl{(m-1)\Kcp}{m\Kcp}$, \cref{prop:download-or-spend-budget} ensures that honest nodes will only process blocks that are produced after $(m-1)\Kcp$.
This allows us to bound the `budget' of blocks that the adversary can use to overthrow \sltcps, and therefore show that there is at least one \sltcp in $\intvl{m\Kcp}{(m+1)\Kcp}$.
This argument is formalized in \cref{lem:one-cp-induction-full}.

Below, we first show the proof for the base case (\ie for the interval $\intvl{0}{\Kcp}$) to highlight the key techniques.
Here, $\Qin{.}{.}$ is the total number of blocks mined in an interval (bounds the adversary's block budget), and the expressions on the left in \eqref{cp-induction-full-margin-condition,cp-induction-full-ppivots-condition} are the minimum number of blocks the adversary needs to produce to ensure that there are no \sltcps, in terms of the number of \sltpps $\Pin{.}{.}$ and number of \sltgood \iindices $\Gin{.}{.}$.

\begin{lemma}
    \label{lem:one-cp-induction-base}
    If honest nodes use the scheduling policy $\dlrulelong$
    and
    \begin{IEEEeqnarray}{C}
        \label{eq:cp-induction-base-margin-condition}
        \forall \intvl{i}{j} \intvlgeq \Kcp, i < \Kcp \colon\quad
        \frac{\goodsepbw}{2} \left( \Gin{i}{j} - \Bin{i}{j} \right) > \Qin{0}{j}, \IEEEeqnarraynumspace\\
        \label{eq:cp-induction-base-ppivots-condition}
        \frac{\goodsepbw}{4} \Pin{0}{\Kcp} > \Qin{0}{2\Kcp},
    \end{IEEEeqnarray}
    then $\exists k_1^* \in \intvl{0}{\Kcp} \colon \predCP{k_1^*}$.
\end{lemma}

Towards proving \appendixRef{\cref{lem:one-cp-induction-base}}, we show two simple corollaries of the \sltcp conditions (\cref{prop:not-cp-exists-interval,prop:not-cp-interval-properties}) and show that in any interval, \sltgood \iindices outnumber \sltbad \iindices by at least the number of \sltpps in that interval.

\begin{proposition}
\label{prop:not-cp-exists-interval}
\begin{IEEEeqnarray}{C}
    \lnot \predCP{k} \qquad\implies\qquad \exists \intvl{i}{j} \ni k \colon\quad \Yin{i}{j} \leq 0.
\end{IEEEeqnarray}
\end{proposition}
\begin{proof}
From \appendixRef{\cref{def:cp}}.
\end{proof}

\begin{proposition}
\label{prop:not-cp-interval-properties}
If $\Yin{i}{j} \leq 0$, then
\begin{IEEEeqnarray}{C}
    \label{eq:not-cp-interval-property1}
    \Nin{i}{j} \geq \Din{i}{j}, 
    \qquad
    \label{eq:not-cp-interval-property2}
    \Gin{i}{j} - \Din{i}{j} \geq \frac{1}{2} \left( \Gin{i}{j} - \Bin{i}{j} \right).
    \IEEEeqnarraynumspace
\end{IEEEeqnarray}
\end{proposition}
\begin{proof}
We obtain \eqref{not-cp-interval-property1} from the definition $\Yat{i} = \Dat{i} - \Nat{i}$.
Then,
\begin{IEEEeqnarray}{rCl}
    \Gin{i}{j} + \Bin{i}{j} &=& \Din{i}{j} + \Nin{i}{j} \\
    \Gin{i}{j} + \Bin{i}{j} &\geq& 2 \Din{i}{j} \\
    2 \Gin{i}{j} - 2 \Din{i}{j} &\geq& \Gin{i}{j} - \Bin{i}{j}.
\end{IEEEeqnarray}
\end{proof}

\begin{proposition}
\label{prop:ppivots-imply-honest-margin}
If $\Pin{i}{j} > 0$, then $\Gin{i}{j} - \Bin{i}{j} \geq \Pin{i}{j}$.
\end{proposition}
\begin{proof}
Let $n = \Pin{i}{j}$.
First, consider $n=1$.
There is exactly one \sltpp $k \in \intvl{i}{j}$.
From \appendixRef{\cref{def:pp}}, $\Xin{0}{i} < \Xin{0}{j}$. Therefore, $\Xin{i}{j} > 0$, hence $\Gin{i}{j} - \Bin{i}{j} \geq 1$.
For the general case, let $k_1,...,k_n$ be the \sltpps in $\intvl{i}{j}$. Then, we apply the $n=1$ case on the disjoint intervals $\intvl{i}{k_1}$, $\intvl{k_1}{k_2}, ...$, $\intvl{k_{n-1}}{j}$ and then sum up.
\end{proof}

\begin{figure}[tb]%
    \centering%
    \begin{tikzpicture}[x=2em,y=2em]%
    
        \begin{scope}[]
            
            \node at (-1.5,0) {\textbf{(a)}};
            \draw (-1,0) -- (4,0);
            
            \node [circle,fill=myParula01Blue,inner sep=2pt] (circle1) at (0,0) {};
            \node [circle,fill=myParula01Blue,inner sep=2pt] (circle2) at (1,0) {};
            \node [circle,fill=myParula01Blue,inner sep=2pt] (circle3) at (2,0) {};
            \node [circle,fill=myParula01Blue,inner sep=2pt] (circle4) at (3,0) {};

            \draw [draw=myParula07Red,rounded corners] ([xshift=-3pt,yshift=-3pt]circle1.south west) rectangle ([xshift=3pt,yshift=3pt]circle2.north east) node [midway,yshift=11pt,xshift=-0.75em] {\textcolor{myParula07Red}{$A$}};
            \draw [draw=myParula07Red,rounded corners] ([xshift=-3pt,yshift=-5pt]circle2.south west) rectangle ([xshift=3pt,yshift=5pt]circle4.north east) node [midway,yshift=13pt,xshift=1.75em] {\textcolor{myParula07Red}{$B$}};
            \draw [draw=myParula07Red,rounded corners] ([xshift=-5pt,yshift=-7pt]circle2.south west) rectangle ([xshift=5pt,yshift=7pt]circle3.north east) node [midway,yshift=15pt] {\textcolor{myParula07Red}{$C$}};
        
        \end{scope}
    
        \begin{scope}[xshift=4.6cm]
            
            \node at (-1.5,0) {\textbf{(b)}};
            \draw (-1,0) -- (4,0);
            
            \node [circle,fill=myParula01Blue,inner sep=2pt] at (0,0) {};
            \node [circle,fill=myParula01Blue,inner sep=2pt] at (0.8,0) {};
            \node [circle,fill=myParula01Blue,inner sep=2pt] at (2.2,0) {};
            \node [circle,fill=myParula01Blue,inner sep=2pt] at (3,0) {};
        
            \draw (-0.5,0) ++(0,0.2) -- ++(0,-0.4);
            \draw (3.5,0) ++(0,0.2) -- ++(0,-0.4);
            
            \draw [draw=myParula07Red,ultra thick] (1.5,0) ++(-0.15,-0.15) -- ++(0.3,0.3);
            \draw [draw=myParula07Red,ultra thick] (1.5,0) ++(-0.15,0.15) -- ++(0.3,-0.3);
            
            \draw [draw=myParula07Red,rounded corners] (-0.3,-0.3) rectangle (1.8,0.3);
            \draw [draw=myParula07Red,rounded corners] (1.2,-0.4) rectangle (3.3,0.4);

        \end{scope}
        
    \end{tikzpicture}%
    \caption{%
    Blue circles represent \sltpps, red crosses represent \iindices with $\Gat{k}=1$ and $\Dat{k}=0$. 
    (a) Given intervals $A,B,C$ all containing the 2nd blue circle from left, interval $C$ is redundant.
    (b) Given $n$ blue circles, the adversary needs at least $n/4$ red crosses to draw a set of intervals satisfying \eqref{intervals-cover-ppivots,intervals-y-condition}. Here is a placement of red crosses relative to blue circles that achieves the minimum number of red crosses.%
    }
    \label{fig:one-cp-proof-figures}
\end{figure}

\begin{proof}[Proof of \appendixRef{\cref{lem:one-cp-induction-base}}]
Due to \eqref{cp-induction-base-ppivots-condition}, there is at least one \sltpp in $\intvl{0}{\Kcp}$ (otherwise $\Pin{0}{\Kcp}=0$).
Suppose for contradiction that there is no \sltcp in $\intvl{0}{\Kcp}$.
Since \sltcps are also \sltpps, it is enough to consider that
none of the \sltpps is a \sltcp.
Then around each \sltpp, there must be at least one interval which violates the combinatorial pivot condition.
Formally, there is a set of intervals $\intvlset$ such that:
\begin{IEEEeqnarray}{C}
    \label{eq:intervals-cover-ppivots}
    \bigcup_{I \in \intvlset} I \supseteq \left\{ k \in \intvl{0}{\Kcp} \colon \predPP{k} \right\} \IEEEeqnarraynumspace\\
    \label{eq:intervals-y-condition}
    \forall I \in \intvlset \colon\quad \Yat{I} \leq 0 \quad \text{(by \cref{prop:not-cp-exists-interval})}. \IEEEeqnarraynumspace
\end{IEEEeqnarray}
Without loss of generality, each interval $I \in \intvlset$ contains at least one \sltpp (removing all intervals that do not contain a \sltpp maintains \eqref{intervals-cover-ppivots,intervals-y-condition}).
Then if $\intvl{i}{j} \in \intvlset$, $i < \Kcp$.

First, consider the large intervals with $|I| \geq \Kcp$.
Consider \iindices $k \in I$ for which $\Gat{k}=1$ (\sltgood) but $\Dat{k}=0$ (block not processed).
From \cref{prop:download-or-spend-budget}, for each such \iindex, all honest nodes process $\goodsepbw$ blocks that are produced no later than $t_k$.
The number of \iindices $k \in I$ with  $\Gat{k} = 1$ and $\Dat{k} = 0$ is $\Gat{I} - \Dat{I}$.
For each such index, there must exist $\goodsepbw$ distinct blocks produced in or before the interval $I$. Therefore, if $I = \intvl{i}{j}$,
\begin{IEEEeqnarray}{rClr}
    \Qin{0}{j} &\geq& \goodsepbw \left( \Gin{i}{j} - \Din{i}{j} \right) & \\
    &\geq& \frac{\goodsepbw}{2} \left( \Gin{i}{j} - \Bin{i}{j} \right) & \quad \text{(by \cref{prop:not-cp-interval-properties}).}
\end{IEEEeqnarray}
This contradicts \eqref{cp-induction-base-margin-condition}.
Therefore, all intervals $I \in \intvlset$ are small ($|I| < \Kcp$).
Then for each $I \in \intvlset$, $I \subset \intvl{0}{2\Kcp}$.
Also, 
\begin{IEEEeqnarray}{Cl}
    \Gat{I} - \Dat{I} \geq \frac{1}{2} \left( \Gat{I} - \Bat{I} \right) \geq \frac{1}{2} \Pat{I}
    &\quad\text{(by \cref{prop:not-cp-interval-properties,prop:ppivots-imply-honest-margin})}. \IEEEeqnarraynumspace
    \label{eq:failed-more-than-ppivots}
\end{IEEEeqnarray}

Consider the \iindices $k \in \intvl{0}{2\Kcp}$ with $\Gat{k} = 1$ and $\Dat{k} = 0$.
Let $\intvlset_k = \{ I \in \intvlset \colon k \in I\}$ be the set of intervals that contain \index $k$.
Let $I^L_k$ be an interval in $\intvlset_k$ that stretches farthest to the left, and let $I^R_k$ be an interval that stretches farthest to the right (these may also be the same).
Note that all other intervals in $\intvlset_k$ are contained in $I^L_k \cup I^R_k$.
Therefore, all intervals in $\intvlset_k$ except $I^L_k$ and $I^R_k$ can be removed from $\intvlset$ while maintaining \eqref{intervals-cover-ppivots,intervals-y-condition} (see \cref{fig:one-cp-proof-figures}(a)).
This process is repeated for all $k \in \intvl{0}{2\Kcp}$ with $\Gat{k} = 1$ and $\Dat{k} = 0$, so that in the resulting set $\intvlset$, each such \iindex $k$ is contained in at most two intervals.
Then,
\begin{IEEEeqnarray}{rCl}
    \sum_{k \in \intvl{0}{2\Kcp} \colon \Gat{k} = 1, \Dat{k} = 0} |\intvlset_k|  &\leq& \sum_{k \in \intvl{0}{2\Kcp} \colon \Gat{k} = 1, \Dat{k} = 0} 2 \IEEEeqnarraynumspace\\
    &=& 2\left( \Gin{0}{2\Kcp} - \Din{0}{2\Kcp} \right). \IEEEeqnarraynumspace \label{eq:cp-sum-rewrite-1}
\end{IEEEeqnarray}
This sum can be rewritten as
\begin{IEEEeqnarray}{rCl}
    \IEEEeqnarraymulticol{3}{l}{\sum_{k \in \intvl{0}{2\Kcp} \colon \Gat{k} = 1, \Dat{k} = 0} |\intvlset_k| = \sum_{I \in \intvlset} \left( \Gat{I} - \Dat{I} \right)} \IEEEeqnarraynumspace\\
    \qquad\qquad&\geq& \sum_{I \in \intvlset} \frac{1}{2} \Pat{I} 
    \geq \frac{1}{2} \Pin{0}{\Kcp} \quad \text{(by \eqref{intervals-cover-ppivots})}.
    \IEEEeqnarraynumspace \label{eq:cp-sum-rewrite-2}
\end{IEEEeqnarray}
From \eqref{cp-sum-rewrite-1,cp-sum-rewrite-2},
\begin{IEEEeqnarray}{rCl}
    \Gin{0}{2\Kcp} - \Din{0}{2\Kcp} &\geq& \frac{1}{4} \Pin{0}{\Kcp}.
\end{IEEEeqnarray}
This can also be seen from \cref{fig:one-cp-proof-figures}(b).
Finally, as shown before, for each \index $k$ with $\Gat{k}=1$ and $\Dat{k}=0$, all honest nodes process at least $\goodsepbw$ distinct blocks produced in or before \iindex $k$ (\cref{prop:download-or-spend-budget}). This gives
\begin{IEEEeqnarray}{rCl}
    \Qin{0}{2\Kcp} &\geq& \goodsepbw \left( \Gin{0}{2\Kcp} - \Din{0}{2\Kcp} \right) 
    \geq \frac{\goodsepbw}{4} \Pin{0}{\Kcp}
    \IEEEeqnarraynumspace
\end{IEEEeqnarray}
which is a contradiction to \eqref{cp-induction-base-ppivots-condition}.
\end{proof}

\Cref{lem:one-cp-induction-full} proves that at least one \sltcp exists in successive intervals of $\Kcp$ length. \Cref{lem:one-cp-induction-full} is proved by induction, where the base case is \appendixRef{\cref{lem:one-cp-induction-base}}. 

\begin{lemma}
    \label{lem:one-cp-induction-full}
    \label{lem:many-pps-one-cps}
    If honest nodes use the scheduling policy $\dlrulelong$
    and
    \begin{IEEEeqnarray}{C}
        \label{eq:cp-induction-full-margin-condition}
        \forall \intvl{i}{j} \intvlgeq \Kcp
        \colon\quad 
        \frac{\goodsepbw}{2} \left( \Gin{i}{j} - \Bin{i}{j} \right) > \Qin{i-2\Kcp}{j}, \IEEEeqnarraynumspace\\
        \label{eq:cp-induction-full-ppivots-condition}
        \forall m \geq 0 \colon\quad 
        \frac{\goodsepbw}{4} \Pin{m\Kcp}{(m+1)\Kcp} > \Qin{(m-2)\Kcp}{(m+2)\Kcp}, \IEEEeqnarraynumspace
    \end{IEEEeqnarray}
    then 
    $\forall m \geq 0 \colon
    \exists k_m^* \in \intvl{m\Kcp}{(m+1)\Kcp} \colon
    \predCP{k_m^*}$.
\end{lemma}
\begin{proof}
    This will be proved through induction.
    For the base case ($m=0$), \appendixRef{\cref{lem:one-cp-induction-base}} shows 
    that $\exists k_1^* \in \intvl{0}{\Kcp} \colon \predCP{k_1^*}$.
    
    For $m \geq 1$, assume that $\exists k_{m-1}^* \in \intvl{(m-1)\Kcp}{m\Kcp}$ such that  $\predCP{k_{m-1}^*}$.
    Now we want to show that $\exists k_{m}^* \in \intvl{m\Kcp}{(m+1)\Kcp}$ such that  $\predCP{k_{m}^*}$.
    Suppose for contradiction that there is no \sltcp in $\intvl{m\Kcp}{(m+1)\Kcp}$.
    As in the proof of \appendixRef{\cref{lem:one-cp-induction-base}}, 
    there is a set of intervals $\intvlset$ such that:
    \begin{IEEEeqnarray}{Cr}
        \label{eq:induction-full-intervals-cover-ppivots}
        \bigcup_{I \in \intvlset} I \supseteq \left\{ k \in \intvl{m\Kcp}{(m+1)\Kcp} \colon \predPP{k} \right\} & \\
        \label{eq:induction-full-intervals-y-condition}
        \forall I \in \intvlset \colon\quad \Yat{I} \leq 0. &
    \end{IEEEeqnarray}
    Without loss of generality, each interval $I \in \intvlset$ contains at least one \sltpp.
    Then if $\intvl{i}{j} \in \intvlset$, $i < (m+1)\Kcp$ and $j > m\Kcp$.
    
    First, consider the large intervals with $|I| \geq \Kcp$.
    Consider \iindices $k \in I$ for which $\Gat{k}=1$ (\sltgood) but $\Dat{k}=0$ (block not processed).
    From \cref{prop:download-or-spend-budget},
    for each such \iindex $k$, all honest nodes process $\goodsepbw$ blocks that are produced
    in the interval $\intvl{k_{m-1}^*}{k}$.
    The number of \iindices $k \in I$ with  $\Gat{k} = 1$ and $\Dat{k} = 0$ is exactly $\Gat{I} - \Dat{I}$.
    For each such index, there must exist $\goodsepbw$ distinct blocks from distinct \BPOs
    that are processed by honest nodes.
    Therefore if $I = \intvl{i}{j}$,
    \begin{IEEEeqnarray}{rClr}
        \Qin{k_{m-1}^*}{j} &\geq& \goodsepbw \left( \Gin{i}{j} - \Din{i}{j} \right) & \IEEEeqnarraynumspace \\
        &\geq& \frac{\goodsepbw}{2} \left( \Gin{i}{j} - \Bin{i}{j} \right) & \quad \text{(from \cref{prop:not-cp-interval-properties}).} \IEEEeqnarraynumspace
    \end{IEEEeqnarray}
    But $k_{m-1}^* > (m-1)\Kcp$ and $i < (m+1)\Kcp$. Therefore $\Qin{k_{m-1}^*}{j} \leq \Qin{i-2\Kcp}{j}$.
    Then we have a contradiction to \eqref{cp-induction-full-margin-condition}.
    Therefore all intervals $I \in \intvlset$ are small ($|I| < \Kcp$).
    Then for each $I \in \intvlset$, $I \subset \intvl{(m-1)\Kcp}{(m+1)\Kcp}$.
    Also, 
    \begin{IEEEeqnarray}{Cr}
        \label{eq:failed-more-than-ppivots-induction}
        \Gat{I} - \Dat{I} \geq \frac{1}{2} \left( \Gat{I} - \Bat{I} \right) \geq \frac{1}{2} \Pat{I} & \quad \text{(\cref{prop:not-cp-interval-properties,prop:ppivots-imply-honest-margin})} \IEEEeqnarraynumspace
    \end{IEEEeqnarray}
    
    Consider the \iindices $k \in \intvl{(m-1)\Kcp}{(m+1)\Kcp}$ with $\Gat{k} = 1$ and $\Dat{k} = 0$.
    Following the arguments in the proof of \appendixRef{\cref{lem:one-cp-induction-base}},
    we can reduce the set $\intvlset$
    so that in the resulting set $\intvlset$, each such \iindex $k$ is contained in at most two intervals.
    Then,
    \begin{IEEEeqnarray}{rCl}
        \IEEEeqnarraymulticol{3}{l}{ \sum_{k \in \intvl{(m-1)\Kcp}{(m+1)\Kcp} \colon \Gat{k} = 1, \Dat{k} = 0} |\intvlset_k| } \nonumber \\
        \quad&\leq& 
        2\left( \Gin{(m-1)\Kcp}{(m+1)\Kcp} - \Din{(m-1)\Kcp}{(m+1)\Kcp} \right). \IEEEeqnarraynumspace
    \end{IEEEeqnarray}
    This sum can be rewritten as
    \begin{IEEEeqnarray}{rCl}
        \IEEEeqnarraymulticol{3}{l}{ \sum_{k \in \intvl{(m-1)\Kcp}{(m+1)\Kcp} \colon \Gat{k} = 1, \Dat{k} = 0} |\intvlset_k| } \\
        \quad&=& \sum_{I \in \intvlset} \left( \Gat{I} - \Dat{I} \right) \\
        &\geq& \sum_{I \in \intvlset} \frac{1}{2} \Pat{I} \\
        &\geq& \frac{1}{2} \Pin{m\Kcp}{(m+1)\Kcp}.
        \IEEEeqnarraynumspace
    \end{IEEEeqnarray}
    Therefore,
    \begin{IEEEeqnarray}{rCl}
        \IEEEeqnarraymulticol{3}{l}{ \Gin{(m-1)\Kcp}{(m+1)\Kcp} - \Din{(m-1)\Kcp}{(m+1)\Kcp} } \nonumber \\
        \quad&\geq& \frac{1}{4} \Pin{m\Kcp}{(m+1)\Kcp}.
    \end{IEEEeqnarray}
    
    Finally, for each \index $k$ with $\Gat{k}=1$ and $\Dat{k}=0$, all honest nodes process at least $\goodsepbw$ distinct blocks produced in or before 
    the most recent \sltcp before $(m-1)\Kcp$.
    By induction assumption, we have a \sltcp $k_{m-2}^* \in \intvl{(m-2)\Kcp}{(m-1)\Kcp}$.
    This gives
    \begin{IEEEeqnarray}{rCl}
        \IEEEeqnarraymulticol{3}{l}{ \Qin{(m-2)\Kcp}{(m+1)\Kcp} } \nonumber \\
        \quad&\geq& \goodsepbw \left( \Gin{(m-1)\Kcp}{(m+1)\Kcp} - \Din{(m-1)\Kcp}{(m+1)\Kcp} \right) \IEEEeqnarraynumspace \\
        &\geq& \frac{\goodsepbw}{4} \Pin{m\Kcp}{(m+1)\Kcp}
    \end{IEEEeqnarray}
    which is a contradiction.
\end{proof}

Finally,
using the fact that, with overwhelming probability, a constant fraction of \iindices are \sltpps,
we calculate the condition on the parameters $\blkrateslot, \slotduration$ in terms of $\goodsep,\goodsepbw$ for which the conditions \eqref{cp-induction-full-margin-condition,cp-induction-full-ppivots-condition} in \appendixRef{\cref{lem:many-pps-one-cps}} hold with overwhelming probability.
Precisely, we show that, with overwhelming probability, for any \iindex $k$ throughout the time horizon, there is at least one \sltcp in the interval $\intvl{k}{k+2\Kcp}$.

\begin{lemma}
\label{lem:pow-security-condition}
If $\frac{\goodsepbw}{16} \frac{(2\probGood-1)^2}{\probGood} > 1$,
then for $\Kcp = \Theta(\kappa^2)$,
$\Khorizon = \poly(\kappa)$,
with overwhelming probability,
for all $k < \Khorizon - 2\Kcp$,
$\exists k^* \in \intvl{k}{k+2\Kcp} \colon \predCP{k^*}$.
\end{lemma}
\begin{proof}
    \sloppy Define the event $\Event_1 = \left\{\forall \intvl{i}{j} \intvlgeq \Kcp \colon \Pin{i}{j} > (1-\delta)\probPP(j-i) \right\}$.
    Suppose that $\Event_1$ occurs, and $\frac{\goodsepbw}{16}\probPP (1-\delta) > 1$ for some $\delta \in (0,1)$.
    Then,
    \begin{IEEEeqnarray}{rCl}
        \forall \intvl{i}{j} \intvlgeq \Kcp \colon\quad
        \frac{\goodsepbw}{4} \Pin{i}{j} &>& \frac{\goodsepbw}{4} (1-\delta)\probPP(j-i) \IEEEeqnarraynumspace \\
        &>& 4(j-i) \\
        &\eqA& \Qin{i-2\Kcp}{j+\Kcp}
    \end{IEEEeqnarray}
    where (a) is because as $\slotduration \to 0$, each non-\sltempty \timeslot has exactly one \BPO. This satisfies \eqref{cp-induction-full-ppivots-condition} in \cref{lem:one-cp-induction-full}. Further,
    \begin{IEEEeqnarray}{C}
        \frac{\goodsepbw}{2} \left( \Gin{i}{j} - \Bin{i}{j} \right)
        \geq \frac{\goodsepbw}{2} \Pin{i}{j} > 3 (j-i) > \Qin{i-2\Kcp}{j}\IEEEeqnarraynumspace
    \end{IEEEeqnarray}
    which satisfies condition \eqref{cp-induction-full-margin-condition} in \cref{lem:one-cp-induction-full}.
    Therefore there is at least one \sltcp in every interval of the form $\intvl{m\Kcp}{(m+1)\Kcp}$.
    It also follows that for all $k$, there is at least one \sltcp in the interval $\intvl{k}{k+2\Kcp}$.
    By choosing $\Kcp = \Omega(\kappa^2)$, $\Khorizon = \poly(\kappa)$, and using \appendixRef{\cref{lem:many-pps}} and a union bound,
    the probability of failure of $\Event_1$ is $\negl(\kappa)$.
\end{proof}

While the analysis above is for the scheduling policy $\dlrulelong$, the proofs only use properties \ref{item:good-download-rule-no-repeat}, \ref{item:good-download-rule-honest-block}, \ref{item:good-download-rule-cutoff}
and thus apply to several other simple scheduling policies.
Another such scheduling policy is
``process only blocks that are consistent with the node's confirmed chain''.
In this work, we did not adopt this rule
because it would fail to recover from a network split, as demonstrated in the \greedyattack mentioned in \appendixRef{\cref{sec:greedy-attack}}.

\subsection{Security of Proof-of-Work Nakamoto Consensus}
\label{sec:fullproof-pow}

In \cref{lem:pow-security-condition}, we showed that under the
\rulelc scheduling policy, \sltcps occur in every $\Kcp$-interval.
This allows us, together with \appendixRef{\cref{lem:cps-stabilize}} (\sltcps stabilize), to prove safety and liveness of the protocol for a suitable confirmation depth $\confDepth$.
Subsequently, we take $\slotduration \to 0$ and $\blkratetime \triangleq \blkrateslot/\slotduration$ in order to model PoW accurately. We then identify the values of $\blkratetime$ for which given an adversary fraction $\beta$, the conditions required in \cref{lem:one-cp-induction-full} for \sltcps to occur hold with overwhelming probability. Finally, since $\goodsepbw$ was an analysis parameter chosen arbitrarily, we maximize over this parameter to find the best possible security--performance tradeoff (\cref{thm:safety-and-liveness-pow}).
The result is plotted for $\DeltaHeader \approx 0$ (reasonable approximation for large block sizes) in \cref{fig:comparison-bddelay-bdbandwidth}.

\RestateThmSafetyAndLivenessPow*

For PoW, we take $\slotduration \to 0$, and we would like to express parameters such as mining rate, confirmation latency, and execution time horizon in terms of real-time rather than the fictitious \timeslots or \iindices.
We use \cref{prop:index-time-bridge-pow} to bridge from \iindices 
to units of real-time,
which uses a Poisson tail bound to show that the inter-arrival time between \BPOs cannot be too large or too small.

\begin{proposition}
    \label{prop:index-time-bridge-pow}
    \begin{IEEEeqnarray}{rl}
        \forall k, K \in \IN \colon &
        \Prob{ \slotduration(t_{k+K} - t_{k}) \geq \frac{K}{\blkratetime(1-\delta)} } \leq e^{ \frac{-K\delta^2}{2(1+\delta)}},
        \IEEEeqnarraynumspace\\
        &
        \Prob{ \slotduration(t_{k+K} - t_{k}) \leq \frac{K}{\blkratetime(1+\delta)} } \leq e^{ \frac{-K\delta^2}{2(1+\delta)}}.
        \IEEEeqnarraynumspace
    \end{IEEEeqnarray}
\end{proposition}
\begin{proof}
    This results from a Poisson tail bound~\cite{poisson_tail} for the number of \BPOs in real time $K/\blkratetime$, and noting that 
    non-\sltempty \timeslots have exactly one \BPO
    for $\slotduration \to 0$.
\end{proof}

To prove \cref{thm:safety-and-liveness-pow}, 
we 
recall that
there is at least one \sltcp in the interval $\intvl{k}{k+2\Kcp}$ (\cref{lem:pow-security-condition}).
Given this, we prove safety and liveness of PoW NC in \cref{lem:safety-and-liveness-comb-pow}.
Finally, in
\cref{thm:safety-and-liveness-pow}, we calculate for given $\beta, \bwtime, \DeltaHeader$, the protocol parameters $\blkrateslot, \slotduration$ for which PoW NC is secure.
In doing so, since $\goodsepbw$ is just an analysis parameter, we optimize over $\goodsepbw$ to find the maximum $\blkratetime$.

\begin{lemma}
    \label{lem:safety-and-liveness-comb-pow}
    If for some $\Kcp > 0$,
    \begin{IEEEeqnarray}{C}
        \label{eq:condition-one-cp-m-pow-safety}
        \forall k \colon \exists k^* \in \intvl{k}{k+2\Kcp} \colon \predCP{k^*},
        \IEEEeqnarraynumspace
    \end{IEEEeqnarray}
    then the PoW Nakamoto consensus protocol $\protocol$ with $\confDepth = 2\Kcp + 1$ satisfies safety.
    Further, if the environment is $(\tput, \Ttput)$-tx-limited with $\tput = (1+\delta)\left(\frac{1}{2} - \frac{1-e^{-\beta\blkrateslot}}{1-e^{-\blkrateslot}}(1+\delta)\right)\blkratetime\slotduration$ and $\Ttput = \frac{2\Kcp}{\blkratetime\slotduration(1+\delta)}$, and
    \begin{IEEEeqnarray}{C}
        \label{eq:cond-index-time-bridge-pow}
        \forall k \in \IN, K \geq \Kcp\colon
        \quad
        \frac{K}{\blkratetime\slotduration(1+\delta)} < t_{k+K} - t_{k} < \frac{K}{\blkratetime\slotduration(1-\delta)},
        \IEEEeqnarraynumspace
    \end{IEEEeqnarray}
    then it also satisfies liveness with $\Tlive = \max\left\{\Ttput,\frac{2\Kcp}{\blkratetime\slotduration(1-\delta)}\right\} + \frac{4\Kcp + 2}{\blkratetime\slotduration(1-\delta)}$.
\end{lemma}
\begin{proof}
    \emph{Safety:}
    For an arbitrary \timeslot $t$, let $k$ be the largest \iindex such that $t_k \leq t$.
    From \eqref{condition-one-cp-m-pow-safety}, every interval of $2\Kcp$ \iindices contains at least one \sltcp. Therefore, there exists $k^* \in \intvl{k-2\Kcp-1}{k-1}$ such that $\predCP{k^*}$.
    Let $b^*$ be the block from \iindex $k^*$.
    Due to \appendixRef{\cref{lem:cps-stabilize}}, for all honest nodes $p,q$ and $t' \geq t$,
    $b^* \in \dC_p(t)$ and $b^* \in \dC_q(t')$.
    But $k^* \geq k-\confDepth$, so the block $b^*$ cannot be $\confDepth$-deep in any chain at \timeslot $t$ Therefore, $\LOG{p}{t}$ is a prefix of $b^*$ which in turn is a prefix of $\dC_q(t')$.
    We can thus conclude that
    either 
    $\LOG{p}{t} \preceq \LOG{q}{t'}$ or $\LOG{q}{t'} \preceq \LOG{p}{t}$.
    Therefore, safety holds.

    For an arbitrary \timeslot $t$, let $k$ be the largest \iindex such that $t_k \leq t$.
    We will first prove that all transactions received in \timeslots $t - \Ttput$ to $t$, which are of total size at most $\tput\Ttput$ as per the tx-limited environment, will be added to the longest processed chains of all nodes by the \timeslot corresponding to \iindex $k + 2\Kcp$.
    Let $\Ktput = \max\{2\Kcp, \blkratetime \slotduration (1+\delta) \Ttput\}$.
    We know that there exists $k^* \in (k,k+2\Ktput]$ such that $\predCP{k^*}$.
    Since $k^*$ is a \sltcp, for all $\intvl{i}{j} \ni k^*$, $\Din{i}{j} > \Nin{i}{j}$ (\appendixRef{\cref{def:cp}} and \eqref{random_walks_X_and_Y}), and hence $\Din{i}{j} > \frac{j-i}{2}$.
    Particularly,
    \begin{IEEEeqnarray}{rCl}
        \implies \Din{k}{k + \Ktput - 1} &>& \frac{\Ktput-1}{2}.
        \IEEEeqnarraynumspace
    \end{IEEEeqnarray}
    Then from \cref{prop:chain-growth-interval},
    \begin{IEEEeqnarray}{rCl}
        L_{\min}(t_{k+\Ktput-1}+\goodsep) - L_{\min}(t_{k+1}-1) &\geq& \Din{k}{k + \Ktput - 1} \nonumber \\
        &\geq& \frac{\Ktput}{2}.
        \IEEEeqnarraynumspace
    \end{IEEEeqnarray}
    This means that the $\Ktput/2$ last blocks in any node's longest processed chain at \timeslot $t_{k+\Ktput-1+\goodsep}$ are from the \iindices $\intvl{k}{k+\Ktput-1}$.
    Among these, the number of blocks produced by the adversary can be, by a concentration bound, at most $\frac{1-e^{-\beta\blkrateslot}}{1-e^{-\blkrateslot}}(1+\delta)\Ktput$.
    Therefore, at least $\left(\frac{1}{2} - \frac{1-e^{-\beta\blkrateslot}}{1-e^{-\blkrateslot}}(1+\delta)\right)\Ktput$ blocks are produced by honest nodes.
    The cumulative size of pending transactions is at most $\tput \Ttput$, which fits in these honest blocks.
    Finally, we use \cref{prop:chain-growth-interval} again to show:
    \begin{IEEEeqnarray}{rCl}
        L_{\min}(t_{k+\Ktput+2\confDepth-1}+\goodsep) - L_{\min}(t_{k+2\confDepth+1}-1) &\geq& \confDepth.
        \IEEEeqnarraynumspace
    \end{IEEEeqnarray}
    Thus, the newly added transactions are $\confDepth$-deep, hence confirmed, by all nodes by \iindex $k + 2\Kcp + 2\confDepth$.
    Finally, with \eqref{cond-index-time-bridge-pow},
    \begin{IEEEeqnarray}{rCl}
        t_{k+\Ktput + 2\confDepth-1}+\goodsep - t &\leq& t_{k + \Ktput + 4\Kcp + 1} + \goodsep - t_{k} \nonumber \\
        &\leq& t_{k + 6\Kcp + 2} - t_{k} \nonumber \\
        &<& \frac{\Ktput + 4\Kcp + 2}{\blkratetime\slotduration(1-\delta)}.
        \IEEEeqnarraynumspace
    \end{IEEEeqnarray}
    Therefore, $\mathsf{tx}\in\LOG{p}{t'}$ for all $t'\geq t+\Tlive$.
\end{proof}

\begin{proof}[Proof of \cref{thm:safety-and-liveness-pow}]
    From \cref{lem:pow-security-condition}, assuming
    \begin{IEEEeqnarray}{C}
        \label{eq:C-equation-pow}
        \frac{\goodsepbw}{16} \frac{(2\probGood-1)^2}{\probGood} (1-\delta) > 1,
    \end{IEEEeqnarray}
    and from a union bound over horizon $\Khorizon = \poly(\kappa)$
    on the result of \cref{prop:index-time-bridge-pow},
    the conditions required for \cref{lem:safety-and-liveness-comb-pow} are satisfied with overwhelming probability.
    Then \cref{lem:safety-and-liveness-comb-pow} guarantees safety and liveness with $\confDepth = 2\Kcp = \Theta(\kappa^2)$ and $\Tlive = \frac{6\Kcp+2}{\blkratetime\slotduration(1-\delta)} = \Theta(\kappa^2)$.

    \Iindices are mapped to real time as $\TliveReal \triangleq \Tlive \slotduration$.
    Further, the event $\{\slotduration t_{\Khorizon} > \frac{\Khorizon}{\blkratetime(1+\delta)} \}$ also occurs except with negligible probability (\cref{prop:index-time-bridge-pow}), and therefore the time horizon $\Khorizon$ \iindices corresponds to at least a time horizon of $\Thorizon \triangleq \frac{\Khorizon}{\blkratetime(1+\delta)}$ real-time units.

    Finally, we take the limit $\slotduration \to 0$.
    With the relations $\blkratetime = \blkrateslot/\slotduration$, $\BWEquation$, and 
    $\probPP = \probPPFormula$,
    \begin{IEEEeqnarray}{C}
        \label{eq:prob-good-equation-pow}
        \probGood = \probGoodFormula \to (1-\beta)e^{-\blkratetime\left(\DeltaHeader + \goodsepbw/\bwtime\right)}.
        \IEEEeqnarraynumspace
    \end{IEEEeqnarray}

    Moreover, the value of $\tput$ from \cref{lem:safety-and-liveness-comb-pow} converges to $(1+\delta)(\frac{1}{2}-\beta(1+\delta))\blkratetime$ in real-time units.
    
    Note that $\goodsepbw$ is an analysis parameter whose value is arbitrary.
    To find the maximum block production rate $\blkratetime$ that the protocol can achieve, we optimize over $\goodsepbw$.
    To find the maximum achievable $\blkratetime$,
    we can take $\delta \to 0$ as we can increase the latency through increasing $\Kcp$ to still satisfy the error bounds.
    Maximizing over $\goodsepbw$ from \eqref{C-equation-pow,prob-good-equation-pow} gives the resulting threshold.
\end{proof}

\section{Proof-of-Stake Model Details}
\label{sec:pos-model-details}

\begin{algorithm}[tb]%
    \caption{%
        Idealized functionality $\FtreePoS$:
        block production lottery and header chain structure
        for PoS
        (helper functions: \cref{sec:algos-reference-helperfunctions})%
    }%
    \label{alg:hdrtree-pos}%
    \begin{algorithmic}[1]%
        \scriptsize%

        \LineComment{$\Call{init}{\genesisHeaderChain, \mathsf{numNodes}}$ and $\Call{verify}{\Chain}$ same as in \appendixRef{\cref{alg:hdrtree-pow}}}
        
        \On{$\Call{isLeader}{P,t}$ \textbf{from} $\mathcal A$ (only for adversarial node $P$) or $\FtreePoS$}
                \label{loc:hdrtree-pos-leader}
            \LineComment{Abstraction of proof-of-stake lottery: each node is chosen leader in each slot with probability $\blkrateslot/N$ independently of other nodes and slots}
            \If{$\Lottery[P,t] = \bot$}
                    \label{loc:hdrtree-pos-rhodefn}
                \State $\Lottery[P,t] \overset{\$}{\gets} (\TRUE \text{ with probability $\blkrateslot/N$, else } \FALSE)$
                    \label{loc:hdrtree-pos-blockproductionlottery}
            \EndIf
            \State \Return{$\Lottery[P,t]$}
        \EndOn

        \On{$\Call{extend}{t', \Chain, \txs}$ \textbf{from} $\mathcal A$ (only for adversarial node $P$) or $\FtreePoS$}
                \label{loc:hdrtree-pos-binding}
            \LineComment{New header chain is valid if parent chain $\Chain$ is valid, $P$ is leader for slot $t'$, and $t'$ is later than the tip of $\Chain$ and is not in the future}
            \If{$(\Chain\in\Tree) \land \FtreePoS.\Call{isLeader}{P,t'} \land (\Chain.\mathsf{time} < t' \leq t)$} \label{loc:hdrtree-pos-check-lottery}
                \LineComment{Produce a new block header extending $\Chain$}
                \State $\Chain' \gets \Chain\|\operatorname{newBlock}(\mathsf{time}\colon t', \mathsf{node}\colon P, \mathsf{txsHash}\colon \operatorname{Hash}(\txs))$ 
                \State $\Tree \gets \Tree \cup \{\Chain'\}$
                    \Comment{Register new header chain in header tree}
                \State \Return{$\Chain'$}
            \EndIf
            \State \Return{$\bot$}
        \EndOn

        \On{$\Call{extend}{\Chain, \txs}$ \textbf{from} node $P$ (possibly adversarial) \textbf{at} \timeslot $t$}
            \State \Return{$\FtreePoS.\Call{extend}{t, \Chain, \txs}$}
        \EndOn

    \end{algorithmic}%
\end{algorithm}%

Details of the PoS-based block production and verification are abstracted through an idealized functionality $\FtreePoS$ whose pseudocode is provided in \appendixRef{\cref{alg:hdrtree-pow}} (\cf~\appendixRef{\cref{alg:hdrtree-pow}}, \cite[Fig.~2]{sleepy}, \cite[Alg.~3]{bwlimitedposlc}).

As in PoW, each node can make one block production attempt per \timeslot that will be successful with probability $\blkrateslot/N$, independently of other nodes and \timeslots
(\myalgref{alg:hdrtree-pos}{loc:hdrtree-pos-blockproductionlottery})%
\footnote{There may be multiple blocks in one \timeslot, as in
the Ouroboros~\cite{kiayias2017ouroboros,david2018ouroboros,badertscher2018ouroboros} and Sleepy Consensus~\cite{sleepy,snowwhite} protocols.}%
, modeling uniform stake.
In PoS, however, (even past) block production opportunities can be `reused' to produce multiple blocks with different parents and/or content, \ie, to equivocate
(\myalgref{alg:hdrtree-pos}{loc:hdrtree-pos-leader,loc:hdrtree-pos-binding}).
\fi

\end{document}